\documentclass[a4paper,aps, noarxiv,onecolumn, accepted=2019-06-19]{quantumarticle}

\usepackage{amsmath}
\usepackage{amssymb}
\usepackage{amsthm}
\usepackage{tikz}
\usetikzlibrary{decorations.markings,decorations.pathreplacing,calc}
\usetikzlibrary{arrows.meta}
\usepackage{algorithm}
\usepackage{algpseudocode}
\usepackage{thm-restate}
\usepackage{verbatim}
\usepackage[framemethod=tikz]{mdframed}
\usepackage{graphicx}
\usepackage{subcaption}
\usepackage{float}

\usepackage[margin=1in]{geometry}
\usepackage[utf8]{inputenc}

\usepackage{lipsum}

\usepackage[colorlinks = true]{hyperref}

\usepackage{xcolor}
\definecolor{darkred}  {rgb}{0.5,0,0}
\definecolor{darkblue} {rgb}{0,0,0.5}
\definecolor{darkgreen}{rgb}{0,0.5,0}

\hypersetup{
  urlcolor   = blue,         
  linkcolor  = darkblue,     
  citecolor  = darkgreen,    
  filecolor  = darkred       
}

\newcommand{\be}{\begin{equation}}
\newcommand{\ee}{\end{equation}}
\newcommand{\ba}{\begin{array}}
	\newcommand{\ea}{\end{array}}
\newcommand{\bea}{\begin{eqnarray}}
\newcommand{\eea}{\end{eqnarray}}

\newcommand{\calL}{{\cal L }}
\newcommand{\calR}{{\cal R }}

\newcommand{\calV}{{\cal V }}

\newcommand{\calS}{{\cal S }}

\newcommand{\calM}{{\cal M }}

\newcommand{\calZ}{{\cal Z }}

\newcommand{\ZZ}{\mathbb{Z}}

\newcommand{\RR}{\mathbb{R}}
\newcommand{\la}{\langle}
\newcommand{\ra}{\rangle}
\newcommand{\EX}{\mathbb{E}}

\newtheorem{dfn}{Definition}
\newtheorem{prop}{Proposition}

\newtheorem{claim}{Claim}
\newtheorem{lemma}{Lemma}
\newtheorem{corol}{Corollary}
\newtheorem{fact}{Fact}

\newtheorem{theorem}{Theorem}
\newtheorem*{theorem*}{Theorem}

\newcommand{\ket}[1]{\ensuremath{\vert#1\rangle}}
\newcommand{\bra}[1]{\ensuremath{\langle #1\vert}}
\newcommand{\bk}[2]{\ensuremath{\langle #1\vert #2\rangle}}
\newcommand{\kb}[2]{\ensuremath{\vert #1 \rangle \langle #2 \vert}}

\renewcommand{\vec}[1]{{#1}}

\def\id{\mbox{\small 1} \!\! \mbox{1}}

\def\id{{\mathchoice {\rm 1\mskip-4mu l} {\rm 1\mskip-4mu l} {\rm 1\mskip-4.5mu l} {\rm 1\mskip-5mu l}}}

\begin{document}

\title{Simulation of quantum circuits by low-rank stabilizer decompositions}

\author{Sergey Bravyi}
\affiliation{IBM T.J. Watson Research Center, Yorktown Heights NY 10598}
\author{Dan Browne}
\affiliation{Department of Physics and Astronomy, University College London, London, UK}
\author{Padraic Calpin}
\affiliation{Department of Physics and Astronomy, University College London, London, UK}
\author{Earl Campbell}
\affiliation{Department of Physics and Astronomy, University of Sheffield, Sheffield, UK}
\author{David Gosset}
\affiliation{IBM T.J. Watson Research Center, Yorktown Heights NY 10598}
\affiliation{Department of Combinatorics \& Optimization and Institute for Quantum Computing, University of Waterloo, Waterloo, Canada}
\author{Mark Howard} 
\affiliation{Department of Physics and Astronomy, University of Sheffield, Sheffield, UK}

\begin{abstract}

Recent work has explored using the stabilizer formalism to classically simulate quantum circuits containing a few non-Clifford gates. The computational cost of such methods is directly related to the notion of {\em stabilizer rank}, which for a pure state $\psi$ is defined to be the smallest integer $\chi$ such that  $\psi$ is a superposition of $\chi$ stabilizer states.  Here we develop a comprehensive mathematical theory of the stabilizer rank and the related approximate stabilizer rank.  We also present a suite of classical simulation algorithms with broader applicability and significantly improved performance over the previous state-of-the-art.  A new feature is the capability to simulate circuits composed of Clifford gates and arbitrary diagonal gates, extending the reach of a previous algorithm specialized to the Clifford+T gate set. We implemented the new simulation methods and used them to simulate quantum algorithms with 40-50 qubits and over 60 non-Clifford gates, without resorting to high-performance computers. We report a simulation of the Quantum Approximate Optimization Algorithm in which we process superpositions of $\chi\sim10^6$ stabilizer states and sample from the full $n$-bit output distribution, improving on previous simulations which used $\sim 10^3$ stabilizer states and sampled only from single-qubit marginals. We also simulated instances of the Hidden Shift algorithm with circuits including up to 64 $T$ gates or 16 CCZ gates; these simulations showcase the performance gains available by optimizing the decomposition of a circuit's non-Clifford components.
\end{abstract}
	
\maketitle
\tableofcontents

\section{Introduction}
It is widely believed that universal quantum computers cannot be efficiently simulated by classical probabilistic algorithms. This belief is partly supported by the fact that
state-of-the-art classical simulators employing modern supercomputers are still limited to a few dozens of qubits~\cite{smelyanskiy2016qhipster,haner20170,pednault2017breaking,chen2018classical}.
At the same time, certain quantum information processing tasks do not require computational universality. For example, quantum error correction
based on stabilizer codes and Pauli noise models~\cite{gottesman1998theory} only requires quantum circuits composed of 
Clifford gates and Pauli measurements--which can be easily simulated  classically 
for thousands of qubits using the Gottesman-Knill theorem~\cite{aaronson04improved,anders2006fast}. Furthermore, it is known that Clifford circuits
can be promoted to universal quantum computation when provided with a plentiful supply of some computational primitive outside the stabilizer operations, such as a non-Clifford gate or magic state~\cite{bravyi2005universal}.  This raises the possibility of simulating quantum circuits with a large number of qubits and few non-Clifford gates.  Aaronson and Gottesman~\cite{aaronson04improved} were the first to propose a classical simulation method covering this situation, with a runtime that scales polynomially with the number of qubits and Clifford gate count but exponentially with the number of non-Clifford gates.  This early work is an intriguing proof of principle but with a very large exponent, limiting potential applications.

Recent algorithmic improvements have helped tame this exponential scaling by significantly decreasing the size of the exponent.  A first step was made by Garcia, Markov and Cross~\cite{garcia2012efficient,Garcia14moreStabRank}, who proposed and studied the decomposition of states into a superposition of stabilizer states. Bravyi, Smith and Smolin~\cite{Bravyi16stabRank} formalized this into the notion of stabilizer rank. The stabilizer rank $\chi(\psi)$  of 
 a pure state $\psi$ is defined as the smallest integer  $\chi$ such that $\psi$
 can be expressed  as a superposition of $\chi$ stabilizer states. 
It can be thought of as a measure of computational non-classicality analogous the Schmidt rank measure of entanglement.  In particular,  $\chi(\psi)$ quantifies the simulation cost of 
stabilizer operations (Clifford gates and Pauli measurements) applied to the initial state  $\psi$. 

It is known that stabilizer operations augmented with  preparation of certain
single-qubit ``magic states" become computationally universal~\cite{bravyi2005universal}.
In particular, any quantum circuit composed of Clifford gates
and $m$ gates $T=|0\ra\la 0|+e^{i\pi/4}|1\ra\la 1|$ can be implemented by
stabilizer operations acting on the initial state $|\psi\ra =|T\ra^{\otimes m}$,
where $|T\ra \propto |0\ra + e^{i\pi/4} |1\ra$.
Thus the stabilizer rank  $\chi(T^{\otimes m})$ provides an upper bound on the simulation cost
of Clifford+$T$ circuits with $m$ $T$-gates.
The authors of  Ref.~\cite{Bravyi16stabRank} used a numerical search method to 
compute the stabilizer rank $\chi(T^{\otimes m})$ for $m\le 6$ finding that
$\chi(T^{\otimes 6})=7$. The numerical search  becomes impractical for
$m>6$ 
and one instead works with suboptimal decompositions by breaking $m$ magic states up into blocks of six or fewer qubits. This yields a classical simulator of Clifford+$T$ circuits running in time $2^{0.48 m}$ with certain polynomial prefactors~\cite{bravyi2016improved}.
More recently, Ref.~\cite{bravyi2016improved} introduced
an approximate version of the stabilizer rank and a method of constructing  approximate
stabilizer decomposition of the magic states $|T\ra^{\otimes m}$. This led to a simulation algorithm with runtime scaling as $2^{0.23m}$
that  samples the output distribution of the target circuit with  a small statistical error.
In practice, it can simulate single-qubit measurements on the output state of Clifford+$T$ circuits with  $m\le 50$ on a standard laptop~\cite{bravyi2016improved}.  A similar class of simulation methods uses Monte Carlo sampling over quasiprobability distributions, where the distribution can be over either a discrete phase space ~\cite{veitch2012negative,pashayan15,Delfosse15rebits}, over the class of stabilizer states~\cite{Howard17robustness} or over stabilizer operations~\cite{OakRidge17}.  These quasiprobability methods are a natural method for simulating noisy circuits but for pure circuits they appear to be slower than simulation methods based on stabilizer rank.  

Here we present a more general set of tools for finding exact and approximate stabilizer decompositions 
as well as improved simulation algorithms based on such decompositions.
A central theme throughout this paper  is generalizing the results  of Refs.~\cite{Bravyi16stabRank,bravyi2016improved}
beyond the Clifford+$T$ setting.  While Clifford+$T$ is a universal gate set, it requires several hundred $T$ gates to synthesize an arbitrary single qubit gate to a high precision (e.g. below $10^{-10}$ error).   Therefore, it would be impractical to simulate such gates using the Clifford+$T$ framework.  We achieve
significant improvements in the simulation runtime by branching out to more general gate sets
including arbitrary-angle $Z$-rotations and CCZ gates. 
Furthermore, we propose more efficient subroutines for simulating the action of Clifford gates
and Pauli measurements on superpositions of $\chi\gg 1$ stabilizer states. 
In practice, this enables us to perform simulations in the regime $\chi \sim 10^6$
with about 50 qubits
on a laptop computer improving upon $\chi\sim 10^3$ simulations reported in
Ref.~\cite{bravyi2016improved}.
The table provided below summarizes new simulation methods,
simulation tasks addressed by each method, and the runtime scaling. 
\begin{figure}[h]
	\includegraphics[height=6.65cm]{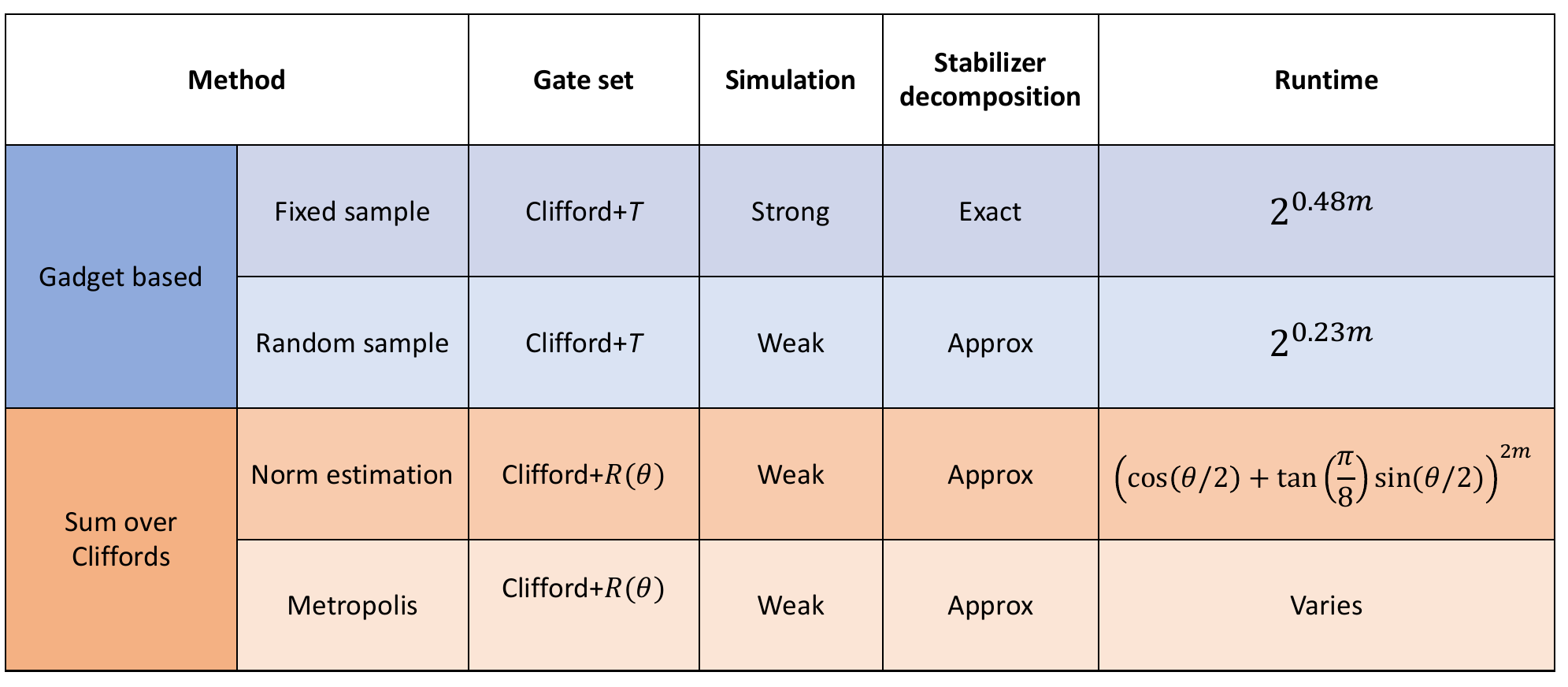}
	\caption{Summary of new simulation methods.
	For simplicity, here we restrict the
	attention to quantum circuits composed of Clifford gates and diagonal single-qubit gates
	$R(\theta)=\mathrm{diag}(1,e^{i\theta})$. The $T$-gate can be obtained as a special
	case $T=R(\pi/4)$. We consider strong and weak simulation tasks where the goal is to
	estimate a single output probability (with a small multiplicative error) and
	sample the output probability distribution (with a small  statistical error)
	respectively.
	The runtime scales exponentially with the non-Clifford gate count $m$
	and polynomially with the number of qubits and the Clifford gate count. For simplicity, here
	we ignore the polynomial prefactors. For a detailed description of our simulation methods, 
	see Section~\ref{sec:simulations}.}
	\label{summary_table}
\end{figure}

On the theory side, we establish some general  properties of the approximate stabilizer rank.
Our main tool is a Sparsification Lemma that shows how to 
convert a dense stabilizer decomposition of a given target state 
(that may contain all possible stabilizer states) to a sparse decomposition
that contains fewer stabilizer states. 
The lemma generalizes the method of random linear codes introduced  in Ref.~\cite{bravyi2016improved}
in the context of Clifford+$T$ circuits. 
It allows us to obtain sparse stabilizer decompositions for the output state of
more general quantum circuits directly without using magic state gadgets. 
Combining the Sparsification Lemma and  convex duality arguments, we relate
the approximate stabilizer rank of a state $\psi$ to a stabilizer fidelity $F(\psi)$ defined
as the maximum overlap between $\psi$ and stabilizer states. Central to these calculations is a new quantity called Stabilizer Extent, which quantifies, in an operationally relevant way, how non-stabilizer a state is. We give necessary and sufficient conditions under which the stabilizer fidelity is multiplicative under the tensor product. Finally, we propose a new strategy for proving lower bounds on the stabilizer rank
of the magic states which uses the machinery of ultra-metric matrices~\cite{MMM,NabenVarga}.

As a main application of our simulation algorithms we envision verification of noisy intermediate-size
quantum circuits~\cite{preskill2018quantum}
in the regime when a brute-force classical simulation may be
 impractical~\cite{aharonov2017interactive,morimae2016post,Jozsa2017}.
For example, a quantum circuit composed of Clifford gates and single-qubit $Z$-rotations
with angles  $\theta_1,\ldots,\theta_m$ can be efficiently simulated using our methods
in the regime when only a few of the angles $\theta_a$ are non-zero or if all the angles 
$\theta_a$ are small in magnitude, see Section~\ref{sec:sum_cliffords}.
By fixing the Clifford part of the circuit and varying the rotation angles $\theta_a$ one can
therefore interpolate between the regimes where the circuit output  can and cannot be verified 
classically. From the experimental perspective, single-qubit $Z$-rotations
are often the most reliable elementary operations~\cite{mckay2017efficient}.
Thus one should expect that the circuit output fidelity should not depend significantly on the
choice of the angles $\theta_a$.

The next section provides a more detailed overview of our results.

\section{Main results}
\label{detailed}
Recall that the Clifford group is a group of unitary $n$-qubit operators
generated by  single-qubit and two-qubit gates
from the set $\{H,S,CX\}$. Here $H$ is the Hadamard gate,
$S=|0\ra\la 0|+i|1\ra\la 1|$ is the phase shift gate, 
and CX=CNOT is the controlled-X gate.
Stabilizer states are $n$-qubit states of the form $|\phi\ra=U|0^n\ra$, where $U$ is a
Clifford operator. We also use $X_j,Y_j,Z_j$ to denote Pauli operators
acting on the $j$-th qubit. Below we also make use of the stabilizer formalism, and refer the unfamiliar reader to the existing literature~\cite{nielsen2002quantum}.

\subsection{Tools for constructing low-rank stabilizer decompositions}
\label{sec:rank_properties}
In this section we summarize our results pertaining to the stabilizer rank and describe methods of decomposing a state into a superposition of stabilizer states.  A reader interested only in the application for simulation of quantum circuits may wish to proceed to Sections~\ref{sec:subroutines}, \ref{sec:simulations}.

\begin{dfn}[\bf Exact stabilizer rank, $\mathbf{\chi}$~\cite{Bravyi16stabRank}]
Suppose $\psi$ is a pure $n$-qubit state. 
The  exact stabilizer rank $\chi(\psi)$ is the smallest integer $k$ such that 
$\psi$ can be written as 
	\begin{equation}
	\label{srank1}
	\vert  \psi \rangle  = \sum_{\alpha=1}^{k} c_\alpha \vert \phi_\alpha \rangle	,
	\end{equation}	
for some  $n$-qubit stabilizer states $\phi_\alpha$ and some complex coefficients $c_\alpha$. 
\end{dfn}
\noindent
By definition, $\chi(\psi)\ge 1$ for all $\psi$ and $\chi(\psi)=1$ iff $\psi$ is a stabilizer state.

\begin{dfn}[\bf Approximate stabilizer rank, $\mathbf{\chi_\delta}$~\cite{bravyi2016improved}]
Suppose $\psi$ is a pure $n$-qubit state such that $\|\psi\|=1$.
Let  $\delta>0$ be a precision parameter.
The approximate stabilizer rank $\chi_\delta(\psi)$ is the smallest integer $k$ such that 
$\| \psi -\psi'\|\le \delta$ for some state $\psi'$ with exact stabilizer rank $k$.
\end{dfn}
\noindent
Note that this definition differs slightly from the one from Ref.~\cite{bravyi2016improved} which is based on the fidelity.
 Our first result provides an upper bound on the approximate stabilizer rank.
\begin{theorem}[\bf Upper bound on $\mathbf{\chi_\delta}$]
\label{thm:randomCvec}
Let $\psi$ be a normalized $n$-qubit state with a stabilizer decomposition 
$\ket{\psi} = \sum_{\alpha=1}^k c_\alpha 	|\phi_\alpha\ra$
where $\ket{\phi_\alpha}$ are normalized stabilizer states and $c_\alpha\in \mathbb{C}$.  Then
\begin{equation}
\chi_\delta(\psi) \leq 1+ \| \vec{c} \|_1^2 / \delta^2 .
\end{equation}
Here $\| \vec{c}\|_1 \equiv  \sum_{\alpha=1}^k |c_\alpha|$. 
\end{theorem}	
We note that the stabilizer decomposition $\ket{\psi} = \sum_{\alpha=1}^k c_\alpha 
|\phi_\alpha\ra$
in the statement of the theorem 
does not have to be
optimal. For example,  it may include all  stabilizer states. The proof of the theorem is provided in Section~\ref{Sec_approx_stab_rank}. It is constructive in the sense that it provides a method of calculating
a state $\psi'$ which is a superposition of $\chi'\approx \delta^{-2}  \| \vec{c} \|_1^2$ stabilizer states 
such that $\|\psi'-\psi\|\le \delta$. Such a state $\psi'$ is obtained 
using a randomized sparsification method. It works by sampling
$\chi'$ stabilizer states $\phi_\alpha$ from the given stabilizer decomposition of $\psi$ at random
with probabilities proportional to $|c_\alpha|$. The state $\psi'$ is then defined as
a superposition of the sampled states $\phi_\alpha$ with equal weights, 
see the Sparsification Lemma and related discussion in Section~\ref{Sec_Lemma_Proof}. The theorem motivates the following definition. 
\begin{dfn}[\bf Stabilizer Extent, $\xi$]
	\label{Cstar_states}
Suppose $\psi$ is a normalized $n$-qubit state. 
Define the stabilizer extent $\xi(\psi)$ as the minimum of $\| \vec{c} \|^2_1$ over all
stabilizer decompositions  $|\psi\ra=\sum_{\alpha=1}^k c_\alpha |\phi_\alpha\ra$
where $\phi_\alpha$ are normalized stabilizer states. 
\end{dfn}
The theorem immediately implies that 
\begin{equation}
\label{cstar}
\chi_\delta(\psi) \leq 1+\xi(\psi) / \delta^2.
\end{equation}
While it is difficult to compute or prove tight bounds for the exact or approximate stabilizer rank, 
we find that $\xi(\psi)$ is a more amenable quantity that can 
be calculated   for many states  $\psi$ relevant in the context of quantum circuit
simulation. 
In particular, we prove 
\begin{prop}[\textbf{Multiplicativity of Stabilizer Extent}]
	\label{multi}
Let $\{ \psi_1,\psi_2,\ldots, \psi_L \}$ be any set of states 
such that each state $\psi_j$ describes a system of at most three qubits.
	Then
	\begin{equation}
	\xi( \psi_1 \otimes \psi_2  \otimes  \ldots \otimes \psi_L) = \prod_{j=1}^L \xi(\psi_j)  .
	\end{equation}
	\label{thm:prod}
\end{prop}
\noindent
This shows that the upper bound of Theorem~\ref{thm:randomCvec} is multiplicative under tensor product in the case of few-qubit states. It remains open whether $\xi$ is multiplicative on arbitrary collections of states.  

The proof of Proposition~\ref{thm:prod} is provided in Section \ref{Sec_Cstar_multi}.  It uses the fact that standard convex duality provides a characterization of $\xi$ in terms of the following quantity. 
\begin{dfn}[\bf Stabilizer Fidelity, $F$]
The stabilizer fidelity, $F(\psi)$, of a state $\psi$ is 
\begin{equation}
F( \psi ) = \mathrm{max}_{\phi} |\bk{\phi}{\psi}|^2,
\label{eq:stabfid}
\end{equation}
where the maximization is over all normalized stabilizer states $\phi$.
\end{dfn}
Proposition \ref{thm:prod} is obtained as a consequence of new results concerning multiplicativity of the stabilizer fidelity. In particular, we apply the classification of entanglement
in three-partite stabilizer states~\cite{ghz}
to derive conditions for the multiplicativity  of $F(\psi)$. 
More precisely, we define a set of quantum states $\mathcal{S}$ which
we call {\em stabilizer aligned} such that 
$F(\phi \otimes \psi) =F(\phi)F(\psi)$ whenever $\phi,\psi\in \mathcal{S}$. 
A state  $\psi$ is called stabilizer aligned
if  the overlap between $\psi$ and any stabilizer projector of rank $2^k$
is at most $2^{k/2} F(\psi)$. 
Remarkably, the set of stabilizer aligned states is closed under tensor product, that is $\phi\otimes \psi \in \mathcal{S}$ whenever $\phi,\psi\in \mathcal{S}$.
Moreover, we show that the stabilizer fidelity is not multiplicative for all states
$\phi \notin \calS$. That is, for any $\phi \notin \calS$ there exists a state $\psi$
such that $F(\phi\otimes \psi)>F(\phi) F(\psi)$. In that sense, our results provide
necessary and sufficient conditions under which the stabilizer fidelity is multiplicative. 

Proposition \ref{thm:prod} enables computation of $\xi(\psi)$ if $\psi$ is a tensor product of
few-qubit states (that involve at most three qubits).
We now describe another large subclass of states $\psi$ relevant for quantum circuit simulation for which we are able to compute $\xi$. To describe these states, 
recall that any diagonal $t$-qubit gate $V$ can be performed using 
a state-injection gadget that contains only stabilizer operations and consumes an ancillary state
$\ket{V}=V\ket{+}^{\otimes t}$ (see the discussion in Section~\ref{sec:simulations} and Figure~\ref{fig_injection}). Here and below $|+\rangle \equiv (|0\rangle + |1\rangle)/\sqrt{2}$.
The gadget also involves a computational basis measurement over $t$ qubits. Let $\vec{x}\in \{0,1\}^t$ be a string of measurement outcomes.
The desired gate $V$ is performed whenever $\vec{x}=0^t$.
However, given some other outcome $\vec{x}\ne 0^t$,
the gadget implements a gate  $V_\vec{x}=C_\vec{x} V$ where 
\[
C_\vec{x}=\prod_{j\, : \, x_j=1} VX_jV^\dag ,
\]
is the required correction, where $X_j$ is the Pauli $X$ operator acting on the jth qubit.
A special class of unitaries are those where the correction $C_\vec{x}$ is always a Clifford operator.
In this case a unitary gate $V$ is equivalent to the preparation of the ancillary state $|V\ra$
modulo stabilizer operations. 
This motivates the following definition.

\begin{dfn}[\bf Clifford magic states]
\label{Dfn_CMS}
Let $V$ be a diagonal $t$-qubit unitary such that $V X_j V^\dagger$ is a Clifford operator for all 
$j$. Such unitary $V$ is said to belong to the $3^{\mathrm{rd}}$ level of the Clifford hierarchy (see e.g. Ref.~\cite{CliffHier}).   The ancillary state  $\vert V \rangle \equiv V \vert +\rangle^{\otimes t}$ is called
a Clifford magic state. 
\end{dfn}
\noindent
For example, $|T\ra^{\otimes m}$   is a Clifford magic state for any integer $m$.
Note that in general the set of Clifford magic states is closed under tensor product. 
\begin{prop}
Let $\psi$ be a Clifford magic state. Then 
$\xi(\psi)=F(\psi)^{-1}$.
	\label{thm:clifmagic}
\end{prop} 
The proof of Proposition \ref{thm:clifmagic} is provided in Section \ref{Sec_Clifford_Magic_States} where it is extended to a slightly broader class of $\psi$. 

We note that $|T^{\otimes m}\rangle$ is  a Clifford magic state \textit{and} a product state and so either Proposition \ref{thm:prod} or Proposition \ref{thm:clifmagic} could be used along with Eq.~(\ref{cstar}) to upper bound its approximate stabilizer rank. In this way one can easily reproduce the upper bound obtained in Ref.~\cite{bravyi2016improved},
namely, 
\begin{equation}
\chi_\delta(T^{\otimes m}) \le O\left(\delta^{-2} \cos{(\pi/8)}^{-2m}\right).
\label{eq:chiT}
\end{equation}
This stands in sharp contrast with the best known lower bound $\chi(T^{\otimes m})=\Omega(m^{1/2})$ established in Ref.~\cite{Bravyi16stabRank}. It should be expected that the stabilizer rank (either exact or approximate)
of the magic states $T^{\otimes m}$ grows exponentially with $m$ 
in the limit $m\to \infty$. Indeed, the polynomial scaling of $\chi_{\delta}(T^{\otimes m})$ with $m$ for a suitably small constant $\delta$, or $\chi(T^{\otimes m})$, would entail complexity theoretic heresies such as BQP=BPP,  or P=NP~\footnote{By simulating a postselective quantum circuit one could solve 3-SAT using a polynomial number of T-gates, see e.g., Ref. \cite{alibaba}.}.  
Remarkably, we have no techniques for proving unconditional super-polynomial
lower bounds.
Here we made partial progress by solving a simplified problem
where stabilizer decompositions of $T^{\otimes m}$ are restricted to
certain product states. For this simplified setting
we prove a tight lower bound on the approximate stabilizer rank of $T^{\otimes m}$
matching the upper bound of Ref.~\cite{bravyi2016improved}.
 To state our result it is more
convenient to work with the magic state $|H\ra=\cos{(\pi/8)}|0\ra + \sin{(\pi/8)}|1\ra$
which is equivalent to $|T\ra$ modulo Clifford gates.
Ref.~\cite{bravyi2016improved} showed that $|H^{\otimes m}\ra$ admits
an approximate stabilizer decomposition $|H^{\otimes m}\ra \approx \sum_{\alpha=1}^k c_\alpha |\phi_\alpha\ra$
where $k\sim \cos{(\pi/8)}^{-2m}$ and $\phi_\alpha$ are product
stabilizer states of the form 
\begin{equation}
|\tilde{x}\rangle =|\tilde{x}_1\rangle \otimes \cdots \otimes |\tilde{x}_m\rangle
\quad \mbox{where} \quad
|\tilde{0}\rangle=|0\rangle \quad \mbox{and} \quad
|\tilde{1}\rangle =|+\rangle. \label{eqn:Hbasis}
\end{equation}
Here $x_i\in \{0,1\}$.
These are the stabilizer states that achieve the maximum overlap with $|H^{\otimes m}\ra$, see Ref.~\cite{bravyi2016improved}.  Here we prove the following lower bound.
\begin{prop}
\label{prop:lower_bound}
Suppose $S\subseteq \{0,1\}^m$ is an arbitrary subset
and $\psi$ is an arbitrary linear combination of states
$|\vec{\tilde{x}}\rangle$ as in \eqref{eqn:Hbasis} with $\vec{x}\in S$ such that $\|\psi\|=1$.  Then 
\begin{equation}
\label{eq1ultra}
|S|\ge |\langle H^{\otimes m}|\psi\rangle |^2 \cdot \cos{(\pi/8)}^{-2m}.
\end{equation}
\end{prop}
The proof of this result which is given in Section \ref{Sec_ultra} makes use of the machinery of ultra-metric  matrices~\cite{MMM,NabenVarga}. We hope that these techniques may lead to further progress   on lower bounding the stabilizer rank.

We conclude this section by summarizing our results pertaining to the exact stabilizer rank.
Prior work focused exclusively on finding the stabilizer rank of $m$-fold tensor products
of magic state $|T\ra$.
A surprising and counter-intuitive result of Ref.~\cite{Bravyi16stabRank}
is that  for small number of magic states ($m\le 6$) the stabilizer rank $\chi(T^{\otimes m})$
scales linearly with $m$. Meanwhile,   $\chi(T^{\otimes m})$ is expected to 
scale exponentially with $m$ in the  limit $m\to \infty$.
Using a numerical search we observed a sharp jump from $\chi(T^{\otimes 6})=7$
to $\chi(T^{\otimes 7})=12$ indicating a transition from the linear to the exponential scaling
at $m=7$. This poses the question of whether other magic states have a linearly scaling stabilizer rank 
(until some critical $m$ is reached) or if $|T\ra$ is an exceptional state due to its special symmetries.
Here we show that the linear scaling for small $m$ is a generic feature.  
\begin{theorem}[\bf Upper bound on $\mathbf{\chi}$]
	\label{Thm_many_copies}
	Let $\psi$ be an $n$-qubit state and then for all $m \leq 5$ we have
	\begin{equation}
	\chi( \psi^{\otimes m} ) \leq \binom{2^n + m -1}{m}
	\end{equation}
	where the round brackets denote the binomial coefficient.
\end{theorem}
For example, this result shows that for any diagonal single-qubit unitary $V$ the associated magic state $\ket{V}$
obeys $\chi(\ket{V}^{\otimes m})\le m+1$ for $m \leq 5$.  For larger $m$, an exponential scaling is expected.  
The proof of Theorem~\ref{Thm_many_copies} (given in Section \ref{Sec_Symmetric_states}) exploits well-known properties of the symmetric subspace and a recently established fact that $n$-qubit stabilizer states form a 3-design~\cite{Webb16,kueng15}. 

\subsection{Subroutines for manipulating low-rank stabilizer decompositions}
\label{sec:subroutines}
Suppose $U$ is a quantum circuit acting on $n$ qubits.
We consider a classical simulation task where the goal is to
sample a bit string $x\in \{0,1\}^n$ from the probability
distribution $P_U(x)=|\langle x|U|0^n\rangle|^2$ with a small statistical error. 

Suppose we are given an approximate stabilizer decomposition of a state $U\ket{0^n}$:
\be
\label{approxU}
\|U|0^n\ra - |\psi\ra\|\le \delta, \qquad |\psi\ra = \sum_{\alpha=1}^k b_\alpha U_\alpha |0^n\ra
\ee
for some coefficients $b_\alpha$ and some Clifford circuits $U_\alpha$. In Section \ref{algorithms} we give algorithms for the following tasks. These algorithms are the main subroutines used in our quantum circuit simulators.
\begin{enumerate}
\item[\bf (a)] Sample $\vec{x}\in \{0,1\}^n$ from the probability distribution
\be
\label{P(x)normalized}
P(\vec{x})=\frac{|\la \vec{x}|\psi\ra|^2}{\|\psi\|^2}.
\ee
\item[\bf (b)] Estimate the norm $\|\psi\|^2$ with a small multiplicative error.
\end{enumerate}
Note that if $\delta$ is small then $P(\vec{x})$ approximates the true  output distribution $P_U(\vec{x})=|\la \vec{x}|U|0^n\ra|^2$ with a small error. Indeed, Eq.~(\ref{approxU}) gives $\| P - P_U\|_1 \le O(\delta)$.

The tasks (a,b)  are closely related. Using the chain rule for conditional probabilities
one can  reduce the sampling task 
to estimation of marginal probabilities of $P(\vec{x})$.
Any marginal probability can be expressed 
as  $\| \Pi \psi \|^2/\|\psi\|^2$, where $\Pi$
is a tensor product of projectors $|0\ra\la 0|$, $|1\ra\la 1|$, and the identity operators. Note that such projectors map stabilizer states to stabilizer states.
Thus $\Pi |\psi\ra$ admits a stabilizer decomposition with $k$ terms 
that can be easily obtained from Eq.~(\ref{approxU}).
Accordingly, 
task~(a) reduces to a sequence of norm estimations for low-rank stabilizer superpositions,
see Section~\ref{Sec_fast_norm} for details.

Section~\ref{clifford_sim} describes a fast Clifford simulator that transforms a
stabilizer state $U_\alpha|0^n\ra$ 
into a certain canonical form 
which we call a CH-form. 
It is analogous to the stabilizer
tableaux~\cite{aaronson04improved} but includes information about the global phase of a state. 
This allows us to simulate each circuit $U_\alpha$ 
in the superposition Eq.~(\ref{approxU})
independently without destroying information about 
the relative phases.  Our C++ implementation of the simulator
performs approximately $5\times 10^6$ Clifford gates per second
for $n=64$ qubits on a laptop computer.

Section~\ref{heuristic} describes a heuristic algorithm for the task~(a).
We construct a Metropolis-type Markov chain such that $P(x)$ is the unique steady distribution of the chain
(under mild additional assumptions).
We show how to implement each Metropolis step in time $O(k n)$.
Unfortunately, the mixing time of the chain is generally unknown.

Section~\ref{Sec_fast_norm} gives an algorithm for the task~(b).
It exploits the fact that the inner product $\la \phi|\phi'\ra$  between two $n$-qubit stabilizer
states $\phi,\phi'$ can be computed exactly in time $O(n^3)$,
see Ref.~\cite{garcia2012efficient,bravyi2016improved}.
We adapt this inner product algorithm to the CH-form 
of stabilizer states in Section~\ref{Sec_fast_norm}.
The naive method of computing the norm relies on the identity
$\|\psi\|^2=\sum_{\alpha,\beta=1}^k b_\alpha^* b_\beta \la \phi_\alpha|\phi_\beta\ra$, where
$|\phi_\alpha\ra=U_\alpha|0^n\ra$. Evaluating all cross terms using the inner 
product algorithm would take time $O(k^2 n^3)$ which
is impractical for large $k$.
Instead, Ref.~\cite{bravyi2016improved} proposed 
a method of estimating, rather than evaluating, the norm.  
It works by computing inner products between $\psi$ and random stabilizer states
drawn from the uniform distribution. 
This method has runtime $O(k n^3)$
offering a significant speedup in the relevant regime of large rank decompositions. 
Here we propose an improved version of this norm estimation 
method combining both conceptual  and implementation improvements.
The new version of the norm estimation subroutine achieves approximately 
50X speedup compared with Ref.~\cite{bravyi2016improved}.

Section~\ref{Sec_fast_norm}  also
describes a rigorous  algorithm  for the task~(a)
based on the norm estimation and the chain rule for conditional probabilities.
It has runtime $O(k n^6)$ which quickly becomes impractical.
However,  if our goal is to sample only $w$ bits from $P(x)$, 
the runtime is only  $O(k n^3 w^3)$.
Thus the sampling method based on the norm estimation may be practical for small values of $w$.

\subsection{Simulation algorithms}
\label{sec:simulations}
Here we describe how to combine ingredients from previous sections to obtain classical simulation algorithms for quantum circuits.  We consider a circuit
\begin{equation}
U=D_m V_mD_{m-1}V_{m-1}\ldots D_1 V_1 D_0
\label{eq:circuit}
\end{equation}
acting on input state $|0^n\rangle$, where $\{D_j\}$ are Clifford circuits and $\{V_j\}$ are non-Clifford gates.  We discuss three different methods:  gadget-based simulation (using either a fixed-sample or random-sample method as described below) and sum-over-Cliffords simulation.  

Let us first summarize the simulation cost of different methods. The gadget-based methods from Refs.~\cite{Bravyi16stabRank, bravyi2016improved} can be used to simulate quantum circuits Eq.~\eqref{eq:circuit} where $\{V_j\}$ are single-qubit T gates. Using the (random-sample) gadget-based method, the asymptotic cost of sampling from a distribution $\delta$-close in total variation distance to the output distribution $P_U(x)=|\langle x|U|0^n\rangle|^2$ is
\begin{equation}
\tilde{O}\left(\chi_{\delta}\left(|T^{\otimes m}\rangle\right)\right)\leq \tilde{O}\left(\delta^{-2} \xi\left(|T^{\otimes m}\rangle\right)\right)=\tilde{O}\left(\delta^{-2}\left(\cos(\pi/8)\right)^{-2m}\right),
\label{eq:CTscaling}
\end{equation}
where we used Theorem \ref{thm:randomCvec} and Proposition~\ref{thm:prod}, and the $\tilde{O}$-notation suppresses a factor polynomial in $m$, $n$, and $\log(\delta^{-1})$, see Ref.~\cite{bravyi2016improved} for details.

We will see how the gadget-based approach can be applied in a slightly more general setting where the circuit contains diagonal gates from the third level of the Clifford hierarchy. Then we introduce the sum-over-Cliffords simulation method which can be applied much more generally.  The cost of $\delta$-approximately sampling from the output distribution $P_U$ for the circuit Eq.~\eqref{eq:circuit} using the sum-over-Cliffords method can be upper bounded as
\begin{equation}
\tilde{O}\bigg(\delta^{-2} \prod_{j=1}^{m} \xi(V_j)\bigg)
\label{eq:zrot}
\end{equation}
where the definition of $\xi$ is extended to unitary matrices in a natural way (see below for a formal definition). For example, if each non-Clifford gate is a single-qubit diagonal rotation of the form $V_j=R(\theta_j)=e^{-i(\theta_j/2) Z}$ with $\theta_j\in [0, \pi/2)$ then we will see that $\xi(V_j)=\xi(V_j|+\rangle)$ and the simulation cost is
\[
\tilde{O}\bigg(\delta^{-2} \prod_{j=1}^{m} \xi(V_j|+\rangle)\bigg)=\tilde{O}\bigg(\delta^{-2} \prod_{j=1}^{m} \left(\cos(\theta_j/2)+\tan(\pi/8)\sin(\theta_j/2)\right)^2\bigg).
\]
In the case $\theta_j=\pi/4$ where all non-Cliffords are $T$ gates, we see that the sum-over-Cliffords method achieves the same asymptotic cost Eq.~\eqref{eq:CTscaling} as the gadget-based method from Ref.~\cite{bravyi2016improved}. However the sum-over-Cliffords method is generally preferred because it is simpler to implement and may be slightly faster, as it manipulates stabilizer states of fewer qubits.

\subsubsection{Gadget-based methods}
\label{sec:gadgetbased}
We begin by reviewing the gadget-based methods for simulating circuits expressed over the Clifford+T gate set. A gadget-based simulation directly emulates the operation of a quantum computer that can implement Clifford operations and has access to a supply of magic states.

It is well known that one can perform such a gate on a quantum computer using a state-injection gadget with classical feedforward dependent on measurement outcomes.  In particular, a $t$-qubit gate $V$ can be implemented by a gadget consuming a magic state $\ket{V}=V \ket{+^ t}$, see  Fig.~\ref{fig_injection} for an example. Let $x\in \{0,1\}^t$ be the measurement outcome.  The gadget implements the desired gate $V$ whenever $x=0^t$.  Otherwise, if  ${x} \neq {0^t}$, the gadget implements a gate $V_x=C^\dagger_x V$ where $C_x$ is the required correction.  If $V$ is in the third level of the Clifford hierarchy, the correction $C_x$ is always a Clifford operator and $\ket{V}$ is a Clifford magic state (recall Definition~\ref{Dfn_CMS}).  Formally, postselecting on outcome $x=0^t$ gives
\begin{equation}
V|\psi\rangle= 2^{t/2} (\id \otimes \bra{0}^{\otimes t}) C^{\prime} |\psi\rangle  \ket{V},
\label{eq:singlegate}
\end{equation}
where $C^{\prime}=\left(\prod_{a=1}^{t} \mathrm{CNOT}_{a,a+t}\right)$ is a Clifford unitary.

Now let $U$ from Eq.~\eqref{eq:circuit} be the full circuit to be simulated and suppose $V_j$ is a diagonal $t_j$-qubit gate. Write $\tau=t_1+t_2\ldots+t_m$. If we replace each non-Clifford gate with the corresponding state-injection gadget we obtain a ``gadgetized'' circuit with $n+\tau$ qubits acting on input state $|0^n\rangle|V_1\rangle|V_2\rangle\ldots |V_m\rangle$.  The gadgetized circuit contains $\tau$ extra single-qubit measurements and Clifford gates. If we postselect the measurement outcomes on $0^{\tau}$ we obtain an identity (cf. Eq.~\eqref{eq:singlegate})
\begin{equation}
U|0^n\rangle=2^{\tau/2} (\id \otimes \bra{0}^{\otimes \tau}) C|0^n\rangle|\Psi\rangle\qquad \quad |\Psi\rangle=|V_1\rangle|V_2\rangle\ldots |V_m\rangle
\label{eq:Uidentity}
\end{equation}
where $C$ is an $n+\tau$-qubit Clifford unitary and we have collected together all of the required magic states into the $\tau$-qubit state $\Psi$. We see a renormalisation factor $2^{\tau/2}$ is required to account for post-selection. 

Eq.~\eqref{eq:Uidentity} shows that the output state $U|0^n\rangle$ of interest has exact stabilizer rank equal to that of the magic state $\Psi$, i.e., $\chi(U|0^n\rangle)=\chi(\Psi)$. Indeed, starting from an exact stabilizer decomposition of $\ket{\Psi}$, we can apply $(\id \otimes \bra{0}^{\otimes \tau}) C$ to each stabilizer state in the decomposition and renormalize to obtain an exact stabilizer decomposition of the output state $U\ket{0^n }$.  Once we have computed an exact stabilizer decomposition of $U|0^n\rangle$ we may use the subroutines from Section~\ref{sec:subroutines} to simulate the quantum computation. For example we may sample from the output distribution $P_U$ or compute a given output probability $P_U(x)$. This was the approach taken in Ref.~\cite{Bravyi16stabRank} and here we call this a fixed-sample gadget-based simulator since it postselects on a fixed single measurement outcome.

Note that in the fixed-sample method one must use an exact (rather than approximate) stabilizer decomposition of the resource state $\Psi$. Indeed, in a fixed-sample simulation if $\ket{\Psi_{\mathrm{\delta}}}$ approximates $\ket{\Psi}$ up to an error $\delta$ then the simulation error could be amplified to $2^{\tau/2} \delta$ when substituting in Eq.~\eqref{eq:Uidentity}.  

The random-sample gadget-based simulation method is a different approach that allows us to use approximate stabilizer decompositions within this framework. Here one selects the post-selected measurement outcome $x\in \{0,1\}^\tau$ uniformly at random.  However, now  we have some measurement outcomes other than $x = 0^\tau$ and so have to account for corrections $C_x$.  Clifford corrections are straightforwardly simulated and this is ensured provided each non-Clifford gate $V_j$ in the circuit is diagonal in the computational basis and contained in the third level of the Clifford hierarchy (e.g., the T gate and CCZ gate). This guarantees that the simulation consuming an approximate magic state $|\Psi_\delta\ra$ achieves an average-case simulation error $O(\delta)$, see Ref.~\cite{bravyi2016improved} for details. 

An important distinction between the two gadget-based methods is that the random-sample method allows one to sample from a probability distribution which approximates $P_U$ but--unlike the fixed-sample method-- in general cannot be used to obtain an accurate estimate of an individual output probability $P_U(x)$.

\begin{figure}
	\centering
	\includegraphics{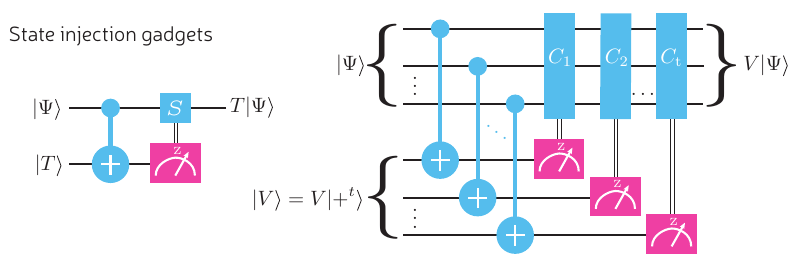}
	\caption{State injection gadgets for single-qubit $T$ gate and general multi-qubit phase gate $V$.  A correction unitary $V X_j V^\dagger$ is required whenever measurement $j$ registers a ``1" outcome.  If all corrections are Clifford then gadgets can be deployed with no additional resource requirements.}
	\label{fig_injection}
\end{figure}

\subsubsection{Sum-over-Cliffords method}
\label{sec:sum_cliffords}
Let $U$ be the quantum circuit Eq.~\eqref{eq:circuit} to be simulated.  We shall construct a sum-over-Cliffords decomposition
\begin{equation}
U=\sum_j c_j K_j
\label{eq:Udecomp}
\end{equation}
where each $K_j$ is a unitary Clifford operator and $c_j$ are some coefficients. This gives
\begin{equation}
U|0^n\rangle=\sum_j c_j K_j|0^n\rangle.
\label{eq:Vphi}
\end{equation}
Applying  Theorem~\ref{thm:randomCvec} one can approximate $U|0^n\ra$ within any desired error $\delta$ by a 
superposition of stabilizer states $\psi$ that contains  
\begin{equation}
k \approx \delta^{-2}  \| \vec{c} \|_1^2
\label{eq:chisim}
\end{equation}
terms.  In this way we can compute an approximate stabilizer decomposition $\psi$ satisfying
\be
\label{approxU1}
\|U|0^n\ra - |\psi\ra\|\le \delta, \qquad |\psi\ra = \sum_{\alpha=1}^k b_\alpha U_\alpha |0^n\ra ,
\ee
for some coefficients $b_\alpha$ and some Clifford circuits $U_\alpha$. Using the methods summarized in the previous section we can then sample from the distribution $P(x)=|\langle x|\psi\rangle|^2$ which $\delta$-approximates the output distribution $P_U$. In particular, one can use either the heuristic Metropolis sampling technique or the rigorous algorithm using norm estimation, which has runtime upper bounded as $O(kn^6)$.

The sum-over-Cliffords decomposition Eq.~\eqref{eq:Udecomp} of $U$ can be obtained by combining decompositions of the constituent non-Clifford gates. If $V_p=\sum_{j} c_j^{(p)} K_j^{(p)}$ for $p=1,2,\ldots, m$, then substituting in Eq.~\eqref{eq:circuit} gives
\[
U=\sum_{j_1,\ldots,j_m} \left(\prod_{p=1}^{m} c_{j_p}^{(p)}\right) D_m K_{j_m}^{(m)}D_{m-1}\ldots D_1 K^{(1)}_{j_1} D_0
\]
which is of the form Eq.~\eqref{eq:Udecomp} with $\|c\|_1^2=\prod_{p=1}^{m} \|c^{(p)}\|_1^2$.  This motivates the following generalization of $\xi$ to unitary operators.

\begin{dfn}[\bf Stabilizer Extent for unitaries, cf. Eq. \ref{eq:zrot}]
	Suppose $W$ is a unitary operator. Define $\xi(W)$ as the minimum of $\| c \|^2_1$ over all
	 decompositions  $W=\sum_j c_j K_j$ where $K_j$ are Clifford unitaries.
\end{dfn}
\noindent
This implies 
\begin{equation}
\xi(U|0^n\rangle)\leq \xi(U) \leq \prod_j \xi(V_j).
\label{eq:cstarU}
\end{equation}

Thus, given $\xi$-optimal decompositions of each non-Clifford gate in the circuit, the asymptotic cost of $\delta$-approximately sampling from $P_U(x)$ using the norm estimation algorithm and the sum-over-Cliffords method is  $\tilde{O}(k)$, and substituting Eq.~\eqref{eq:cstarU} in Eq.~\eqref{eq:chisim} we recover Eq.~\eqref{eq:zrot}.

Note that for any gate $V_j$ which acts on $O(1)$ qubits we may compute a $\xi$-optimal sum-over-Cliffords decomposition in constant time by an exhaustive search. Below we describe decompositions for commonly used non-Clifford gates. We use the following lemma which ``lifts''  a stabilizer decomposition of the resource state $|V\rangle=V|+^t\rangle$ to a sum-over-Cliffords decomposition of $V$.

\begin{lemma}[\bf Lifting lemma]
Suppose $V$ is a diagonal $t$-qubit unitary and
	\begin{equation}
		\label{Eq_Udecomp}
			V|+^t\rangle=\ket{V} = \sum_j c_j \ket{ \phi_j }.
	\end{equation}	
Suppose further that $ \ket{ \phi_j }$ are equatorial stabilizer states so that $ \ket{ \phi_j }= K_j \ket{+^t}$ where $K_j$ is a diagonal Clifford for all $j$. Then
	\begin{equation}
		\label{Eq_UdecompUnitary}
	V = \sum_j c_j K_j ,
\end{equation}	
and therefore $\xi(V) \leq || c ||_1^2$.  Furthermore, if the equatorial stabilizer decomposition Eq.~(\ref{Eq_Udecomp}) achieves the optimal value $\| c \|_1^2 = \xi(\ket{V})$ then  $\xi(\ket{V})=\xi(V)$.
\end{lemma}
\begin{proof}
Since $U$ and $\{K_j\}$ are diagonal in the computational basis we may write
\begin{equation}
V= \sum_{x} e^{i \theta(x) } \kb{x}{x} \qquad K_j = \sum_{x}  e^{i \theta_{j}(x)} \kb{x}{x}
\label{eq:lift1}
\end{equation}
where $\theta, \theta_j$ are functions $\mathbb{F}_2^t \rightarrow \mathbb{R}$.  For all $x\in \{0,1\}^t$ we have
\begin{equation}
\frac{1}{2^{t/2}} e^{i\theta(x)}=\langle x|V|+^t\rangle=\langle x|\sum_{j} c_j K_j|+^t\rangle=\frac{1}{2^{t/2}}  \sum_{j} c_j  e^{i \theta_j(x)}
\label{eq:lift2}
\end{equation}
Combining Eqs.~(\ref{eq:lift1},\ref{eq:lift2}) and cancelling the factors of $2^{-t/2}$ gives Eq.~\eqref{Eq_UdecompUnitary} and the remaining statements of the lemma are immediate corollaries.
\end{proof}

For single-qubit diagonal rotations $R(\theta)=e^{-i(\theta/2) Z}$,  we have
\begin{equation}
	R(\theta) \ket{+} =  \left(\cos(\theta/2)-\sin(\theta/2)\right) \ket{+} +  \sqrt{2} \sin(\theta/2) e^{- i \pi / 4} S \ket{+},
\end{equation}	
which is an optimal decomposition with respect to $\xi$ and is similar to Eq.~\eqref{Eq_SingleQubitsDecomp}.  Therefore, we can use the lifting lemma to obtain an optimal decomposition
\begin{equation}
	R(\theta)  =  \left(\cos(\theta/2)-\sin(\theta/2)\right) \id +  \sqrt{2} e^{-  i \pi / 4} \sin(\theta/2) S \label{eqn:Rtheta}
\end{equation}
and conclude 
\begin{equation}
	\xi( R(\theta)  ) = \xi( R(\theta) \ket{+}   )  = \left(\cos(\theta/2)+\tan(\pi/8)\sin(\theta /2)\right)^2.
\end{equation}

The doubly controlled $Z$ gate (CCZ) is another useful example.  In Section~\ref{Sec_Clifford_Magic_States} we show that
\begin{align}
 \ket{CCZ}   =   \frac{2}{9} &( \id+ CZ_{1,2}X_3 )(\id + CZ_{1,3} X_2)(\id + CZ_{2,3}X_1)\ket{+^3}, \\ \nonumber
=	    \frac{2}{9} & \big(  \id + CZ_{1,2} + CZ_{1,3} + CZ_{2,3} + CZ_{1,2}CZ_{1,3}Z_1 + CZ_{1,2} CZ_{2,3} Z_2   \\ \nonumber 
	&  + CZ_{1,3} CZ_{2,3} Z_3  - CZ_{1,2}CZ_{1,3}CZ_{2,3} Z_1 Z_2 Z_3 \big)  \ket{+^3} ,
\end{align}	
is an optimal decomposition with respect to $\xi$.  Deploying the lifting lemma we have
\begin{align}
	CCZ   = \frac{2}{9} & \big(  \id + CZ_{1,2} + CZ_{1,3} + CZ_{2,3} + CZ_{1,2}CZ_{1,3}Z_1 + CZ_{1,2} CZ_{2,3} Z_2  \label{eqn:CCZ_decomposition}   \\ \nonumber 
	& + CZ_{1,3} CZ_{2,3} Z_3  - CZ_{1,2}CZ_{1,3}CZ_{2,3} Z_1 Z_2 Z_3 \big) ,
\end{align}	
and conclude
\begin{equation}
	\xi( CCZ  ) = \xi( \ket{CCZ}   )  = 16/9 .
\end{equation}
Recall that since this is a Clifford magic state we have  $ \xi( \ket{CCZ}   ) = 1 / F( \ket{CCZ})$ and notice that the stabilizer fidelity is achieved by the equatorial stabilizer state $\ket{+^3}$.  We remark that the above recipe for an optimal sum-over-Cliffords decomposition can be generalised to any Clifford magic state for which the stabilizer fidelity is achieved by some equatorial stabilizer state.

These optimal sum-over-Cliffords decompositions will be used in the numerics of the following Section. 

\subsection{Implementation and simulation results}
\label{sec:numericresults}
In this section we  report numerical results obtained  by simulating two quantum algorithms.
First, we use the sum-over-Cliffords method to simulate 
the Quantum Approximate Optimization  (QAOA) algorithm due to Farhi et al~\cite{farhi2014quantum}. This algorithm allows us to explore the performance of our simulator for circuits containing Cliffords and diagonal rotations. 
This simulation involves $n=50$ qubits, about $60$ non-Clifford gates, and a few hundred Clifford gates. We note that QAOA circuits have been previously used to benchmark classical simulators in Ref.~\cite{fried2017qtorch}. Secondly, we simulate the Hidden Shift algorithm for bent functions due to  Roetteler~\cite{Roetteler09}.
This algorithm  was also used to benchmark the Clifford+$T$ simulator of
Ref.~\cite{bravyi2016improved} which, in the terminology of the previous section, is a gadget-based simulator where sparsification is achieved via suitable choice of a random linear code. We extend this methodology to a Clifford+$CCZ$ simulator of the same circuits. We also simulate the Hidden Shift circuits using the new Sum-over-Cliffords method wherein sparsification is achieved by appealing to the $\xi$ quantity.

\subsubsection{Quantum approximate optimization algorithm}
\begin{figure}[ht]
\centering
\includegraphics[height=7.5cm]{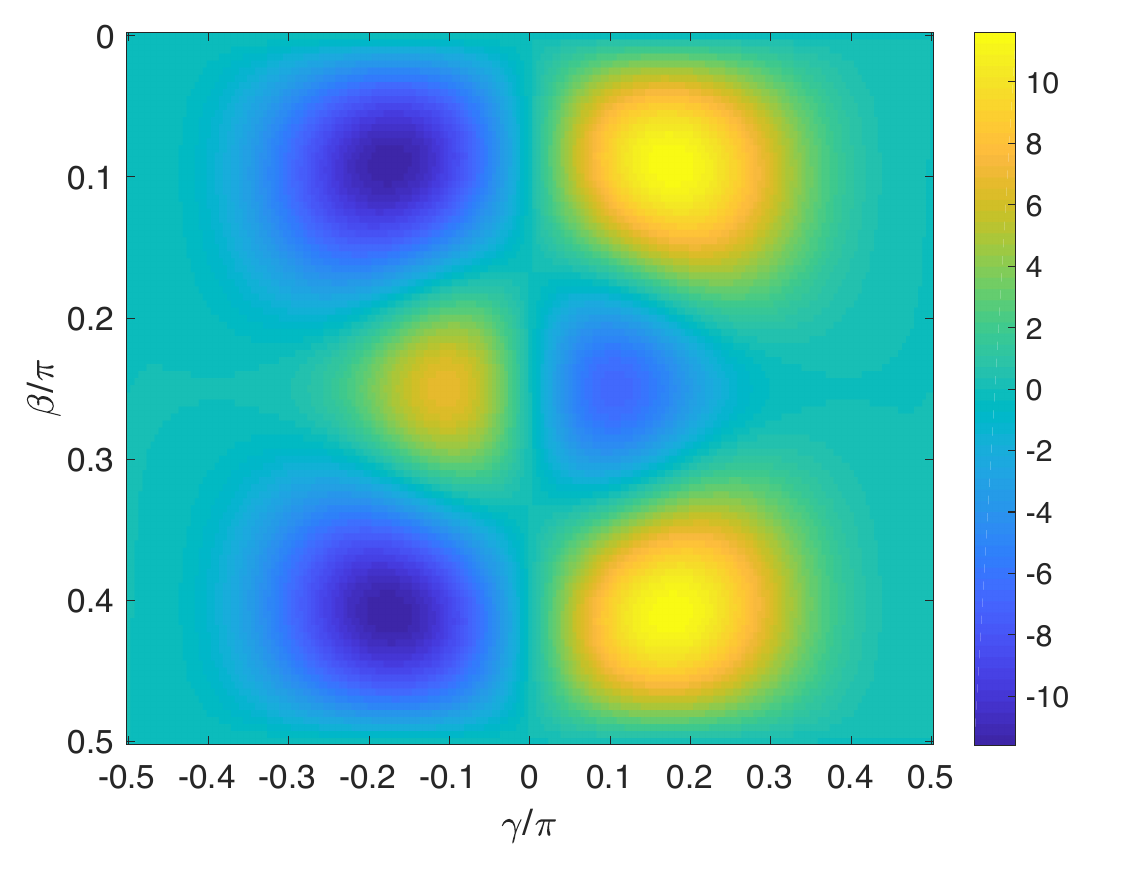}
\caption{The expected value of the cost function $E(\beta,\gamma)$
computed using the Monte Carlo method 
by Van~den~Nest~\cite{nest2009simulating}.
We consider a randomly generated instance of the Max E3LIN2 problem
with $n=50$ qubits and degree $D=4$.
}
\label{fig:QAOA3}
\end{figure}

Here we consider 
the Quantum Approximate Optimization Algorithm applied to the Max E3LIN2 problem~\cite{farhi2014quantum}.
The problem is to maximize  an objective function 
\[
C=\frac12 \; \sum_{1\le u<v<w\le n}\;  d_{uvw} z_u z_v z_w
\]
that depends on $n$ binary variables $z_1,\ldots,z_n\in\{-1,1\}$.
Here $d_{uvw}\in \{0,\pm 1\}$ are some coefficients. 
Let
\[
m=\sum_{u<v<w} |d_{uvw}| 
\]
be the number of non-zero terms in $C$.
Let us say that an instance of the E3LIN2 problem
has degree $D$ if each variable $z_u$
appears in exactly $D$ terms $\pm z_u z_v z_w$
(depending on the values of $n$ and $D$ there could be
one variable that appears in less than $D$ terms).

Following Ref.~\cite{farhi2014quantum} we
consider a family of variational states
\[
|\psi_{\beta,\gamma}\ra = U|0^n\rangle  \qquad \quad U=e^{-i\beta  B} e^{-i\gamma \hat{C}} H^{\otimes n}
\]
where  $\beta,\gamma\in \RR$ are variational parameters,
$B=X_1+\ldots+X_n$ is the transverse field operator,
and $\hat{C}$ is a diagonal operator obtained from $C$
by replacing the variables $z_u$ with the Pauli operators $Z_u$. 
The QAOA algorithm attempts to choose $\beta$ and $\gamma$  maximizing
the expected value of the objective function,
\[
E(\beta,\gamma)=\la \psi_{\beta,\gamma} |\hat{C}|\psi_{\beta,\gamma}\ra.
\]
Once a good choice of $\beta,\gamma$ is made, the QAOA algorithm
samples  $z\in \{-1,1\}^n$  from a probability distribution
$P(z)=|\la z|\psi_{\beta,\gamma}\ra|^2$ by 
preparing the state $|\psi_{\beta,\gamma}\ra$ on a quantum computer 
and measuring each qubit of 
$|\psi_{\beta,\gamma}\ra$.
(In this section we assume that output bits take values $\pm 1$ rather than $0,1$.)
By definition, the expected value of $C(z)$  coincides with $E(\beta,\gamma)$.
By generating sufficiently many samples one can produce a string $z$ such that
$C(z)\ge E(\beta,\gamma)$, see Ref.~\cite{farhi2014quantum} for details.

Our numerical results described  below were obtained for 
a single randomly generated instance of the problem with $n=50$ qubits and degree $D=4$.
We empirically observed that the expected value   $E(\beta,\gamma)$ does not depend significantly
on the choice of the problem instance for fixed $n$ and $D$.
Since the cost function has a symmetry $C(-z)=-C(z)$, finding the maximum and the minimum
values of $C$ are equivalent problems. 

A special feature of the QAOA circuits making them suitable
for benchmarking classical simulators is the ability to verify  that the simulator is working properly.
This is achieved by  computing the expected value $E(\beta,\gamma)$ using two independent
methods and cross checking the final answers. Our first method of computing $E(\beta,\gamma)$ is 
 a classical Monte Carlo algorithm due to Van~den~Nest~\cite{nest2009simulating}. 
It allows one to compute expected values
$\la \omega |F|\omega\ra$, where $F$ is an arbitrary sparse Hamiltonian
and $|\omega\ra$ is a so-called computationally tractable state.
Let us choose  $F=e^{i\beta B} \hat{C} e^{-i\beta B}$
and $|\omega\ra=e^{-i\gamma \hat{C}}|+^{\otimes n}\ra$
so that $\la \omega |F|\omega\ra=E(\beta,\gamma)$.
The algorithm of Ref.~\cite{nest2009simulating}
allows one to estimate $\la \omega |F|\omega\ra$
with an additive error $\epsilon$ in time $O(m^4 \epsilon^{-2})$.
The plot of $E(\beta,\gamma)$  is shown on Fig.~\ref{fig:QAOA3}.

Our second method of computing $E(\beta,\gamma)$
is the sum-over-Cliffords/Metropolis simulator
described in Section~\ref{sec:sum_cliffords}. We used this method to simulate
the QAOA circuit $U$ defined above.  For our choice $n=50$ and $D=4$ 
the unitary  $e^{-i\gamma \hat{C}}$ can be implemented by
a circuit that contains $m=66$ $Z$-rotations $e^{i(\gamma/2)Z}$ 
 and a few hundred Clifford gates. To keep the number of non-Clifford gates sufficiently small we restricted the simulations to the line $\beta=\pi/4$.
As can be seen from Fig.~\ref{fig:QAOA3}, this line contains
a local maximum and a local minimum of $E(\beta,\gamma)$ (we note that $\beta=\pi/4$ is also the choice made by Farhi et al.~\cite{farhi2014quantum}). With this choice the cost function is a function of a single parameter $\gamma$ and we may write
\[
E(\gamma)=\la 0^n|U^\dag \hat{C} U|0^n\ra=\sum_{\vec{z} \in \{0,1\}^n} P_U(\vec{z}) C(\vec{z}).
\]

between the ``exact" value $E(\gamma)$  computed by the Monte Carlo method and
its estimate  $E_{sim}(\gamma)$ obtained using the sum-over-Cliffords/Metropolis  simulator (while the Monte Carlo method is not perfect, we expect the errors to be negligible for our purposes). While the plot only shows $\gamma\geq 0$, note that due to the symmetry of the cost function $C(z)=-C(-z)$ we have $E(\gamma)=-E(-\gamma)$. 
The estimate $E_{sim}(\gamma)$ is defined as  
\[
E_{sim}(\gamma)=\frac1s \sum_{j=1}^s C(\vec{z}^j), \qquad s=4\cdot 10^4
\]
where  $\vec{z}^1,\ldots,\vec{z}^s$ are samples from the distribution $P(\vec{z})$
describing the output of the simulator, see Eq.~(\ref{P(x)normalized}).
Generating all of the data used to produce Fig.~\ref{fig:QAOAfull}a took less than 3 days on a laptop computer, with the most costly data points taking several hours. The number of stabilizer states $k$ used to approximate $U|0^n\ra$ is shown in Fig.~\ref{fig:QAOAfull}b; it was chosen as in Eq.~\eqref{eq:chisim} with $\delta\leq 0.15$ for all values of $\gamma$. This toy example demonstrates that our algorithm is capable of processing superpositions of $k\sim 10^6$ stabilizer states for $n=50$ qubits.

\begin{figure}
	\centering
	\includegraphics[width=\columnwidth]{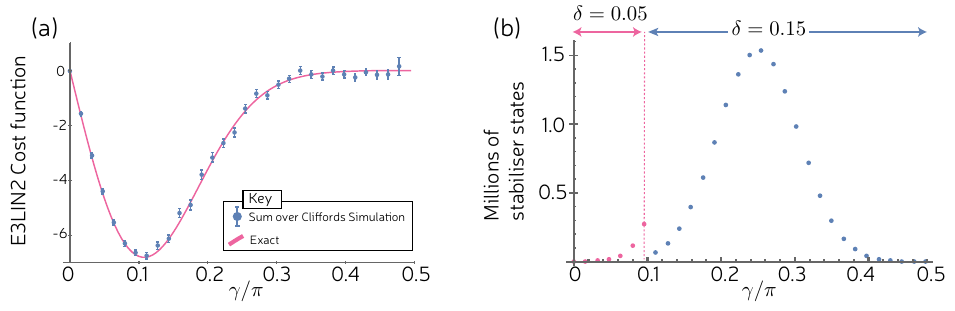}
	\caption{{\em Classical simulation of the QAOA algorithm:}
		(\textit{a}) Comparison between $E(\gamma)$  and its
		estimate $E_{sim}(\gamma)$ obtained using the sum-over-Cliffords/Metropolis simulator.
		We consider a randomly generated instance of the problem with $n=50$ qubits
		and degree $D=4$.
		For each data point  $10^4$ Metropolis steps 
		were performed to approach the steady distribution $P(\vec{z})$.
		The estimate $E_{sim}(\gamma)$ was obtained by
		averaging the cost function $C(\vec{z})$ over a subsequent $s=4\cdot 10^4$ samples 
		$\vec{x}$  from the output distribution of the simulator.
		Error bars represent the statistical error estimated using
		the MATLAB code due to Wolff~\cite{wolff2004monte} (for estimating errors in Markov chain Monte Carlo data) 
		(\textit{b}) The number of stabilizer states $k$ 
		used by the sum-over-Cliffords simulator was chosen as in Eq.~\eqref{eq:chisim} with $\delta=0.05$ for pink data points and $\delta=0.15$ for blue data points.
	}
	\label{fig:QAOAfull}
\end{figure}

\subsubsection{The hidden shift algorithm}
In this section, we describe the results of simulations applied to a family of quantum circuits that solve the Hidden Shift Problem \cite{van_dam_quantum_2006} for non-linear Boolean functions \cite{Roetteler09}. These circuits are identical to those simulated in \cite{bravyi2016improved} and further details of this quantum algorithm and its circuit instantiation can be found in Section F of the Supplemental Material of \cite{bravyi2016improved}. Briefly, the goal is to learn a hidden shift string $s\in \mathbb{F}_2^{n}$ by measuring the output state $|s\ra$ of the circuit $U$ applied to computational basis input $|0^{\otimes n}\ra$. The number of non-Clifford gates in $U$ can easily be controlled (we may choose any even number of Toffoli gates) and so the exponentially growing overhead in simulation time can be observed.

We will use both the gadget-based method of Section \ref{sec:gadgetbased} and the Sum-over-Cliffords method of \ref{sec:sum_cliffords}. Due to the high number of non-Clifford gates the exact stabilizer rank, $\chi$, is prohibitively high and so some sort of sparsification/approximation must be used, leading to $\chi_\delta$ instead. In principle we could apply the sparsification Lemma \ref{lem:randomCvec} in the gadget-based setting, but we prefer to use the random code method of \cite{bravyi2016improved} to enable a comparison with that work. The simulation timings in Fig.~\ref{fig:HiddenShiftTimes} consist of four trend lines which can be broken down as
\begin{itemize}
\item $T_{GB}$: The gadget-based random code method of \cite{bravyi2016improved}, wherein each Toffoli gate in $U$ is decomposed in terms of a stabilizer circuit using 4 $T$ gadgets. When a gadgetized version of $U$ uses a total of $t$  $|T\ra$-type magic states, then $|T^{\otimes t}\ra$ is approximated by a state $|\mathcal{L}\ra$ where $\mathcal{L} \subseteq \mathbb{F}_2^t$ is a linear subspace i.e., random code (Compare with Eq.~\eqref{eq:cmag}). We then have that $\chi_\delta(|T^{\otimes t}\ra)$ is the number of vectors in $\mathcal{L}$.
\item $CCZ_{GB}$: The gadget-based random code method of \cite{bravyi2016improved}, wherein each Toffoli gate in $U$ is implemented via a $CCZ$ gadget (as discussed e.g., in \cite{Howard17robustness}). When gadgetized $U$ uses a total of $u$  $|CCZ\ra$-type magic states, then $|CCZ^{\otimes u}\ra$ is approximated by a state $|\mathcal{L}\ra$ (see Eq.~\eqref{eq:cmag}) where $\mathcal{L} \subseteq \mathbb{F}_2^{3u}$ is a linear subspace/random code and $\chi_\delta(|CCZ^{\otimes u}\ra)=|\mathcal{L}|$.
\item $T_{SoC}$: The Sum-over-Cliffords method outlined in Sec.~\ref{algorithms}, wherein each Toffoli gate in $U$ is decomposed in terms of a stabilizer circuit using 4 $T$ gates. Each $T$ gate is subsequently decomposed into Clifford gates, $T=c_0I+c_1S$, with weightings as in Eq.~\eqref{eqn:Rtheta}.
\item $CCZ_{SoC}$: The Sum-over-Cliffords method outlined in Sec.~\ref{algorithms}, wherein each Toffoli gate in $U$ written as $CCZ$ which is subsequently decomposed (optimally in terms of $\xi$) into Cliffords as in Eq.~\eqref{eqn:CCZ_decomposition}.
\end{itemize}

The quantity that eventually determines the simulation overhead for both the $T$-based and $CCZ$-based schemes is $F$, the overlap with the closest stabilizer state. Recall $\xi(T)=\xi(|T\rangle)=1/{F(|T\rangle)}$ and likewise for $CCZ$. We have
\begin{align}
F(T) &=|\la +|T\ra|^2=\cos(\pi/8)^2=\frac{1}{2}+\frac{1}{2\sqrt{2}}\approx 0.853, \label{eqn:Tovlap}\\
F(CCZ) &=|\la +^{\otimes 3}|CCZ\ra|^2=\left(\frac{3}{4}\right)^2=\frac{9}{16}. \label{eqn:CCZovlap}
\end{align}

Note that we are using the variable $u$ to denote the number of Toffoli (equivalently $CCZ$) gates in our Hidden Shift circuit. Using the Random Code method, for a target infidelity $\Delta$ we chose a corresponding stabilizer rank $2^k$ where \cite{bravyi2016improved} stipulates
\begin{align}
\log_2 k_T &=\lfloor \log_2\left(4 \cos(\pi/8)^{-8u}/\Delta\right)\rfloor,\\
\log_2 k_{CCZ} &= \lfloor \log_2\left(4 \left(\tfrac{3}{4}\right)^{-2u}/\Delta\right)\rfloor.
\end{align}
Using the Sum-over-Cliffords method, for a target error $\delta$ we chose $k$ as in Lemma~\ref{lem:randomCvec} so that
\begin{align}
k_T &=\left\lfloor \left({\cos(\pi/8)}^{-4u}/\delta\right)^2\right\rfloor,\\
k_{CCZ} &=\lfloor \left(({3}/{4})^{-u}/\delta\right)^2\rfloor.
\end{align}
In either case, we see that there are significant savings to be had by using CCZ gates/states directly versus breaking them down into 4 $T$ gates/states each. For a fixed precision the scaling with $u$ (number of $CCZ$ gates) goes as
\begin{align}
T:&\quad \left(\frac{1}{\cos \pi/8}\right)^{8u} \approx   2^{0.914u}, 
\\ \text{vs.}\quad  
CCZ:&\qquad \left(\frac{16}{9}\right)^u\approx 2^{0.83u}.
\end{align}
This is apparent from the different slopes of the $T$- and $CCZ$- based versions of the simulations in Fig.~\ref{fig:HiddenShiftTimes}.

Absolute comparisons between the gadget-based and Sum-over-Cliffords method are complicated by various implementation details and the amount of optimization applied to each (i.e., more in the latter case). Broadly speaking, however, we observe that the Sum-over-Cliffords method is as fast, if not faster, than the gadget-based method. This is true \emph{despite the fact that Sum-over-Cliffords is completely general in its applicability} whereas the gadget-based technique is only applicable for non-Clifford gates from the third level of the Clifford hierarchy (i.e. those with state-injection gadgets having Clifford corrections). Not only can Sum-over-Cliffords handle gates outside the third level, its performance often \emph{improves} in such situations. For example, a circuit with many small-angle rotation gates requires a number, $k$, of samples that is smaller as the rotation angle moves away from $\pi/4$ i.e., the $T$ case (recall  Eq.~\eqref{eqn:Rtheta}).

    \begin{figure*}[t!h]
        \centering
        \begin{subfigure}[b]{0.47\textwidth}
            \centering
            \includegraphics[width=\textwidth]{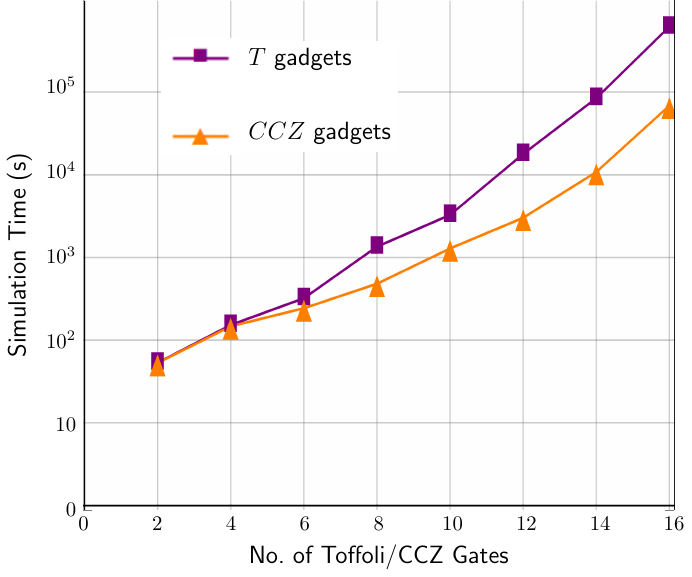}
            \caption[HiddenShift]%
            {{\small Simulation time for Hidden Shift circuits using the gadget-based random code method from Ref.~\cite{bravyi2016improved}.}}    
            \label{fig:HiddenShift_a}
        \end{subfigure}
        \hfill
        \begin{subfigure}[b]{0.47\textwidth}  
            \centering 
            \includegraphics[width=\textwidth]{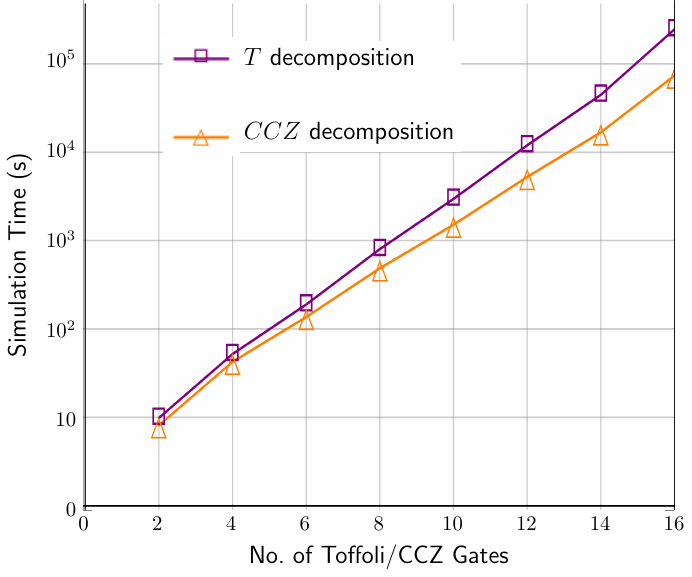}
            \caption[]%
            {{\small Simulation time for Hidden Shift circuits using the Sum-over-Cliffords method from \ref{clifford_sim}.}}    
            \label{fig:HiddenShift_b}
        \end{subfigure}
        \vskip\baselineskip
        \begin{subfigure}[b]{0.5\textwidth}   
            \centering 
            \includegraphics[width=\textwidth]{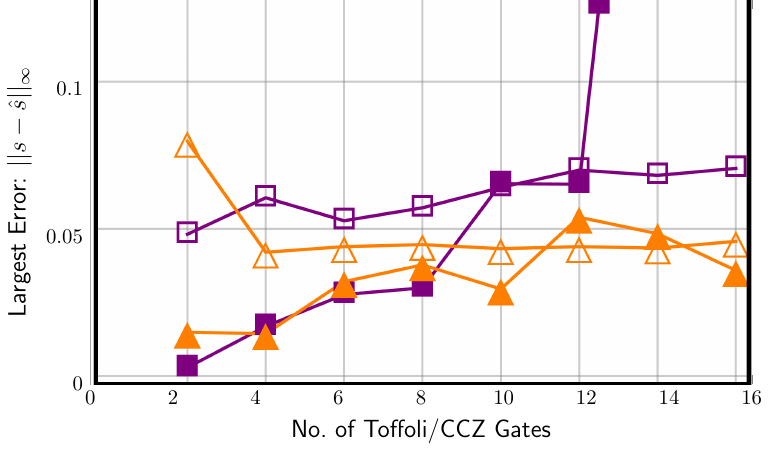}
            \caption[]%
            {{\small Approximation error between the true hidden shift bitstring, $\vec{s}$, and the simulated vector of marginal probabilities, $\vec{\hat{s}}$, for the simulations in Fig.~\ref{fig:HiddenShift_a} and \ref{fig:HiddenShift_b}. The infinity norm gives the largest discrepancy between any individual bit $s_i$ and the corresponding estimate $\hat{s}_i$. Two outlier data points (filled rectangles) whose coordinates are at $(14, 0.304)$ and $(16, 0.512)$, are omitted from this plot for clarity}}    
            \label{fig:HiddenShift_c}
        \end{subfigure}
        \quad
        \begin{subfigure}[b]{0.45\textwidth}   
            \centering 
            \includegraphics[width=\textwidth]{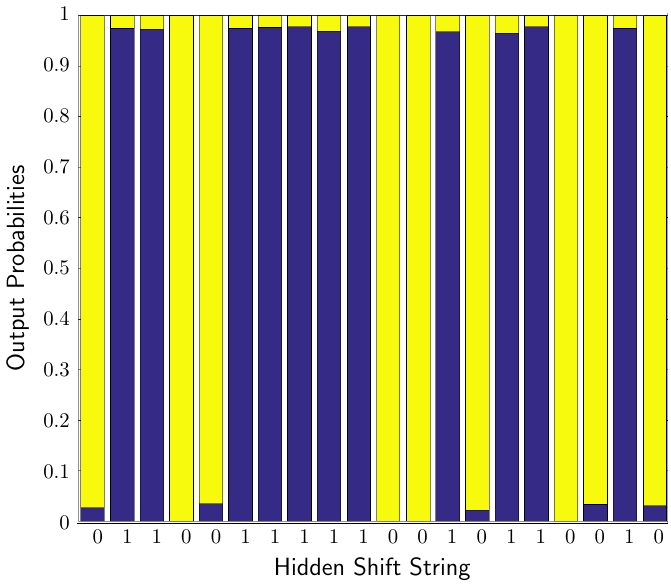}
            \caption[]%
            {{\small Simulated output, $\vec{\hat{s}}$, versus the true shift string, $\vec{s}$, for the case $T_{SoC}$ with 16 Toffoli gates (i.e corresponding to the open rectangle on the right of \ref{fig:HiddenShift_c}).}}    
            \label{fig:HiddenShift_d}
        \end{subfigure}
        \caption[ The average and standard deviation of critical parameters ]
        {\small Timings and errors for simulations of 40-qubit Hidden Shift circuits with varying numbers of non-Clifford gates. Every Toffoli gate is either recast as a $CCZ$ gate (via Hadamards on the target) or as a circuit comprising 4 $T$ gates and additional Stabilizer operations (\cite{bravyi2016improved}). We fixed precision parameters $\delta=0.3$ and $\Delta=0.3$ for the sum-over-Clifford simulations and gadget-based simulations respectively. Simulations were run on Dual Intel Xeon 1.90GHz processors using Matlab. } 
        \label{fig:HiddenShiftTimes}
    \end{figure*}

\section{Discussion}
\label{discussion}
To put our results in a broader context, let us briefly discuss alternative methods for 
classical simulation of quantum circuits. Vector-based simulators~\cite{de2007massively,smelyanskiy2016qhipster,haner20170}
represent $n$-qubit quantum states by complex vectors of size $2^n$
stored in a classical memory.
The state vector is updated
upon application of each gate by performing sparse
matrix-vector multiplication. The memory footprint limits 
the method to small number of qubits. For example, 
H{\"a}ner and Steiger~\cite{haner20170}
reported a simulation of
quantum circuits with $n=45$ qubits and a few hundred gates 
using a  supercomputer with $0.5$ petabytes of memory.
In certain special cases 
the memory footprint can be reduced 
by recasting the simulation problem as a 
tensor network contraction~\cite{markov2008simulating,boixo2017simulation,aaronson2016complexity}.
Several tensor-based simulators have been developed~\cite{pednault2017breaking,li2018quantum,chen2018classical} 
for geometrically local  shallow quantum  circuits that include only nearest-neighbor
gates on a 2D grid of qubits~\cite{boixo2018characterizing}.
These methods enabled simulations of systems with more than $100$ qubits~\cite{chen2018classical}.
However, it is expected~\cite{alibaba} that for general (geometrically non-local) circuits 
of size $poly(n)$  the runtime of tensor-based simulators scales as $2^{n-o(n)}$.

In contrast, Clifford simulators described in the present paper are applicable to large-scale circuits
without any locality properties as long as the circuit is dominated by Clifford gates. 
This regime may be important for verification of first fault-tolerant quantum circuits
where  logical non-Clifford gates are expected to be scarce due to their high implementation
cost~\cite{fowler2013surface,jones2013low}.
Another advantage of Clifford simulators is their ability to sample the output
distribution of the circuit (as opposed to computing individual output amplitudes).
This is more close to what one would expect from the actual quantum computer. 
For example, a single run of the heuristic sum-over-Cliffords simulator
described in Section~\ref{heuristic} produces thousands of samples from the (approximate) output distribution. 
In contrast, a single run of a tensor-based simulator typically computes a single amplitude of the
output state.  Thus we believe that our techniques 
extend the reach of classical simulation algorithms complementing
the existing vector- or tensor-based simulators.

A version of the sum-over-Cliffords simulator using the Metropolis sampling method is also publicly available
as part of \texttt{Qiskit-Aer}, the classical simulation framework of IBM's quantum
programming suite \texttt{Qiskit}~\cite{Qiskit}. This enables classical simulation and verification of quantum
circuits built in Qiskit on system sizes above $30$ qubits, which quickly become inaccessible with the
default vector-based method. This version also supports parallel processing over the stabilizer state decomposition,
which improves the performance of the Metropolis step.

Let us briefly comment on how simulators based on the stabilizer rank compare
with quasi-probability  methods~\cite{pashayan15,Delfosse15rebits,kocia2017discrete}.
The latter use a discrete Wigner function representation of quantum states
and Monte Carlo sampling
to approximate a given output probability of the target circuit with a small
additive error. Negativity of the Wigner function is an important parameter
that quantifies severity of the ``sign problem" associated with the Monte Carlo sampling.
The negativity also controls the runtime of quasi-probability methods. 
For example, the simulator proposed in~\cite{pashayan15}
has runtime $\epsilon^{-2} M^2$, where $M$ is the negativity and $\epsilon$
is the desired approximation error. In contrast to stabilizer rank simulators,
quasi-probability methods do not directly apply to stabilizer
operations on qubits since the latter are not known to have a non-negative Wigner function 
representation~\cite{Delfosse15rebits,karanjai2018contextuality}.
Furthermore, such methods are not well-suited for sampling the output
distribution since this task requires a small {\em multiplicative} error in 
approximating individual output probabilities. 

Our work leaves several  open questions. 
Since the efficiency of Clifford simulators hinges on the ability to find low-rank
stabilizer decompositions of multi-qubit magic states, 
improved techniques for finding such decompositions are of great interest. 
For example, consider a magic state $|\psi\ra=U|+\ra^{\otimes n}$, where 
$U$ is a diagonal circuit composed of $Z,CZ$, and $CCZ$ gates.
We anticipate that a low-rank exact stabilizer decomposition of $\psi$ can be 
found by computing the {\em transversal number}~\cite{alon1990transversal} of a suitable hypergraph describing
the placement of CCZ gates. Such low-rank decompositions may lead to more efficient
simulation algorithms for Clifford+CCZ circuits. We leave as an open question whether
the stabilizer extent $\xi(\psi)$ is multiplicative under tensor products for general states $\psi$.
Finally, it is of great interest to derive lower bounds on the stabilizer rank
of $n$-qubit magic states scaling exponentially with $n$.

\section{Subroutines}
\label{algorithms}

Throughout this section we use the following notations.
Suppose $x\in  \{0,1\}^n$ is a bit string.
We shall consider $x$  as a row vector  and write $x^T$ for the transposed column vector. 
The Hamming weight of $x$ denoted $|x|$ is the number of ones in $x$.
The support of $x$  is the subset of indices $j\in [n]$ such that $x_j=1$.
Given a single-qubit operator
$P$  let $P(x)$ be an $n$-qubit product operator
that applies $P$ to each qubit in the support of $x$, that is, 
$P(x)=P^{x_1}\otimes \cdots \otimes P^{x_n}$.
We shall use the notation $\oplus$ for the addition of binary vectors modulo two.
Let  $x\cdot y\equiv \sum_{j=1}^n x_jy_j$.

\subsection{Phase-sensitive Clifford simulator}
\label{clifford_sim}
\newcommand{\modtwo}{\; (\mathrm{mod}\, 2)}
In this section we describe a  Clifford simulator based on stabilizer tableau~\cite{aaronson04improved}  that keeps track of the global phase of stabilizer states. 
We shall consider Clifford circuits expressed using a gate set
\be
\label{CLgates}
S, \quad CZ, \quad CX,
\quad H.
\ee
Here $CZ$ and $CX$ are controlled-$Z$ and -$X$ gates, 
$H$ is the Hadamard gate, and $S=|0\ra\la 0|+i|1\ra\la 1|$.

First let us define a data format to describe stabilizer states. 
Suppose $U$ is a unitary Clifford operator.
We say that $U$ is a control-type or
{\em C-type} operator
if
\be
\label{Ctype} 
U|0^n\ra=|0^n\ra.
\ee  
For example, the gates $S,CZ,CX$ and any product of such gates are C-type operators.
We say that $U$ is a Hadamard-type or {\em H-type} operator if
$U$ is a tensor product of the Hadamard and the identity gates. 
Previously known results on canonical decompositions of Clifford circuits~\cite{nest2008classical,garcia2012efficient,maslov2017shorter}
imply that any $n$-qubit stabilizer state $\phi$ can be expressed as 
\be
\label{phi}
|\phi\ra = \omega U_C U_H  |s\ra,
\ee
where $U_C$ and $U_H$ are C-type and H-type Clifford operators,
$s\in \{0,1\}^n$ is a basis vector, and $\omega$ is a complex number. 
We shall refer to the decomposition Eq.~(\ref{phi}) as a CH-form of $\phi$.
Note that this form may be non-unique.

We shall describe the unitary $U_C$ by its stabilizer tableaux,
that is, a list of Pauli operators $U_C^{-1} Z_p U_C$ and $U_C^{-1} X_p U_C$.
The global phase of $U_C$ is fixed by Eq.~(\ref{Ctype}).
Using Eq.~(\ref{Ctype}) one can check that 
$U^{-1}_C Z_p U_C$ is a tensor product of Pauli $Z$ and the identity operators $I$.
Thus the stabilizer tableaux of $U_C$ can be described by binary matrices
$F,G,M$ of size $n\times n$ and a phase vector $\gamma\in \ZZ_4^n$ such that 
\be
\label{tableaux}
U^{-1}_C Z_p U_C = \prod_{j=1}^n Z_j^{G_{p,j}} \quad \mbox{and} \quad 
U^{-1}_C X_p U_C = i^{\gamma_p}  \prod_{j=1}^n X_j^{F_{p,j}} Z_j^{M_{p,j}}
\ee
for all $p=1,\ldots,n$. Here  $X^0\equiv Z^0\equiv I$. 
We shall describe the unitary $U_H$ by a string $v\in \{0,1\}^n$ such that 
\be
\label{Htype}
U_H=H(v)\equiv H_1^{v_1} \otimes H_2^{v_2} \otimes \cdots \otimes H_n^{v_n}.
\ee
To summarize,  the CH-form is fully specified by the data $(F,G,M,\gamma,v,s,\omega)$.
Let us agree that  $\omega=1$ whenever it is omitted.

Below we  describe an algorithm that takes as input a
sequence of Clifford gates $U_1,\ldots,U_m$ from the gate set Eq.~(\ref{CLgates})
and outputs the CH-form of a stabilizer state
\be
\label{phi1}
|\phi\ra=U_m\cdots U_2 U_1 |0^n\ra.
\ee
The runtime is $O(n)$ per each gate $S,CZ,CX$ and $O(n^2)$ per
each Hadamard gate. 
We also show how to compute an amplitude $\la x|\phi\ra$
and sample $x$ from the distribution $|\la x|\phi\ra|^2$
assuming that $\phi$ is specified by its CH-form. These tasks 
take time $O(n^2)$.
Finally, we consider projective gates $(I+P)/2$, where $P$ is a Pauli
operator. We show how to simulate projective gates in time $O(n^2)$.

{\bf Simulation of unitary gates.}
The initial state $|0^n\ra$ has a trivial CH-form with $s=0^n$ and $U_C=U_H=I$.
Thus we initialize the CH data as   $G=F=I$,  $M$ is the zero matrix,
and  $\gamma,v,s$ are zero vectors.
Suppose $\phi$ is a stabilizer state with the CH form  
\[
|\phi\ra=U_C U_H|s\ra
\]
described by the data $(F,G,M,\gamma,v,s)$.
Consider a gate 
$\Gamma \in \{ S,  CZ, CX, H\}$
applied to some subset of qubits. 
The state $\Gamma |\phi\ra$ has a CH-form
\be
\label{C'H'}
\Gamma |\phi\ra=\Gamma U_C U_H|s\ra = \omega' U_C' U_H'|s'\ra
\ee
with the corresponding  data $(F',G',M',\gamma',v',s',\omega')$.
Let us show how to compute this data. 

The case $\Gamma \in \{S, CZ, CX\}$ is trivial: one can absorb 
$\Gamma$ into the C-layer obtaining $U_C'=\Gamma U_C$. The  stabilizer tableaux of $U_C$ is updated using
the standard Aaronson-Gottesman algorithm~\cite{aaronson04improved} 
(explicit update rules are provided at the 
end of this section). This update takes time $O(n)$.

Let $\Gamma=H_p$ be the Hadamard gate applied to a qubit $p\in [n]$.
Commuting  $H_p$  through the C- and H-layer using the identity $H_p=2^{-1/2}(X_p+Z_p)$
and Eq.~(\ref{tableaux})
one gets 
\be
\label{H1}
H_p|\phi\ra = 2^{-1/2}U_C U_H [ (-1)^\alpha |t\ra + i^{\gamma_p} (-1)^\beta |u\ra),
\ee
where $t,u\in \{0,1\}^n$ are defined by
\be
\label{H2}
t_j = s_j \oplus G_{p,j} v_j 
\quad \mbox{and} \quad 
u_j = s_j  \oplus F_{p,j}\bar{v}_j  \oplus M_{p,j} v_j
\ee
for $j=1,\ldots,n$. Here and below $\bar{v}_j\equiv 1-v_j$. Furthermore,
\be
\label{H4}
\alpha = \sum_{j=1}^n G_{p,j} \bar{v}_j s_j 
\quad
\mbox{and} 
\quad
\beta=
\sum_{j=1}^n M_{p,j} \bar{v}_j s_j   + F_{p,j}v_j (M_{p,j}+s_j).
\ee
The case $t=u$ is trivial: Eq.~(\ref{H1}) gives the desired CH-form
of $H_p|\phi\ra$ with  $s'=t=u$
and $\omega'=2^{-1/2}[ (-1)^\alpha + i^{\gamma_p} (-1)^\beta]$.
From now on assume that $t\ne u$. 
\begin{prop}
\label{prop:H}
Suppose $t,u\in \{0,1\}^n$ are distinct strings
and $\delta \in \ZZ_4$.
Then  the state
$U_H(|t\ra +  i^\delta  |u\ra)$ has a CH-form 
\be
\label{H5}
U_H( |t\ra + i^\delta |u\ra) = \omega W_C W_H |s'\ra,
\ee
where the C-layer $W_C$ consists of $O(n)$ gates
from the set $\{S,CZ,CX\}$.  The decomposition Eq.~(\ref{H5})
can be computed  in time $O(n)$.
\end{prop}
Choosing $\delta=\gamma_p+2(\alpha+\beta) {\pmod 4}$ and
substituting
Eq.~(\ref{H5}) into Eq.~(\ref{H1}) one gets
\be
\label{H6}
H_p|\phi\ra = 2^{-1/2} (-1)^\alpha \omega \cdot U_C W_C  \cdot W_H |s'\ra.
\ee
This gives the desired CH-form of $H_p|\phi\ra$ with
\be
\label{H7}
\omega'=2^{-1/2}(-1)^\alpha \omega, \quad U_C'=U_C W_C,
\quad U_H'=W_H.
\ee
Finally, one needs to compute the  stabilizer tableaux of $U_C'$.
Since $W_C$ consists of $O(n)$ gates $S,CZ,CX$
it suffices to give update rules for the stabilizer tableaux of $U_C$
under the right multiplications $U_C \gets U_C \Gamma$
with $\Gamma \in \{S,CZ,CX\}$.
Explicit update rules are provided at the end of this section
(this  is a straightforward application of the stabilizer formalism).
Each update rule takes time $O(n)$. Since
$W_C$ contains  $O(n)$ gates,  the full simulation cost of
the Hadamard gate is $O(n^2)$.

\begin{proof}[Proof of Prosposition~\ref{prop:H}]
We shall construct a C-type circuit  $V_C$  and bit strings
$y,z\in \{0,1\}^n$  such that
\begin{itemize}
\item $y$ and $z$ differ on a single bit $q\in [n]$,
\item $U_H |t\ra =V_C U_H |y\ra$,
\item $U_H |u\ra =V_C U_H |z\ra$.
\end{itemize}
Then 
\be
U_H(|t\ra + i^\delta |u\ra) =V_C U_H(|y\ra + i^\delta |z\ra).
\ee
Since $y_i=z_i$ for $i\ne q$  and $y_q\ne z_q$,
 the state  
$U_H(|y\ra + i^\delta |z\ra)$
is a tensor product of single-qubit states
$H^{v_i}|y_i\ra$ on qubits $i\ne q$ and
a stabilizer state $H^{v_q}(|y_q\ra + i^\delta |z_q\ra)$ on qubit $q$.
Let us write 
\[
H^{v_q}(|y_q\ra + i^\delta |z_q\ra)=\omega S^a H^b |c\ra
\]
for some  $a,b,c\in \{0,1\}$ and some complex number $\omega$.
We arrive at 
\[
U_H(|t\ra + i^\delta |u\ra) =  
\omega(V_C S_q^a) (U_H H_q^{b\oplus v_q}) |s'\ra,
\]
where $s'_q=c$ and $s'_i=y_i=z_i$ for $i\ne q$.
This is the desired form Eq.~(\ref{H5})
with $W_C=V_C S_q^a$
and $W_H=U_H H_q^{b\oplus v_q}$.

It remains to construct $V_C,y,z$ as above. 
We shall choose $V_C$ such that 
\be
\label{HCCH}
U_H V_C U_H = \prod_{i\in [n]\setminus q\, : \, t_i\ne u_i}\; \; CX_{q,i}
\ee
for some qubit $q\in [n]$ such that $t_q\ne u_q$. 
The circuit in the righthand side of Eq.~(\ref{HCCH}) 
 maps $t,u$ to strings $y,z$ that differ only on the $q$-th bit.
Accordingly, $V_C U_H|t\ra=U_H|y\ra$ and $V_C U_H|u\ra=U_H|z\ra$,
as desired.  

For each $b\in \{0,1\}$ define a subset
\[
\calV_b=\{ i\in [n] \, : \, v_i=b \quad \mbox{and} \quad t_i \ne u_i\}.
\]
Here $v\in \{0,1\}^n$ defines the H-layer $U_H$, see Eq.~(\ref{Htype}).
By assumption,   at least one of the subsets $\calV_b$ is non-empty.

Suppose first that $\calV_0\ne \emptyset$. Let $q$ be the first qubit of $\calV_0$. Define
\[
V_C = \prod_{i\in \calV_0\setminus q } CX_{q,i} \cdot \prod_{i\in \calV_1 } CZ_{q,i}.
\]
Here $CX_{q,i}$ has control $q$ and target $i$. 
If $\calV_0=\{q\}$ then gates $CX_{q,i}$ are skipped.
Likewise, if $\calV_1=\emptyset$ then the gates $CZ_{q,i}$ are skipped.
Simple algebra shows that $V_C$  obeys Eq.~(\ref{HCCH}).

Suppose now that $\calV_0= \emptyset$. Then $\calV_1\ne \emptyset$
since $t\ne u$.
Let $q$ be the first qubit of $\calV_1$. Define
\[
V_C=\prod_{i\in \calV_1\setminus q } CX_{i,q}.
\]
Let us agree that $V_C=I$ if $\calV_1=\{q\}$.
Simple algebra shows that $V_C$ obeys Eq.~(\ref{HCCH}).

In both cases  the strings $y,z$ have the form 
\[
\mbox{if $t_q=1$ then $y=u\oplus e_q$ and $z=u$},
\]
\[
\mbox{if $t_q=0$ then $y=t$ and $z=t\oplus e_q$}.
\]
Here $e_q\in \{0,1\}^n$ is a string with a single non-zero at the $q$-th bit.
\end{proof}

In the rest of this section we provide  rules for updating the stabilizer
tableaux of $U_C$ under the left and the right multiplications
$U_C\gets \Gamma U_C$ and $U_C \gets U_C \Gamma$, where
$\Gamma$ is one of the gates $S,CZ,CX$. 
We shall write $\calL[\Gamma]$ and $\calR[\Gamma]$ for the left and the right
multiplication by $\Gamma$. 
Below $p=1,\ldots,n$.
All  phase vector updates are performed modulo four.
\[
\calR[S_q] \, : \, \left\{
\ba{rcl}
M_{p,q} &\gets& M_{p,q} \oplus F_{p,q}  \\
\gamma_p &\gets&  \gamma_p - F_{p,q} \\
\ea \right.
\quad  \quad
 \calL[S_q] \, : \, \left\{
\ba{rcl}
M_{q,p} &\gets& M_{q,p} \oplus G_{q,p}\\
\gamma_q &\gets&  \gamma_q - 1  \\
\ea \right.
\]
\[
\calR[CZ_{q,r}] \, : \, \left\{
\ba{rcl}
M_{p,q} &\gets& M_{p,q} \oplus F_{p,r}  \\
M_{p,r} &\gets& M_{p,r} \oplus F_{p,q}  \\
\gamma_p &\gets&  \gamma_p + 2F_{p,q}F_{p,r}  \\
\ea \right.
\quad \quad
\calL[CZ_{q,r}] \, : \, \left\{
\ba{rcl}
M_{q,p} &\gets& M_{q,p} \oplus G_{r,p}  \\
M_{r,p} &\gets& M_{r,p} \oplus G_{q,p}  \\
\ea \right.
\]
\[
 \calR[CX_{q,r}] \, : \, \left\{
\ba{rcl}
G_{p,q} &\gets& G_{p,q} \oplus G_{p,r}  \\
F_{p,r} &\gets& F_{p,r} \oplus F_{p,q}  \\
M_{p,q} &\gets& M_{p,q} \oplus M_{p,r}\\
\ea \right.
\quad
 \quad
\calL[CX_{q,r}] \, : \, \left\{
\ba{rcl}
G_{r,p} &\gets& G_{r,p} \oplus G_{q,p}  \\
F_{q,p} &\gets& F_{q,p} \oplus F_{r,p}  \\
M_{q,p} &\gets& M_{q,p} \oplus M_{r,p}\\
\gamma_q &\gets& \gamma_q+ \gamma_r + 2(MF^T)_{q,r} \\
\ea \right.
\]

\noindent
{\bf Simulating measurements. }
Let $x\in \{0,1\}^n$ be a basis vector. Using Eqs.~(\ref{Ctype},\ref{tableaux}) one gets
\be
\label{amplitude}
\la x|U_C U_H|s\ra =\la 0^n| \left( \prod_{p=1}^n U_C^{-1} X_p^{x_p} U_C \right) U_H|s\ra
\equiv \la 0^n |Q U_H|s\ra.
\ee
Note that $Q$ is a product of $|x|$ Pauli operators that appear in Eq.~(\ref{tableaux}).
It can be computed inductively in time $O(n^2)$ by setting 
$Q=I$ and performing updates $Q\gets Q \cdot U_C^{-1} X_p^{x_p} U_C$ 
for each $p$ with $x_p=1$.  
Write $Q=i^\mu Z(t)X(u)$ for some $\mu\in \ZZ_4$ and
$t,u\in \{0,1\}^n$. Note that   $u=xF\modtwo$. Then
\be
\label{amplitude1}
\la x|U_C U_H|s\ra = 
\la 0^n|QU_H|s\ra=
2^{-|v|/2} i^\mu \prod_{j\, : \, v_j=1} (-1)^{u_j s_j}  \prod_{j\, : \, v_j=0} \la u_j|s_j\ra.
\ee
Thus computing the amplitude $\la x|U_C U_H|s\ra$ takes time $O(n^2)$.

Consider a probability distribution $P(x)=|\la x|U_CU_H|s\ra|^2$.
From Eq.~(\ref{amplitude1}) one infers that 
$P(x)=2^{-|v|}$ if $u_j=s_j$ for all bits $j$ with $v_j=0$
and $P(x)=0$ otherwise. 
Since $U_C$ preserves the Pauli commutation rules, one has 
$FG^T=I \modtwo$.  Thus
$x=wG^T\modtwo$, where $w\in \{0,1\}^n$ is a row vector satisfying $w_j=s_j$ if $v_j=0$.
The remaining bits of $w$ are picked uniformly at random.  
Thus one can sample $x$ from $P(x)$ as follows:
\begin{itemize}
\item Set $w=s$.
\item For each $j$ such that $v_j=1$ flip the $j$-th bit of $w$ with probability $1/2$.
\item Output $x=wG^T \modtwo$.
\end{itemize}
This takes time $O(n^2)$.
Finally, consider a projective gate $\Gamma=(I+P)/2$, where $P=P^\dag$ is a Pauli operator.
We have 
\[
\Gamma |\phi\ra=\Gamma U_C U_H |s\ra = (1/2)U_C U_H(I+ Q)|s\ra,
\]
where $Q$ is a Pauli operator that can be computed in time  $O(n^2)$ using the
stabilizer tableaux of $U_C$. 
Write $(I+Q)|s\ra=|s\ra + i^\delta |t\ra$ for some $t\in \{0,1\}^n$ and $\delta\in\ZZ_4$.
We can now compute the CH-form of $\Gamma|\phi\ra$ using Proposition~\ref{prop:H}
in the same fashion as was done above for the Hadamard gate.

\subsection{Heuristic Metropolis simulator}
\label{heuristic}Consider a state
$|\psi\ra=\sum_{\alpha=1}^k b_\alpha |\phi_\alpha\ra$,
where $\phi_1,\ldots,\phi_k$ are $n$-qubit stabilizer states.
We assume that all states $\phi_\alpha$ are specified by their CH-form.
This form can be efficiently computed using 
the Clifford simulator of Section~\ref{clifford_sim}.
Our goal is to sample $x\in \{0,1\}^n$ from the probability distribution
\[
P(x)=\frac{|\la x|\psi\ra|^2}{\|\psi\|^2}.
\]
To this end define a Metropolis-type Markov chain $\calM$ with a state space
$\Omega=\{ x\in \{0,1\}^n\, : \, P(x)>0\}$.
 Suppose the current state of the chain $x\in \Omega$.
Then the next state $x'$ is generated as follows.\\
\begin{itemize}
\item Pick an integer $j\in [n]$ uniformly at random and let $y=x\oplus e_j$.
\item If $P(y)\ge P(x)$ then set $x'=y$.
\item Otherwise generate a random bit $b\in \{0,1\}$ such that $\mathrm{Pr}(b=1)=P(y)/P(x)$.
\item If $b=1$ then set $x'=y$. Otherwise set $x'=x$.
\end{itemize}
We shall refer to the mapping $x\to x'$ as a Metropolis step.
Let us make a simplifying assumption that the chain $\calM$ is irreducible, that is,
for any pair of strings $x,y\in \Omega$  there exist a path
$x^0=x,x^1,\ldots,x^L=y\in \Omega$ such that $x^i$ and $x^{i+1}$ differ on 
a single bit for all $i$. Then 
$P(x)$ is the unique steady distribution of $\calM$.
One can (approximately) sample $x$ from $P(x)$ by 
implementing $T\gg 1$ Metropolis steps 
starting from some (random) initial state $x_{in}\in \Omega$
and using the final state as the output string. 

We claim that one can implement $T$ Metropolis  steps  in time 
\[
O(k n T) + O(k n^2). 
\]
Here the term $O(kn^2)$ is the cost of computing
the initial probability $P(x_{in})$ using the algorithm of Section~\ref{clifford_sim}.
Indeed, suppose we have already implemented
several steps  reaching some state $x\in \Omega$.
Let $y=x\oplus e_j$ be a proposed next state. 
Consider some fixed stabilizer state $\phi\equiv \phi_\alpha$ 
that contributes to $\psi$
and let $|\phi\ra=U_C U_H |s\ra$ be its CH-form.
Then
\be
\la y|\phi\ra=
\la x\oplus e_j|U_CU_H|s\ra = \la 0^n| U_C^{-1} X_j U_C \cdot Q_x U_H|s\ra,
\ee
where 
\[
Q_x\equiv \prod_{p=1}^n U_C^{-1} X_p^{x_p} U_C.
\]
Note that computing $Q_x$ for the initial state 
$x=x_{in}$ and $\alpha=1,\ldots,k$ takes time
$O(kn^2)$. 
Suppose $Q_x$ has been already computed.
 Since $U_C^{-1} X_j U_C$
is determined by the stabilizer tableaux of $U_C$, see Eq.~(\ref{tableaux}),
one can compute the product $Q_y= U_C^{-1} X_j U_C \cdot Q_x$ in time $O(n)$. 
Then the amplitude $\la y|\phi\ra=\la 0^n|Q_yU_H|s\ra$ can be
computed in time $O(n)$. This shows that the ratio
\[
\frac{P(y)}{P(x)}=\left| \frac{\sum_{\alpha=1}^k b_\alpha \la y|\phi_\alpha\ra}
{\sum_{\alpha=1}^k b_\alpha \la x|\phi_\alpha\ra} \right|^2.
\]
can be computed in time $O(k n)$
provided that one saves the Pauli $Q_x$
for each stabilizer term $\phi_\alpha$ after each Metropolis step.
This  achieves the runtime scaling quoted above.

In general there is no reason to expect that the Metropolis chain defined above is
irreducible. Furthermore, its mixing time is generally unknown. 
Thus the proposed algorithm should be considered as a heuristic. However, the numeric
results shown in Fig.~\ref{fig:QAOA1} were obtained using the Metropolis
method to sample from the output distribution of the QAOA circuit.

We expect that the Metropolis chain  may be rapidly mixing 
in the case when $\psi$ approximates the output state of some small-depth quantum circuit.
In particular, if $P(x)$ is the exact output distribution of a constant-depth circuit
and each Metropolis step flips $O(1)$ bits, one can use isoperimetric inequalities
derived in Refs.~\cite{eldar2017local,crosson2017quantum} to show that 
$P(x)$ is the uniqiue steady state of $\calM$
and its mixing time is at most $poly(n)$.

\subsection{Fast norm estimation}
\label{Sec_fast_norm}
As before, consider a state
$|\psi\ra=\sum_{\alpha=1}^k b_\alpha |\phi_\alpha\ra$,
where $\phi_1,\ldots,\phi_k$ are $n$-qubit stabilizer states
specified by their CH-form.
Recall that our goal is to 
estimate the norm $\|\psi\|^2$ and to sample the probability
distribution $P(x)\sim |\la x|\psi\ra|^2$.
In this section 
we describe an algorithm that takes 
as input the target state $\psi$,
error tolerance parameters $\epsilon,\delta>0$, and
outputs a random number $\eta$ such that 
\be
\label{approx}
(1-\epsilon)\| \psi\|^2 \le \eta \le (1+\epsilon)\|\psi\|^2 
\ee
with probability at least $1-\delta$. The algorithm has runtime  
\be
\label{runtime}
O(k n^3 \epsilon^{-2} \log{\delta^{-1}}).
\ee
The key  idea  proposed in Ref.~\cite{bravyi2016improved}
is to estimate $\|\psi\|^2$ by computing inner products
between $\psi$ and randomly chosen stabilizer states $\phi$.
It can be shown~\cite{bravyi2016improved} that
the quantity $\eta \equiv 2^n |\la \phi|\psi\ra|^2$ is an unbiased estimator of 
$\|\psi\|^2$ with the standard deviation  $\approx \|\psi\|^2$,
provided that $\phi$ is drawn from the uniform distribution on the set of
stabilizer states. 
Thus the empirical mean of $\eta$ provides an estimate of $\|\psi\|^2$ with a small multiplicative error. 
The quantity $\eta$ can be computed in time $O(k n^3)$ since
$\la \phi|\psi\ra=\sum_{\alpha=1}^k b_\alpha \la \phi|\phi_\alpha\ra$
and the inner product between stabilizer states can be computed in time $O(n^3)$.

Here we improve upon the algorithm of Ref.~\cite{bravyi2016improved} 
in two respects. First, we show that 
the random stabilizer state $\phi$ used in the norm estimation method
can be drawn from a certain subset of stabilizer states that we call equatorial states. 
By definition, a stabilizer state $\phi$ is called equatorial iff it has equal amplitude
on each basis vector.  Sampling an equatorial state from the uniform distribution 
is particularly simple: all it takes is  tossing an unbiased coin $O(n^2)$ times. 
Secondly, we greatly simplify computation of the inner products $\la \phi |\phi_\alpha\ra$.
This is achieved by using the CH-form to describe stabilizer states
and by introducing a more efficient (and simpler) algorithm for computing
certain exponential sums (see Lemma~\ref{lemma:expo_sum} below).
 
We shall now formally describe the norm estimation algorithm.
Let $\calM_n$ be the set of symmetric $n\times n$ matrices $M$ 
with off-diagonal entries $\in \{0,1\}$ and diagonal 
entries $\in \{0,1,2,3\}$.
For any matrix $A\in \calM_n$ define a stabilizer state 
\begin{equation}
\label{eqa1}
|\phi_A\rangle =2^{-n/2} \sum_{x\in \{0,1\}^n}\; i^{x A x^T} |x\rangle.
\end{equation}
We shall refer to $\phi_A$ 
as an {\em equatorial state} (note that  $\phi_A$  lies
on the equator of the Bloch sphere for $n=1$). 
\begin{lemma}[\bf Norm Estimation]
\label{lemma:norm}
Let $\psi$ be an arbitrary $n$-qubit state.
Define a random variable
\begin{equation}
\label{xi}
\eta_A = 2^n |\langle \phi_A|\psi\rangle|^2,
\end{equation}
where $A\in \calM_n$ is chosen uniformly at random.
Then $\eta_A$ has mean  $\|\psi\|^2$
and its variance is at most $\|\psi\|^4$.
\end{lemma}

\begin{lemma}[\bf Inner Product]
\label{lemma:inner}
Suppose $|\phi\ra=U_CU_H|s\ra$ is a stabilizer state 
in the CH-form, where $U_H=H(v)$
and $U_C$ has a stabilizer tableaux $(F,G,M,\gamma)$.
 Suppose $\phi_A$ is an equatorial state
specified by a matrix  $A\in \calM_n$.  Define a matrix
$J\in \calM_n$ such that 
$\mathrm{diag}(J)=\gamma$ and $J_{a,b}=(MF^T)_{a,b} \modtwo$
for $a\ne b$. Define a matrix
\[
K=G^T(A+J)G.
\]
Then 
\be
\label{inner_explicit}
\la \phi | \phi_A\ra=2^{-(n+|v|)/2}  \cdot i^{sKs^T}\cdot (-1)^{s\cdot v} 
\sum_{x\le v} i^{xKx^T + 2x(s+sK)^T}.
\ee
Here the sum is over  $n$-bit strings $x$ satisfying
$x_j\le v_j$ for all $j$.
\end{lemma}
Since the Pauli operators $U_C^{-1} X_p U_C$
pairwise commute, $MF^T \modtwo$ is a symmetric matrix,
see Eq.~(\ref{tableaux}).
 Therefore $K$ is a symmetric
matrix and thus 
$i^{xKx^T}$  depends only
on off-diagonal elements of $K$ modulo two and diagonal elements of $K$
modulo four. Thus
the sum that appears in Eq.~(\ref{inner_explicit}) can be expressed as 
\[
\calZ(B)=\sum_{x\in \{0,1\}^{|v|}} i^{x Bx^T}
\]
for a suitable matrix $B\in \calM_{|v|}$,
namely, a restriction of the matrix $K+2\mathrm{diag}(s+sK)$
onto the subset of rows and columns $j$ with $v_j=1$.
We shall refer to $\calZ(B)$ as an exponential sum
associated with $B$. 
\begin{lemma}[\bf Exponential Sum]
\label{lemma:expo_sum}
There is a deterministic algorithm with a runtime 
$O(n^3)$ that takes as input a matrix $B\in \calM_n$ and
outputs integers $p,q\ge 0$ and $\alpha,\beta\in \{0,1\}$
such that $\calZ(B)=\alpha 2^p + i\beta2^q$.
\end{lemma}

The desired estimate of $\|\psi\|^2$ can now be obtained by sampling
i.i.d. random matrices $A_1,\ldots,A_L\in \calM_n$ and computing the empirical mean
$\eta= L^{-1}(\eta_{A_1}+\ldots +\eta_{A_L})$. Indeed, 
Lemma~\ref{lemma:norm} and the Chebyshev inequality imply that
$\eta$ achieves the desired approximation Eq.~(\ref{approx})
with probability at least $3/4$ if $L=4\epsilon^{-2}$. 
The error probability can be reduced to any desired level $\delta$
by generating $K=O(\log{\delta^{-1}})$ independent estimates $\eta^1,\ldots,\eta^K$
as above such that each estimate $\eta^a$ 
satisfies  Eq.~(\ref{approx}) with probability at least $3/4$.
Let $\eta_{med}$ be the median of $\eta^1,\ldots,\eta^K$.
Then standard arguments show that $\eta_{med}$ satisfies Eq.~(\ref{approx}) with probability at least $1-\delta$.
Computing each sample $\eta_{A_i}$ using Lemmas~\ref{lemma:inner},\ref{lemma:expo_sum} takes time $O(kn^3)$.
Since the total number of samples 
is $KL=O(\epsilon^{-2}\log{\delta^{-1}})$, we arrive at Eq.~(\ref{runtime}).

Finally, let us sketch how to use the norm estimation for sampling 
$x\in \{0,1\}^n$ from  a distribution 
\[
P(x)=\frac{|\la x|\psi\ra|^2}{\|\psi\|^2}.
\] 
Let $P_w(x_1,\ldots,x_w)$ be the marginal distribution  describing the first $w$
bits. We have $P_w(x)=\| \Pi \psi\|^2/\|\psi\|^2$, where
$\Pi$ projects the $j$-th qubit onto the state $x_j$ for $1\le j\le w$.
It can be written as 
\[
\Pi=2^{-w} \prod_{j=1}^w (I+(-1)^{x_j}Z_j)
\]
One can compute a rank-$k$ stabilizer decomposition of 
the state $\Pi |\psi\ra$ in time $O(kw n^2)$ using the Clifford simulator
of Section~\ref{clifford_sim}. 
By estimating the norms $\|\psi\|^2$ and $\|\Pi \psi\|^2$
one can approximate any
marginal probability $P_w(x)$ with a small multiplicative error.
In the same fashion one can approximate conditional probabilities
\[
P_w(x_w|x_1,\ldots,x_{w-1})=\frac{P_w(x_1,\ldots,x_w)}{P_{w-1}(x_1,\ldots,x_{w-1})}.
\]
Now one can sample the bits of $x$ one by one using the chain rule
\[
P(x)=P_1(x_1)P_2(x_2|x_1)\cdots P_n(x_n|x_1,\ldots,x_{n-1}).
\]
Clearly, the same method can be used to sample any marginal distribution of $P(x)$.

To avoid accumulation of errors, each of $O(n)$ steps in the chain rule requires an estimate
of the marginal probabilities  $P_w(x)$ with a multiplicative error $O(n^{-1})$.
(This guarantees that the full probability $P(x)$ is estimated using the chain rule
within a small multiplicative error.) 
This would require setting the precision $\epsilon$ in the norm estimation method as
$\epsilon=O(n^{-1})$. Thus the cost of each norm estimation would be $O(k n^3 \epsilon^{-2}) =O(k n^5)$.
Since the total number of norm estimations is $\Omega(n)$, the overall
runtime for generating a single sample from $P(x)$ with a small error would scale as
$O(k n^6)$. This quickly becomes impractical. 
However,  if our goal is to sample only $w$ bits from $P(x)$, a similar analysis
shows that the overall runtime scales as $O(k n^3 w^3)$.
Thus the sampling method based on the norm estimation is practical only for small values of $w$.
In contrast, Metropolis simulator allows one to sample all $n$ output bits
and has runtime  $O(k nT)$, where $T$ is the mixing time (which is generally unknown).

In the rest of this section we prove Lemmas~\ref{lemma:norm},\ref{lemma:inner},\ref{lemma:expo_sum}.

\begin{proof}[Proof of Lemma~\ref{lemma:norm}]
Let 
\[
Q_1=\EX_A |\phi_A\ra \la \phi_A|
\quad \mbox{and} \quad 
Q_2=\EX_A |\phi_A\ra \la \phi_A|^{\otimes 2}.
\]
Since the distribution of $A$ is invariant under shifts 
$A_{j,j}\gets A_{j,j}+2$, one concludes that $Q_1$ commutes with single-qubit Pauli-$Z$ operators.
Thus $Q_1$ is diagonal in the $Z$-basis. Furthermore, all diagonal matrix elements  of 
$|\phi_A\ra\la \phi_A|$ are equal to $2^{-n}$. This proves
$Q_1=2^{-n} I$ and thus  $\eta_A$ has expected value $2^n \la \psi|Q_1|\psi\ra=\|\psi\|^2$.

By definition,
\[
Q_2=4^{-n} \sum_{w,x,y,z} E(w,x,y,z) \cdot |w,x\rangle\langle y,z|\quad \mbox{where} \quad
E(w,x,y,z)=\EX_A \, i^{wAw^T + xAx^T - yAy^T -zAz^T}.
\]
Here the sum runs over all $n$-bit strings. 
We shall use the following fact.
\begin{prop}[\bf Ref.~\cite{bremner2016average}]
$E(w,x,y,z)=0$ unless 
$w+x=y+z {\pmod 4}$ and at least two of the strings $w,x,y$ coincide. 
\end{prop}
\begin{proof}
By definition, diagonal entries $A_{p,p}\in \ZZ_4$ 
and off-diagonal entries $A_{p,q}=A_{q,p}\in \ZZ_2$
are i.i.d. uniform random variables.
The entry $A_{p,p}$ contributes a factor $i^{A_{p,p}(w_p+x_p-y_p-z_p)}$ to $E(w,x,y,z)$.
Thus  $E(w,x,y,z)=0$ unless
\be
\label{E0part1}
w_p+x_p=y_p+z_p {\pmod 4}
\ee
for all $p$.  This proves the first claim.
The entry $A_{p,q}=A_{q,p}$
contributes a factor 
\begin{equation}
	(-1)^{A_{p,q}(w_p w_q + x_p x_q - y_p y_q - z_p z_q)} \nonumber
\end{equation}	
to $E(w,x,y,z)$. Thus  $E(w,x,y,z)=0$ unless
\be
\label{E0part2}
w_p w_q + x_p x_q - y_p y_q - z_p z_q = 0 {\pmod 2}.
\ee
From Eq.~(\ref{E0part1}) one gets $z_p = w_p+x_p+y_p {\pmod 2}$.
Substituting this expression for $z_p$ into Eq.~(\ref{E0part2}) one
 concludes
that $E(w,x,y,z)=0$ unless
\be
\label{E0part3}
(w_p x_q + w_q x_p) + (x_p y_q + x_q y_p) + (y_p w_q + y_q w_p) = 0 {\pmod 2}
\ee 
for all $p<q$. If $w=x=y$ then there remains  nothing to prove. 
Otherwise, there exists an index $p\in [n]$ such that 
exactly two of the variables $w_p,x_p,y_p$ coincide.
Since Eq.~(\ref{E0part3}) 
is symmetric under permutations of $w,x,y$, 
assume wlog that $x_p=y_p\ne w_p$.
Consider two cases.\\
{\em Case~1:} $x_p=y_p=0$ and $w_p=1$. Substituting this into Eq.~(\ref{E0part3}) one gets
$y_q=x_q$ for all $q\ne p$. Thus $x=y$.  \\
{\em Case~2:} $x_p=y_p=1$ and $w_p=0$. Substituting this into Eq.~(\ref{E0part3}) one gets
$y_q+x_q+w_q+w_q=0{\pmod 2}$ for all $q\ne p$, that is, $x=y$.\\
We conclude that at least two of the strings $w,x,y$ coincide.
\end{proof}
Let us consider the cases when $E(w,x,y,z)\ne 0$.
 {\em Case~1:} $w=x$.
Then $y+z=2x {\pmod 4}$ which is possible only if $y=z$
and thus $w=x=y=z$.
{\em Case~2:} $w=y$. Then $x=z$ and $E(y,x,y,x)=1$.
{\em Case~3:} $w=z$. Then $x=y$ and $E(z,x,x,z)=1$.
The above shows that 
non-zero contributions to $Q_2$ come only from the terms
$E(w,x,w,x)=E(w,x,x,w)=1$.  Thus 
\[
Q_2=4^{-n} (I+\mathrm{SWAP}) - 4^n \sum_x |x,x\rangle\langle x,x|,
\]
Here the last term is introduced to avoid overcounting since the terms with
$w=x=y=z$ appear in all three cases. We arrive at
\[
\EX_A (\eta_A^2) = 4^n \langle \psi^{\otimes 2} |Q_2|\psi^{\otimes 2}\rangle 
\le  \langle \psi^{\otimes 2} |I+\mathrm{SWAP}|\psi^{\otimes 2}\rangle
=2\|\psi\|^4.
\]
It follows that $\eta_A$ has variance at most $\|\psi\|^4$.
\end{proof}
\begin{proof}[Proof of Lemma~\ref{lemma:inner}]
Let $a\in \{0,1\}^n$ be an arbitrary string. From Eq.~(\ref{tableaux}) one easily gets
\[
U_C^{-1} X(a) U_C = \prod_{p=1}^n U_C^{-1} X_p^{a_p} U_C 
=i^{aJa^T} \cdot X(aF \modtwo )Z(aM \modtwo).
\]
Here $J\in \calM_n$ is defined in the statement of the lemma.
It follows that 
\[
U_C^{-1} |a\ra =U_C^{-1} X(a) U_C |0^n\ra = i^{aJa^T} |aF \modtwo\ra.
\] 
Therefore 
\[
U_C^{-1}|\phi_A\ra = 2^{-n/2} \sum_{x\in \{0,1\}^n} i^{x(A+J)x^T} |xF \modtwo\ra.
\]
Recall that $F G^T \modtwo=I$. Perform a change of variable $x=yG^T \modtwo$.
Then $x=yG^T + 2u$ for some integer vector $u$.
Using the fact that $A$ and $J$ are symmetric matrices one gets
\[
x(A+J)x^T=yG^T (A+J)G y^T + 4u(A+J)Gy^T + 4u(A+J)u^T.
\]
Denoting $K=G^T (A+J)G$ one gets
\be
\label{half1}
U_C^{-1}|\phi_A\ra 
=2^{-n/2} \sum_{y\in \{0,1\}^n} i^{y K y^T} |y\ra.
\ee
We have
\be
\label{half2}
U_H|s\ra = 2^{-|v|/2} \sum_{x\le v} (-1)^{s\cdot v + s\cdot x} |s\oplus x\ra,
\ee
Taking the inner product
of the states Eqs.~(\ref{half1},\ref{half2}) gives
\be
\label{aux_inner}
\la \phi |\phi_A\ra=\la s| U_H U_C^{-1} |\phi_A\ra = 2^{-(n+|v|)/2} (-1)^{s\cdot v} 
\sum_{x\le v} (-1)^{s\cdot x}\cdot  i^{ (s\oplus x)K (s\oplus x)^T}.
\ee
Writing $s\oplus x=s+x+2u$ for some integer vector $u$
and using the fact that $K$ is symmetric one gets
\[
(s\oplus x)K (s\oplus x)^T = (s+x)K(s+x)^T + 4uK(s+x)^T + 4uKu^T.
\]
It follows that 
\[
i^{(s\oplus x)K (s\oplus x)^T}=i^{sKs^T + xKx^T + 2xKs^T}.
\]
Combining this and Eq.~(\ref{aux_inner}) proves Eq.~(\ref{inner_explicit}).
\end{proof}

\begin{proof}[Proof of Lemma~\ref{lemma:expo_sum}]
Define a binary upper-triangular matrix $M$ of size $n\times n$
such that 
$M_{\alpha,\beta}=B_{\alpha,\beta}$  for $\alpha<\beta$.
Define binary vectors $L,K\in \{0,1\}^n$ such that 
$B_{\alpha,\alpha}=2L_\alpha + K_\alpha$ for all $\alpha$.
Then  $i^{xBx^T}=i^{q(x)}$, where $q\, : \, \{0,1\}^n \to \ZZ_4$
is a binary quadratic form  defined as 
\begin{equation}
\label{q}
q(x)=2\sum_{1\le \alpha< \beta \le n} M_{\alpha,\beta} x_\alpha x_\beta  + \sum_{1\le \alpha\le n} (2L_\alpha + K_\alpha) x_\alpha
{\pmod 4}.
\end{equation}
Our goal is to compute the exponential sum
\begin{equation}
\label{Z}
\calZ \equiv \sum_{x\in \{0,1\}^n} i^{q(x)}.
\end{equation}
The first observation is that exponential sums associated with $\ZZ_2$-valued quadratic forms
can be computed recursively.  Indeed, assume that $K_\alpha=0$ for all $\alpha$. 
Then 
\begin{equation}
\label{Zreal}
\calZ=\sum_{x\in \{0,1\}^n}  (-1)^{Q(x)} \quad  \mbox{where} \quad Q(x)=xM x^T + L x^T   {\pmod 2}.
\end{equation}
It will be convenient to consider more general quadratic forms $Q(x)$ as in Eq.~(\ref{Zreal})
where $M$ is an arbitrary binary matrix. We allow $M$ to be  non-symmetric
and  have non-zero diagonal. 

Consider first the trivial case when $M$ is a symmetric matrix. 
In this case all quadratic terms in $Q(x)$ cancel each other, that is, $Q(x)$ is linear.
Thus  $\calZ=2^n$ if $L=\mathrm{diag}(M)$ and $\calZ=0$ otherwise.

Suppose now that  $M$ is non-symmetric. 
We can assume wlog that $M_{1,2}\ne M_{2,1}$ (otherwise permute the variables). 
Then $M_{1,2}+M_{2,1}=1 {\pmod 2}$.
Write $x=(x_1,x_2,y)$ with $y\in \{0,1\}^{n-2}$. Define a partial sum 
\begin{equation}
\label{eq1partialsum}
\calZ(y) = \sum_{x_1,x_2\in \{0,1\}}\; (-1)^{Q(x_1,x_2,y)}=
\sum_{x_1,x_2\in \{0,1\}}\; (-1)^{x_1 x_2 + \mu_1(y)  x_1 + \mu_2(y) x_2 + Q_{else}(y)},
\end{equation}
where $Q_{else}(y)$ includes all terms in $Q(x)$ that do not depend on $x_1,x_2$, 
\[
\mu_1(y) = L_1  + M_{1,1} +  \sum_{3\le \alpha \le n} (M_{1, \alpha} + M_{\alpha,1}) y_\alpha
\equiv L_1 + M_{1,1} + m_1  y^T,
\]
\[
\mu_2(y)= L_2 +M_{2,2} +  \sum_{3\le \alpha \le n} (M_{2, \alpha}  + M_{\alpha,2}) y_\alpha 
\equiv L_2 + M_{2,2} +  m_2 y^T. 
\]
Here $m_1,m_2$ are row vectors of length $n-2$. 
A simple algebra shows that 
\begin{equation}
\label{identity}
\sum_{x_1,x_2\in \{0,1\}}\;  (-1)^{x_1 x_2 + \mu_1  x_1 + \mu_2 x_2}= 2(-1)^{\mu_1 \mu_2}
\qquad \mbox{for all $\mu_1,\mu_2\in \{0,1\}$}.
\end{equation}
Substituting this identity into Eq.~(\ref{eq1partialsum}) gives
\begin{equation}
\calZ=\sum_{y\in \{0,1\}^{n-2}}\;  \calZ(y)=2 (-1)^{(L_1+M_{1,1})( L_2 + M_{2,2}) } \sum_{y\in \{0,1\}^{n-2}}\; (-1)^{Q'(y)},
\end{equation}
where $Q'(y)$ is a quadratic form that depends on  $n-2$ variables:
\begin{equation}
\label{eq5}
Q'(y)=y( M_{else} + m_1^T m_2 ) y^T + (L_{else} +  [L_1 + M_{1,1}] m_2 + [L_2 + M_{2,2}] m_1 )y^T
\end{equation}
The matrix $M_{else}$ and the vector $L_{else}$ are determined by
$Q_{else}(y)=yM_{else} y^T + L_{else} y^T$.
We have reduced the exponential sum problem with $n$ variables
to the one with $n-2$ variables. 
Clearly, the coefficients of $Q'(y)$ can be  computed 
in time $O(n^2)$. The overall runtime is $\sum_{k=1}^n O(k^2)=O(n^3)$.
This gives an algorithm for computing the exponential sum for a $\ZZ_2$-valued
quadratic form. 

{\em Remark:} The most time-consuming step is getting the matrix
$M_{else} + m_1^T m_2$. Since the arithmetics is mod-2, this amounts to flipping all
bits of $M_{else}$ in a submatrix formed by rows $i\in m_1$ and by columns $j\in m_2$.

Consider now a $\ZZ_4$-valued form $q(x)$ defined in Eq.~(\ref{q}).
Define a $\ZZ_2$-valued form 
\begin{equation}
\label{Q}
Q(x)=\sum_{1\le \alpha<\beta \le n} (M_{\alpha,\beta} + K_\alpha K_\beta) x_\alpha x_\beta
+ \sum_{1\le \alpha \le n} K_\alpha x_\alpha x_{n+1} + \sum_{1\le \alpha\le n} L_\alpha x_\alpha {\pmod 2}.
\end{equation}
\begin{prop}
Let $\calZ$ be the exponential sum defined by Eqs.~(\ref{q},\ref{Z}). Then 
\begin{equation}
\label{real}
\mathrm{Re}(\calZ)=\frac12 \sum_{x\in \{0,1\}^{n+1}} \; (-1)^{Q(x)}
\quad \mbox{and} \quad
\mathrm{Im}(\calZ)=\frac12 \sum_{x\in \{0,1\}^{n+1}} \; (-1)^{Q(x)+x_{n+1}}.
\end{equation}
\end{prop}
\begin{proof}
Write $q(x)=2r(x) + Kx^T  {\pmod 4}$, where $r(x)$ is a $\ZZ_2$-valued quadratic form.
Consider some $x\in \{0,1\}^n$ and let $\omega \equiv K x^T {\pmod 2}$.
One can easily check that
\[
i^{Kx^T } = (-1)^{\sum_{1\le \alpha<\beta\le n}\; K_\alpha K_\beta x_\alpha x_\beta} \cdot i^\omega.
\]
By definition $\omega\in \{0,1\}$ so that
\[
\mathrm{Re}(i^\omega)=\frac12 (1+(-1)^\omega) \quad \mbox{and} \quad
\mathrm{Im}(i^\omega)=\frac12 (1-(-1)^\omega).
\]
Define a $\ZZ_2$-valued form $Q'(x)=r(x) + \sum_{1\le \alpha<\beta\le n}\; K_\alpha K_\beta x_\alpha x_\beta$.
Then 
\[
\mathrm{Re}(i^{q(x)})=\frac12\left[  (-1)^{Q'(x)} + (-1)^{Q'(x) + K x^T}\right] \quad
\mbox{and} \quad
\mathrm{Im}(i^{q(x)})=\frac12\left[  (-1)^{Q'(x)} - (-1)^{Q'(x) + K x^T}\right].
\]
Finally, add an extra variable $x_{n+1}$ such that the two terms in the square brackets 
correspond to $x_{n+1}=0$ and $x_{n+1}=1$ respectively. We arrive at Eq.~(\ref{real})
with $Q(x,x_{n+1})=Q'(x) +x_{n+1}(K x^T)$.
\end{proof}
\end{proof}

{\em Remark:} Computing   exponential sums
associated with the real and imaginary parts
of $\calZ$ takes about the same time as computing a single exponential sum
Eq.~(\ref{Zreal}) because the forms $Q(x)$ and $Q(x)+x_{n+1}$
in Lemma~2 have the same quadratic parts. 

Numerics shows that the new algorithm 
for computing exponential sums
achieves a significant speedup as is shown in Table~\ref{table:1}. Altogether, the use of the phase-sensitive Clifford simulator, sampling with equatorial states, and the improved Exponential Sum routine lead to a significant performance increase in simulations. In Table~\ref{table:HStimings}, we compare the performance of the simulator in Ref.~\cite{bravyi2016improved} and this paper, when estimating the output probabilities of the Hidden Shift problem on $40$-qubits with the Sum-over-Cliffords method (see also Sections~\ref{sec:simulations} and \ref{sec:numericresults}).

\begin{table}[!h]
\centerline{
\begin{tabular}{r|c|c|c|c|c|c}
Number of variables $n$ & $\bf 10$ & $\bf 20$ & $\bf 30$ & $\bf 40$  & $\bf 50$ & $\bf 60$  \\
\hline
New runtime & $0.016$ & $0.017$ & $0.021$ & $0.023$ & $0.030$ & $0.036$  \\
\hline
BG16 runtime & $0.42$ & $0.50$ & $0.77$ & $1.10$ & $1.40$ & $1.72$  \\
\end{tabular}}
\caption{Average runtime in milliseconds  of
the new algorithm for computing exponential sums
and comparison with the algorithm of Ref.~\cite{bravyi2016improved}.
Both simulations were performed on a Linux PC with a 3.2$\mathrm{GHz}$ Intel~i5-6500 CPU.
}
\label{table:1}
\end{table}

\begin{table}[h]
\centerline{
	\begin{tabular}{c|c|c|c}
	Number of CCZ Gates & 2 & 4 & 6 \\
	\hline
	Number of states $\chi_{\Delta}$ & 39 & 149 & 497 \\
	\hline
	New Runtime $\left(\mathrm{s}\right)$ & 0.30 & 1.02 & 3.82 \\
	\hline
	BG16 Runtime $\left(\mathrm{s}\right)$ & 5.22 & 27.94 & 100.11 \\
	\end{tabular}
}
\caption{Average runtime of the Norm Estimation step in seconds, for the new implementation compared with that of Ref.~\cite{bravyi2016improved}. Norm Estimation is used to compute single qubit marginals on a $40$-qubit state, with precision $\Delta=0.3$. Both simulations were single-threaded, and run on a Linux PC with a 3.2$\mathrm{GHz}$ Intel i5-6500 CPU.}
\label{table:HStimings}
\end{table}

\section{Stabilizer rank}
In this Section, we describe bounds on the exact and approximate stabilizer rank. In subsection \ref{Sec_Symmetric_states}, we give the proof of Theorem~\ref{Thm_many_copies}, which proceeds by establishing an upper bound on the exact stabilizer rank of states symmetric under permutations of certain subsystems. As a consequence we will see that $\chi(\psi^{\otimes t}) \ll \chi(\psi)^t$ for modest $t$. In subsection \ref{Sec_Lemma_Proof} we prove Theorem~\ref{thm:randomCvec} using a Sparsification lemma that allows us to convert exact stabilizer decompositions into approximate stabilizer decompositions (with possibly fewer terms). In subsection \ref{Sec_Clifford_Magic_States}, we study the approximate stabilizer rank of Clifford magic states and establish Proposition \ref{thm:clifmagic}. Finally, in subsection \ref{Sec_ultra} we turn our attention to lower bounds and prove Proposition \ref{prop:lower_bound}.

\subsection{Exact stabilizer rank}
\label{Sec_Symmetric_states}
 Let us denote $\mathrm{Sym}_{n,t}$ as the subspace that is symmetric with respect to swaps between $t$ partitions with each partition holding $n$ qubits.  For instance, any $n$-qubit state $\psi$ satisfies $\psi^{\otimes t} \in \mathrm{Sym}_{n,t}$ for any $t$. Although the symmetric subspace also contains states entangled across these partitions.  Throughout this section we use $\mathrm{dim}( \ldots)$ to denote the dimension of a vector space and  $\mathrm{span}( \ldots)$ to denote the vector space spanned by a set of vectors.  Let us agree that when we write $\mathrm{dim}( \mathbb{S} )$ where $\mathbb{S}$ is a set of vectors (rather than a vector space) this means the dimension of the vector space spanned by $\mathbb{S}$.

This section provides a proof of Thm.~\ref{Thm_many_copies}, though we shall actually prove a more general result regarding the stabilizer rank of a subspace defined as follows
\begin{dfn}
	We define stabilizer rank $\chi(P)$ of a subspace $P$ to be the minimum $\chi$ such that there exists a set of $\chi$ stabilizer states $\mathbb{S}=\{ \phi_1,  \phi_2, \ldots , \phi_\chi \}$ satisfying $P \subset \mathrm{span}[ \mathbb{S} ]$. 
\end{dfn}  
Notice that given a set of stabilizer states $\mathbb{S}$ such that $\mathrm{Sym}_{n,t} \subseteq \mathrm{span}(\mathbb{S}) $, it follows that every element of the space $\mathrm{Sym}_{n,t}$ can be decomposed in terms of $|\mathbb{S}|$ stabilizer states.  Therefore, if $\Psi \in \mathrm{Sym}_{n,t}$ then $\chi(\Psi) \leq \chi(\mathrm{Sym}_{n,t})$.   As a special case, if $\Psi = \psi^{\otimes t}$ then  $\chi(\psi^{\otimes t}) \leq \chi(\mathrm{Sym}_{n,t})$.  Therefore,  Thm.~\ref{Thm_many_copies} follows as a corollary of the following result
\begin{lemma}
	\label{symmetric}
	Consider $\mathrm{Sym}_{n,t}$ for some nonzero $n$ and $t$.  It follows that for all $t \leq 5$ we have
	\begin{equation}
		\label{symmetric_inequality}
	\chi( \mathrm{Sym}_{n,t} ) = \mathrm{dim}[\mathrm{Sym}_{n,t}] = \binom{2^n + t -1}{t} ,
	\end{equation}
	where the round brackets denotes the binomial coefficient.
\end{lemma}
This has the direct and elegant consequence that for all single qubit states $\psi$ we have $\chi( \psi ^{\otimes t}  ) \leq t+1$ whenever $t \leq 5 $. 

\begin{proof}[Proof of Lemma \ref{symmetric}]  First we show that Eq.~\eqref{symmetric_inequality} holds for some $n$ and $t$ whenever there exists a set of stabilizer states $\mathbb{S}$ with the following properties:
	\begin{enumerate}
		 \item every $\Phi \in \mathbb{S}$ satisfies $\Phi \in \mathrm{Sym}_{n,t}$; and
\item $\mathrm{dim}(\mathrm{Sym}_{n,t}) = \mathrm{dim}( \mathbb{S} )$.
	\end{enumerate}	
 For any set of vectors $\mathbb{S}$, there exists a subset $\mathbb{S}' \subseteq \mathbb{S}$  that is a minimal spanning set, with $\mathrm{span}(\mathbb{S}')=\mathrm{span}(\mathbb{S})$ and $|\mathbb{S}'|=\mathrm{dim}(\mathbb{S})$.  Therefore, given a set that spans the symmetric space we can conclude that $\chi(\mathrm{Sym}_{n,t}  ) \leq \mathrm{dim}(\mathbb{S})$.  Furthermore, if $\mathbb{S}$ has the swap invariance property then $\mathrm{span}(\mathbb{S}) \subseteq \mathrm{Sym}_{n,t} $ and $\mathrm{dim}(\mathbb{S}) \leq \mathrm{dim}(\mathrm{Sym}_{n,t})$.   Combining these inequalities gives $\chi(\mathrm{Sym}_{n,t}  ) \leq\mathrm{dim}(\mathrm{Sym}_{n,t})$.  It is obvious that $\mathrm{dim}(\mathrm{Sym}_{n,t}) \leq \chi(\mathrm{Sym}_{n,t}  )$ and so $\chi(\mathrm{Sym}_{n,t}  ) = \mathrm{dim}(\mathrm{Sym}_{n,t})$.  Lastly, the dimension of the symmetric space is well-known and can for example be found in Ref.~\cite{zhu16}. 

Next, it remains to find a set $\mathbb{S}$ with the aforementioned properties for certain values of $n$ and $t$.  We consider sets of stabilizer states of the form $\mathbb{S}_{n,t}=\{ \ket{\phi_j}^{\otimes t} \}_j$ where $\{ \ket{\phi_j} \}_j =: \mathrm{STAB}_n$ is the set of all $n$-qubit stabilizer states.  This ensures property 1.  It remains to show when $\mathbb{S}_{n,t}$ has sufficiently large dimension (property 2).  We observe that the operator
\begin{equation}
	 \sigma_{n,t} :=\frac{1}{|\mathrm{STAB}_n|} \sum_{\psi_j \in \mathrm{STAB}_n} \kb{\psi_j}{\psi_j}^{\otimes t}
\end{equation}	
satisfies
\begin{equation}
	 \mathrm{rank}( \sigma_{n,t} ) = \mathrm{dim}( \mathbb{S}_{n,t}  ) .
\end{equation}
and so property 2 also holds whenever $\mathrm{rank}( \sigma_{n,t} )=\mathrm{dim}( \mathrm{Sym}_{n,t}  )$. 

Let us consider when $t \leq 3$ with no constraints on $n$.  We will use that the stabilizer states form a projective $3$-design~\cite{Webb16,kueng15,zhu16}.  The relevant property of such designs is that for $t \leq 3$ we know
\begin{equation}
	\sigma_{n,t} \propto \Pi_{n,t} ,
\end{equation}
where $\Pi_{n,t}$ is the projector onto $\mathrm{Sym}_{n,t}$.  Therefore, $\mathrm{rank}( \sigma_{n,t} )=\mathrm{rank}( \Pi_{n,t} )  =\mathrm{dim}( \mathrm{Sym}_{n,t}  )$ and the lemma is proven for  the case of $t \leq 3$.

For $t=4$, it is known that the stabilizer states are not a projective $4$-design and so $\sigma_{n,4}$ is not proportional to the symmetric projector~\cite{zhu16}.  However, the stabilizer states ``fail gracefully" to be a projective $4$-design~\cite{zhu16}, such that the deviation of $\sigma_{n,4}$ from $\Pi_{n,4}$ is sufficiently small that we still have $\mathrm{rank}(\sigma_{n,4})=\mathrm{rank}(\Pi_{n,4})$. Ref.~\cite{gross2017schur} extends this result such that we can also deduce the following
\begin{claim}
	\label{RankEq}
	For all $n$ and $t \leq 5$ we have  $\mathrm{rank}(\sigma_{n,t})=\mathrm{rank}(\Pi_{n,t})$.
\end{claim}	
This suffices to prove Lem.~\ref{symmetric}.  In contrast, this proof technique can not extend to $t>5$ due to the stabilizer testing algorithm of Ref~\cite{gross2017schur}. This algorithm shows that there exists a projector $W$ such that $\mathrm{Tr}[W\sigma_{n,6} ]=0$ but  $\mathrm{Tr}[W\Pi_{n,6}] \neq 0$, which entails $\mathrm{rank}(\sigma_{n,6})<\mathrm{rank}(\Pi_{n,6})$.    

Although Claim~\ref{RankEq} can be deduced from Ref.~\cite{gross2017schur}, it is not explicitly shown, so we provide the details here.  Examples 4.27 and 4.28 of Ref.~\cite{gross2017schur},  show that
\begin{align}
	\sigma_{n,4} & \propto   \Pi_{n,4}  + a_{n}\Pi_{n,4}  P_{[4]}^{\otimes n} \Pi_{n,4} , \\ \nonumber 
	\sigma_{n,5} & \propto   \Pi_{n,5}  + b_{n}\Pi_{n,5}  P_{[5]}^{\otimes n} \Pi_{n,5} .
\end{align}
where $a_{n}$ and $b_n$ are positive constants and $P_{[4]}$ and $P_{[5]}$ are projectors onto a stabilizer code 
\begin{align}
	P_{[4]} & = \frac{1}{4}( \id^{\otimes 4} + X^{\otimes 4} + Y^{\otimes 4} + Z^{\otimes 4})  \\ \nonumber
	P_{[5]} & = P_{[4]} \otimes \id 
\end{align}
Since $P_{[4]}$ and $P_{[5]}$ are positive operators, so too are $a_n\Pi_{n,4} P_{[4]}^{\otimes n} \Pi_{n,4}$ and $b_n\Pi_{n,5} P_{[5]}^{\otimes n} \Pi_{n,5}$.  In general, if $M$ and $N$ are positive operators we have $\mathrm{rank}(M+N)\geq \mathrm{rank}(M)$.  Therefore, for $t=4, 5$ we have $\mathrm{rank}(	\sigma_{n,t})\geq \mathrm{rank}( \Pi_{n,t})$, which implies the desired rank equivalence and completes the proof. \end{proof}

\begin{figure}
	\centering
	\includegraphics[width=250pt]{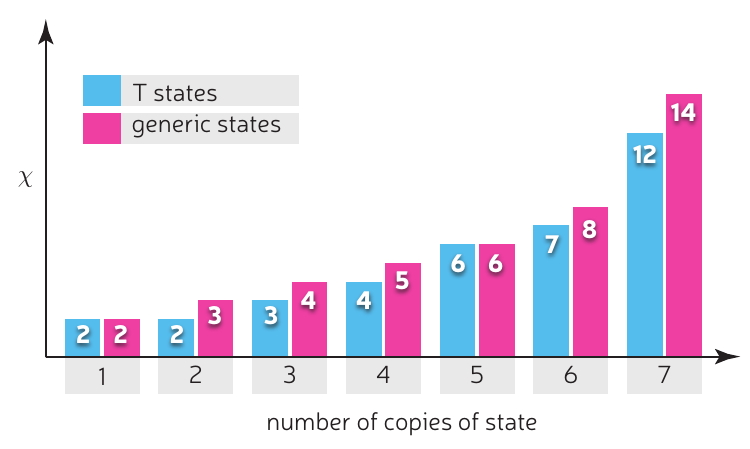}
	\caption{The exact stabilizer rank (numerically found) for $n$ copies of a single qubit state: for the $T$ state and for generic single qubit states.}
	\label{SingleQubitStates}
\end{figure}

\begin{table}

\centering

	\begin{tabular}{|c||ccccc|} 	  \hline
		&	$t=1$ &  $t=2$ &  $t=3$ & $t=4$ & $t=5$ \\  \hline  \hline
		$n=1$  &	2 &  1.73205 &  1.5874 & 1.49535 & 1.43097 \\ 
		$n=2$  &  4 &  3.16228 & 2.71442 & 2.4323 & 2.23685 \\
		$n=3$ &   8 &  6           & 4.93242 & 4.26215&  3.79966 \\   \hline
	\end{tabular}	
	\caption{Upper bounds on $\chi(\psi^{\otimes t})^{1/t}$ where $\psi$ is an $n$ qubit state.  Asymptotically we have $\chi(\psi^{\otimes N}) \leq (\chi(\psi^{\otimes t})^{1/t})^N$.  Since lower values lead to lower simulation overhead we see a significant advantage in using blocks of size up to 5.} 	\label{Tab_Numbers}
\end{table}

We reflect that we have proved Lem.~\ref{symmetric}, from which Thm.~\ref{Thm_many_copies} follows immediately.  For single qubit states ($n=1$) this entails that 
\begin{equation}
		\chi( \psi^{\otimes t}) \leq t + 1 , \forall t \leq 5 \label{EqOneQubit} .
\end{equation}	 
The rest of this subsection discusses numerical experiments into whether this inequality is tight.

Clearly the bound is loose for stabilizer states since then we have $\chi( \psi^{\otimes t})=1 < t+1$.  However, Clifford magic states are also exceptional for many $t$ values.  Bravyi, Smith and Smolin~\cite{Bravyi16stabRank} discuss the stabilizer rank of single qubit states that are an eigenstate of some Clifford unitary.  For instance, the $\ket{T}$  Clifford magic states are exceptional in that for $ 2 \leq t \leq 4$ we have that $\chi(T^{\otimes t})=t < t+1$, which we illustrate in Fig.~\ref{SingleQubitStates}.  We remark that $\ket{T}$ has the Clifford symmetry $C_T \ket{T}=\ket{T}$ for $C_T = TXT^\dagger$. In total there are 12 single qubit states in the Clifford orbit of $\ket{T}$. An additional class of Clifford symmetric states is the Clifford orbit of the face state $\ket{f}$ 
\begin{equation}
	\kb{f}{f} = \frac{1}{2} \left( \id + \frac{X + Y + Z}{\sqrt{3}} \right) ,
\end{equation}	
which comprises 8 different states.  The face state is an eigenstate of the Clifford $C_F=e^{-i \pi /12}SH$  that cyclically permutes Pauli $X$,$Y$ and $Z$.   Bravyi, Smith and Smolin reported (see conjecture 1 of Ref.~\cite{Bravyi16stabRank}) that $\chi(f^{\otimes t})$ appears to equal $\chi(T^{\otimes t})$, providing another class of states where Eq.~\eqref{EqOneQubit} is not tight.

Next, we ask if there are any other single qubit states for which Eq.~\eqref{EqOneQubit} is not tight.  We proceed by a heuristic, numerical search, extending the search method of Ref.~\cite{Bravyi16stabRank}. To find a decomposition of a state $\ket{\psi}$, we use an objective function $F_{\Psi}\left(\{\ket{\phi_{j}}\}\right)=||\Pi\ket{\Psi}||$ where $\Pi$ is a projector onto $\mathrm{span}\left(\{\ket{\phi_{j}}\}\right)$.   We start by choosing a set of $k$ random stabilizer states $\{\ket{\phi_{j}}\}$, with $k=2$ on the first run.  Random stabilizer states were obtained by generating a random binary matrix, using the algorithm of Garcia et al. to convert it to a canonical stabilizer tableau, and computing the corresponding state vector~\cite{garcia2012efficient}.  Let the value of the objective function at a given timestep be $F$. We update one stabilizer state in the set by applying a random Pauli projector, and evaluate the objective function on the new set $F_{\Psi}\left(\{\ket{\phi_{j}}\}'\right)=F'$. If $F'>F$ then we accept the move, otherwise the new decomposition is accepted with a probability $p=\text{exp}\left[-\beta\left(F-F'\right)\right]$, where $\beta$ is an inverse temperature parameter that decreases as the walk proceeds~\cite{Bravyi16stabRank}.  If $F$ equals 1 at any point in the walk, we halt and conclude $\chi(\Psi) \leq k$.  If $F$ does not converge to unity within a constant number of steps, we increment $k$ and start again.

Random typical states were generated as $\ket{\psi}=U\ket{0}$, where $U$ are Haar random unitaries.  We sampled 1000 Harr random states and numerically estimated the stabilizer rank of $\Psi=\psi^{\otimes t}$ using the above method. In every instance, the best decomposition we found saturated the inequalities of Eq.~\eqref{EqOneQubit}. We also examined conjecture 1 of~\cite{Bravyi16stabRank}, by searching for decompositions of single-qubit Clifford magic states. All decompositions found were below the bound of Eq.~\eqref{EqOneQubit}.

Although these numerical searches were not exhaustive, the results support the hypothesis that Eq.~\eqref{EqOneQubit} is an equality for typical single qubit states.  This supports the conjecture that Eq.~\eqref{EqOneQubit} is tight, if and only if the state has no Clifford symmetries. 

As a closing remark, we comment on consequences of these results for simulation overheads.  If a circuit contains many copies of the same multi-qubit phase gate, simulation overheads are reduced by working with blocks of magic states as shown in Table.~\ref{Tab_Numbers}.

\label{Sec_approx_stab_rank}

\subsection{Sparsification Lemma}
\label{Sec_Lemma_Proof}

Our new bounds  on the approximate stabilizer rank in Theorem~\ref{thm:randomCvec}
are obtained using the following   lemma.
It shows how to convert a stabilizer decomposition of some target state $\psi$ 
with a small $l_1$ norm to 
a sparse stabilizer decomposition of $\psi$.
\begin{lemma}[\bf Sparsification]
	\label{lem:randomCvec}
Let $\psi$ be a normalized $n$-qubit state with a decomposition $\ket{\psi} = \sum_j c_j \ket{\phi_j} $ where all $\phi_j $ are normalized stabilizer states and $c_j \in \mathbb{C}$.  For any integer $k$ there exists a distribution of random quantum states $|\Omega\rangle$ of the form 
$\ket{\Omega}=\frac{\|c\|_1}{k} \sum_{\alpha=1}^k \ket{\omega_\alpha}$ where each $\ket{\omega_\alpha}$ is (up to a global phase) one of the states $\{\ket{\phi_j} \}$ and
\begin{equation}
 \mathbb{E}\left( \, \| \psi -\Omega \|^2\right)  = \frac{\|c \|^2_1}{k},
\label{eq:meanomega}
\end{equation}
where $\|\vec{c} \|_1 := \sum_j |c_j|$ and $\| \psi \| = \sqrt{ \bk{\psi}{\psi} }$.
\end{lemma}	

Theorem~\ref{thm:randomCvec} is a simple corollary of Lemma~\ref{lem:randomCvec}.
Indeed, assume that all $\phi_j$ are stabilizer states. Choosing $k=( \| c \|_1 / \delta )^2 $ we find that the right-hand side is upper-bounded by $\delta^2$.  Therefore there exists at least one $|\Omega\rangle$ (which is manifestly a sum of $k$ stabilizer states) that $\delta$-approximates $|\psi\rangle$.  This proves Theorem~\ref{thm:randomCvec}. 

Note that we can use Markov's inequality and Eq.~\eqref{eq:meanomega} to lower bound the probability that a randomly chosen $\Omega$ is a good approximation to $\psi$, e.g., 

\[
\mathrm{Pr}\left[\| \psi -\Omega \|^2\geq 2\delta^2\right]\leq 1/2  \quad \text{for } \quad k\geq \frac{\|c \|^2_1}{\delta^2}.
\]
Suppose that we randomly choose some $|\Omega\rangle$ as prescribed above. Can we estimate how well it approximates $\psi$? The following Lemma can be used for this purpose.
\begin{lemma}[\bf Sparsification tail bound]
Let $\psi,\Omega,k$ be as in Lemma \ref{lem:randomCvec}.  If we choose $k\geq \frac{\|c \|^2_1}{\delta^2}$ then 
\begin{equation}
\mathbb{E}\left[\langle \Omega|\Omega\rangle-1\right] \leq \delta^2 ,
\label{eq:expomega}
\end{equation}
and
\begin{equation}
\mathrm{Pr}\left[\|\psi-\Omega\|^2\leq \langle \Omega|\Omega\rangle-1+\delta^2\right] \geq 1-2\exp{\left(-\frac{\delta^2}{8F(\psi)}\right)}.
\label{eq:event}
\end{equation}
\label{lem:tailbound}
\end{lemma}

Note that we are interested in cases where the stabilizer fidelity $F(\psi)$ is exponentially small as a function of the number of qubits $n$. In such cases the Lemma states that
\[
\|\psi-\Omega\|^2\leq \langle \Omega|\Omega\rangle-1+\delta^2 ,
\]
with all but vanishingly small probability if $n$ is sufficiently large. Moreover, the quantity $\langle \Omega|\Omega\rangle$ appearing in the above can be approximated to a given relative error using the norm estimation algorithm from Section \ref{Sec_fast_norm} which has runtime scaling linearly with $k$.

\begin{proof}[Proof of Lemma~\ref{lem:randomCvec}]
Define a probability distribution $p_j := |c_j| / || c ||_1$ and write
\begin{equation}
\ket{\psi} = \|c\|_1\sum_j p_j 	\ket{W_j} 
\end{equation}	
where $\ket{W_j}:= (c_j/|c_j|) 	\ket{\phi_j} $ are normalized stabilizer states. Now define a random variable $|\omega\rangle$ which is equal to $|W_j\rangle$ with probability $p_j$.   Then
\begin{equation}
|\psi\rangle=\|c\|_1\mathbb{E}\left[|\omega\rangle\right].
\end{equation}	
Let $k$ be a positive integer and consider a random state 
\begin{equation}
|\Omega\rangle=\frac{\|c\|_1}{k}\sum_{\alpha=1}^{k} |\omega_\alpha\rangle ,
\label{eq:Omega}
\end{equation}
where $\omega_1,\omega_2,\ldots,\omega_k$ are i.i.d random copies of $|\omega\rangle$.  By construction,  on average we have
\begin{equation}
 \mathbb{E} [ \bk{\psi}{\Omega} ] =  \mathbb{E} [ \bk{\Omega}{\psi} ] = 1
\label{eq:overlap}
\end{equation}
even though for any particular random sample $ \bk{\Omega}{\psi} \neq 1$.  In general, not only will $\Omega$ not be proportional to $\psi$, but $\Omega$  will not be correctly normalized.  However, the normalization can be bounded in expectation as follows
\begin{align}
\mathbb{E}\left[\langle \Omega |\Omega\rangle\right]&=\frac{\|c\|_1^2}{k^2}\mathbb{E}\left[\sum_{\alpha=1}^{k} \langle \omega_\alpha|\omega_\alpha\rangle\right]+\frac{\|c\|_1^2}{k^2}\mathbb{E}\left[\sum_{\alpha\neq \beta} \langle \omega_\alpha|\omega_\beta\rangle\right]\\
&=\|c\|_1^2\frac{\mathbb{E}\left[\langle \omega|\omega\rangle\right]}{k}+\frac{1}{k^2} k(k-1)\\
&\leq 1+\frac{\|c\|_1^2}{k}
\label{eq:omom}
\end{align}
where in the second line we used the fact that $\|c\|_1^2\mathbb{E}\left[\langle \omega_\alpha|\omega_\beta\rangle\right]=\langle \psi|\psi\rangle$ for $\alpha\neq \beta$. 

We are interested in the expected error 
\begin{align}
 \mathbb{E}\left[ \|\psi\rangle -|\Omega\rangle\|^2\right] & = \mathbb{E}\left[ \bk{\Omega}{\Omega} \right]- \mathbb{E}\left[ \bk{\Omega}{\psi} \right]-\mathbb{E}\left[ \bk{\psi}{\Omega} \right]+\mathbb{E}\left[ \bk{\psi}{\psi} \right]
\label{eq:experror}
\end{align} 
Using  $\bk{\psi}{\psi}=1$, Eq.~\eqref{eq:omom} and Eq.~\eqref{eq:overlap} we find 
\begin{align}
  \mathbb{E}\left[ \|\psi\rangle -|\Omega\rangle\|^2\right] & =  \frac{\|c\|_1^2}{k}.
\end{align}
This completes the proof of Lemma ~\ref{lem:randomCvec}.
\end{proof}
\begin{proof}[Proof of Lemma \ref{lem:tailbound}]
Equation \eqref{eq:expomega} follows directly from Eq.~\eqref{eq:omom} and the choice of $k$. Define random variables
\[
X_\alpha=\|c\|_1 \mathrm{Re}(\langle \psi|\omega_\alpha\rangle) \qquad 1\leq\alpha \leq k
\]
and let 
\[
\bar{X}=\frac{1}{k}\sum_{\alpha=1}^{k} X_\alpha=\mathrm{Re}(\langle \psi|\Omega\rangle).
\]
Then
\begin{equation}
\left|\mathrm{Re}(\langle \psi|\Omega\rangle)-1\right|=\left|\bar{X}-\mathbb{E}[\bar{X}]\right|.
\label{eq:repart}
\end{equation}
Now $\bar{X}$ is a sample mean of $k$ independent and identically distributed random variables $X_\alpha$,  each of which is bounded as
\begin{equation}
|X_\alpha|\leq \|c\|_1 |\langle \psi|\omega_\alpha\rangle| \leq \|c\|_1 \sqrt{F(\psi)}
\label{eq:xkbound}
\end{equation}
where in the last inequality we used the definition of stabilizer fidelity. Applying Hoeffding's inequality \cite{hoeffding1963probability} and using Eqs.~(\ref{eq:repart}, \ref{eq:xkbound}) gives
\begin{equation}
\mathrm{Pr}\left[\left|\mathrm{Re}(\langle \psi|\Omega\rangle)-1\right|\geq \frac{\delta^2}{2}\right]\leq 2\exp\left(-\frac{2k\delta^4}{4\left(2\|c\|_1\sqrt{F(\psi)}\right)^2}\right) \leq 2\exp\left(-\frac{\delta^2}{8F(\psi)}\right)
\label{eq:hoeff}
\end{equation}
where we used $k\geq \|c\|_1^2/\delta^2$. Finally, applying the triangle inequality to Eq.~\eqref{eq:experror} gives
\begin{equation}
\|\psi-\Omega\|^2 \leq \langle \Omega|\Omega\rangle-1+2\left|1-\mathrm{Re}(\langle \psi|\Omega\rangle)\right| 
\label{eq:normdiff}
\end{equation}
Combining Eqs.~(\ref{eq:normdiff}, \ref{eq:hoeff}) completes the proof.

\end{proof}

\subsection{Approximate stabilizer rank of Clifford magic states}
\label{Sec_Clifford_Magic_States}
Proposition \ref{thm:clifmagic} asserts that $\xi(\psi)=F(\psi)^{-1}$ when $\psi$ is a Clifford magic state (Def \ref{Dfn_CMS}). In fact, this relation holds for a wider class of $\psi$ and we comment on this at the end of the following proof. Recall that a Clifford magic state $\psi$ is stabilized by a group of Clifford unitaries with generators $Q_j:=V X_j V^\dagger$.  We denote this group as $\mathcal{Q}:=\langle Q_j \rangle = \langle V X_j V^\dagger \rangle$. Here we describe upper bounds on the approximate stabilizer rank of Clifford magic states. We begin with the proof of Proposition \ref{thm:clifmagic}

\begin{proof}[Proof of Proposition \ref{thm:clifmagic}]
From the definition of Clifford magic states, we have
\begin{align}
P_\psi = |\psi\rangle\langle\psi| & = V \frac{1}{2^n} \prod_j (\id + X_j) V^\dagger \\ \nonumber
& = \frac{1}{2^n} \prod_j (\id + Q_j ) \\ \nonumber
& = \frac{1}{| \mathcal{Q} |} \sum_{q \in \mathcal{Q}} q 
\end{align}
Let $\phi_0$ be a stabilizer state such that $|\langle \psi|\phi_0\rangle|^2>0$. Then
\begin{align}
|\psi\rangle & = \frac{\ket{\psi}\bk{\psi}{\phi_0}}{\bk{\psi}{\phi_0}}  \label{eq:cmag1}\\ \nonumber
& = \left[ \frac{1}{|\mathcal{Q}|} \sum_{q\in \mathcal{Q}} q \right]  \frac{\ket{\phi_0}}{\bk{\psi}{\phi_0}} \\ \nonumber
& = \frac{1}{|\mathcal{Q}|\langle \psi|\phi_0\rangle}\sum_{q\in \mathcal{Q}}q|\phi_0\rangle.
\end{align}
Using Eq.~(\ref{eq:cmag1}) and the fact that $q|\phi_0\rangle$ is a stabilizer state for all $q\in \mathcal{Q}$ we immediately obtain
\[
|| \vec{c} ||_1^2 = \frac{1}{|\langle \psi|\phi_0\rangle|^2},
\]
for this decomposition.  To minimise $|| \vec{c} ||_1^2$ it is natural to use the stabilizer state with the larger possible overlap, $F(\psi)= \mathrm{max}_{\phi_0}|\langle \psi|\phi_0\rangle|^2$, which we call the stabilizer fidelity.  Therefore, once we have found a $\phi_0$ attaining the maximum, we have a decomposition achieving $|| \vec{c} ||_1^2 =  F(\psi)^{-1}$.   This discussion suffices to prove that 
\[
\xi(\psi) \leq F(\psi)^{-1}.
\]
To establish the converse consider any stabilizer decomposition 
\[
|\psi\rangle=\sum_{j=1}^{\chi}c_j |\phi_j\rangle.
\]
Taking the inner product with $\psi$ we get
\[
1=\left|\sum_{j=1}^{\chi}c_j \langle \psi |\phi_j\rangle \right|\leq \|c\|_1 \sqrt{F(\psi)},
\]
where we used the fact that $|\langle \psi|\phi_j\rangle|^2\leq F(\psi)$. Squaring the above completes the proof.

More generally, let $\mathcal{Q}$ be \emph{any} subgroup of the Clifford group satisfying $\kb{\psi}{\psi}=|\mathcal{Q}|^{-1}\sum_{q\in\mathcal{Q}} q$ and with exactly one group element (the identity) stabilizing $\ket{\phi_0}$. The above proof goes through unmodified, but admits a wider class of states for which $\xi(\psi)=F(\psi)^{-1}$ including the face state, $\ket{f}$, satisfying
\begin{align}
\kb{f}{f}=\frac{1}{2}\left(\id+\frac{(X+Y+Z)}{\sqrt{3}}\right)=\frac{1}{|\mathcal{Q}|}\sum_{q\in\mathcal{Q}} q
\end{align}
where $\mathcal{Q}=\{\id,C_F,C_F^2\}$ and $C_F=e^{-i \pi /12}SH$ is the Clifford that cyclically permutes Pauli $X$,$Y$ and $Z$. 
\end{proof}
The $\ket{T}^{\otimes n}$ state is the most well known example of a Clifford magic state.  It has been shown (see Lemma~2 of Ref.~\cite{Campbell11} or Lemma~2 of Ref.~\cite{bravyi2016improved}) that  $F(T^{\otimes n})^{-1}=|\bk{+}{T}|^{2n}$ and so $\ket{+}^{\otimes n}$ can be used to generate the decomposition with optimal $\xi(\psi)$.  Combining this with Lemma~\ref{lem:randomCvec} gives the same upper bound on $\chi_\delta(T^{\otimes n})$ as was previously shown in Ref.~\cite{bravyi2016improved}.   However, the techniques are slightly different.  Our Lemma~\ref{lem:randomCvec} randomly selects a subset of terms from the decomposition, whereas Ref.~\cite{bravyi2016improved} randomly select a subset of terms that form a random linear code.  We remark that the random linear code construction also generalises to all Clifford magic states. For any linear code $\mathcal{L} \subseteq \mathbb{F}_2^n$ we can associate a subgroup $\mathcal{Q}_\mathcal{L} \subseteq \mathcal{Q}$.  That is, given a decomposition as in Eq.~\eqref{eq:cmag1} with group $\mathcal{Q}$, we can choose a random subgroup $\mathcal{Q_L} \subseteq \mathcal{Q}$ and define the normalised approximate state 
\begin{equation}
	|\mathcal{L} \rangle \propto \sum_{q\in \mathcal{Q_L}}q|\phi_0\rangle .
\label{eq:cmag}
\end{equation}

Following analogous steps to those in Ref.~\cite{bravyi2016improved}, one can show that this approach gives the same asymptotic scaling of $\chi_\delta$ as in Lemma~\ref{lem:randomCvec}.  While the behaviour of $\chi_\delta$  is identical, it may be easier to implement a simulator working with random subgroups than random subsets. 

As a further example, let us consider the Clifford magic state corresponding to a CCZ (control-control-Z) gate,
\begin{align}
		\ket{CCZ} = CCZ \ket{+}\ket{+}\ket{+} = \frac{1}{\sqrt{8}} \sum_{a,b,c \in \{0,1\}} (-1)^{abc} \ket{a}\ket{b}\ket{c}
\end{align}
This magic state is the ``+1" eigenstate for a group $\mathcal{Q}$ with three generators of the form $CCZ \cdot X_j \cdot CCZ^\dagger$.  More explicitly these generators are
\begin{align}
	Q_1   & =  CCZ \cdot X_1 \cdot CCZ^\dagger  =  X_1  CZ_{2,3} \\ \nonumber
	Q_2  & =  CCZ \cdot X_2 \cdot CCZ^\dagger  =  X_2  CZ_{1,3} \\ \nonumber
	Q_3  & =  CCZ \cdot X_3 \cdot CCZ^\dagger =  X_3  CZ_{1,2} 
\end{align}
where $CZ_{i,j}$ denotes a control-Z between qubits $i$ and $j$.  One can straightforwardly confirm that $F(CCZ)= |\bk{+++}{CCZ}|^2=9/16$, and that
\begin{align}
	\label{CCZ_decomp}
	 \ket{CCZ} = \frac{1}{6}\sum_{Q \in \mathcal{Q}} Q\ket{+++} ,
\end{align}	
has $|| \vec{c} ||_1^2 = 16/9$.  Using this decomposition for many CCZ states shows $ \chi_\delta( CCZ^{\otimes t} ) \leq \delta^{-2} (9/16)^t \sim \delta^{-2} 1.778^t$.  Note that this is slower exponential scaling than obtained by synthesizing each CCZ with 4 $T$-gates and using $ \chi_\delta( T^{\otimes 4t} ) \leq \delta^{-2} 1.884^t$. It is conceivable that a better decomposition exists since $\xi$ only provides an upper bound on the approximate stabilizer rank. 

One could obtain better decompositions if the stabilizer fidelity is not multiplicative, but we show later (see Corollary~\ref{Cor_single_qubits}) that $F(T^{\otimes t})=F(T)^t$ and $F(CCZ^{\otimes t})=F(CCZ)^t$. However, one of the significant open questions remaining from this work is whether stabilizer fidelity is always multiplicative for all Clifford magic states.  Lastly, we remark that one can lift the above stabilizer decomposition to obtain a Clifford unitary decomposition of CCZ that can be used for an approximate sum-over-Cliffords simulator.

\subsection{Lower bound based on ultra-metric matrices}
\label{Sec_ultra}
Previous sections give explicit stabilizer decompositions of states and therefore upper bounds on the stabilizer rank.   Yet we have no techniques that provide lower bounds on the stabilizer rank that scale exponentially with the number of copies.  Here we present results in this direction. Let $|H\rangle=\cos{(\pi/8)}|0\rangle + \sin{(\pi/8)}|1\rangle$ be the magic state which is Clifford equivalent to $\ket{T}$. We would like to approximate $n$ copies of $|H\rangle$ by a low-rank linear combination
of stabilizer states 
\[
|\tilde{x}\rangle =|\tilde{x}_1\rangle \otimes \cdots \otimes |\tilde{x}_n\rangle
\quad \mbox{where} \quad
|\tilde{0}\rangle=|0\rangle \quad \mbox{and} \quad
|\tilde{1}\rangle =|+\rangle.
\]
Here we derive a lower bound on the rank of such approximations stated earlier as Prop.~\ref{prop:lower_bound}.  We first restate this result as follows
\begin{theorem}
\label{thm:main}
Suppose $S\subseteq \{0,1\}^n$ is an arbitrary subset
and $\phi$ is an arbitrary   linear combination of states
$|\tilde{x}\rangle$ with $x\in S$ such that $\|\phi\|=1$.  Then 
\begin{equation}
|S|\ge |\langle H^{\otimes n}|\phi\rangle |^2 \cdot \cos{(\pi/8)}^{-2n}.
\end{equation}
\end{theorem}
\begin{proof}
Let $\chi=|S|$ and $S=\{x^1,x^2,\ldots,x^\chi\}$ for some bit strings $x^i$.
The orthogonal projector onto a linear subspace spanned by
the states $|\tilde{x}^1\rangle,\ldots,|\tilde{x}^\chi\rangle$ has the form
\begin{equation}
\label{eq2}
\Pi=\sum_{i,j=1}^\chi (G^{-1})_{i,j} |\tilde{x}^i\rangle\langle \tilde{x}^j|,
\end{equation}
where $G$ is the Gram matrix defined by
$G_{i,j}=\langle \tilde{x}^i | \tilde{x}^j\rangle = t^{|x^i \oplus x^j|}$,
with $t=2^{-1/2}$. Here and below $\oplus$ denotes addition of bit strings
modulo two.
Noting that $\langle\tilde{x}|H^{\otimes n}\rangle = \cos{(\pi/8)}^{n}$ for all $x$ one gets
\begin{equation}
\label{eq3}
|\langle H^{\otimes n}|\phi\rangle |^2 \le \langle H^{\otimes n} |\Pi|H^{\otimes n}\rangle
=\cos{(\pi/8)}^{2n} \sum_{i,j=1}^\chi (G^{-1})_{i,j}
\le \chi \cos{(\pi/8)}^{2n}.
\end{equation}
The last inequality follows from 
\begin{lemma}
Suppose $x^1,\ldots,x^\chi \in \{0,1\}^n$ are distinct bit strings
and $0<t < 1$ is a real number. 
Let $G$ be a matrix of size $\chi$ with entries
\begin{equation}
\label{Gt}
G_{i,j}=t^{|x^i \oplus x^j|}.
\end{equation}
Then $G$ is invertible and 
\begin{equation}
\label{ubound}
\sum_{i,j=1}^\chi (G^{-1})_{i,j} \le \chi.
\end{equation}
\end{lemma}
\begin{proof}
Let $|1\rangle,|2\rangle,\ldots,|\chi\rangle$ be the basis vectors of $\RR^\chi$
such that $G_{i,j}= \langle i|G|j\rangle$.
We claim that Eq.~(\ref{ubound}) holds whenever one can 
find a family of matrices $G_\sigma$ and probabilities $p_\sigma\ge 0$ such that
\begin{enumerate}
\item[(a)] $G=\sum_\sigma p_\sigma G_\sigma$ and $\sum_\sigma p_\sigma=1$
\item[(b)] $G_\sigma$ is positive definite 
\item[(c)] $0\le \langle i|G_\sigma|j\rangle \le 1$
and $\langle i|G_\sigma|i\rangle=1$ 
\item[(d)]  $\langle i|G_\sigma^{-1}|j\rangle \le 0$ for $i\ne j$ 
\end{enumerate}
Indeed, let $|e\rangle$ be the  all-ones vector, $|e\rangle=\sum_{i=1}^\chi |i\rangle$.
We have to prove that $\langle e|G^{-1}|e\rangle\le \chi$.
Conditions~(a,b) imply that  $G$ is positive definite (and thus invertible). 
Noting that  the function $f(x)=x^{-1}$ is operator convex on the interval $(0,\infty)$
one gets 
\begin{equation}
\label{upper1}
\langle e|G^{-1}|e\rangle \le \sum_\sigma p_\sigma \langle e|G_\sigma^{-1} |e\rangle.
\end{equation}
From conditions~(c,d) one gets
\[
\langle i|G^{-1}_\sigma|j\rangle \le \langle i|G^{-1}_\sigma|j\rangle \langle j|G_\sigma|i\rangle
\]
for $i\ne j$ with the equality for $i=j$. 
Therefore 
\begin{equation}
\label{upper2}
\langle e|G^{-1}_\sigma|e\rangle = \sum_{i,j=1}^\chi \langle i| G^{-1}_\sigma|j\rangle 
\le \sum_{i,j=1}^\chi \langle i| G^{-1}_\sigma|j\rangle  \langle j| G_\sigma|i\rangle  = \mbox{Tr}(G^{-1}_\sigma G_\sigma) =
\mbox{Tr}(I)= \chi.
\end{equation}
Substituting this into Eq.~(\ref{upper1}) gives $\langle e|G^{-1}|e\rangle\le \chi \sum_\sigma p_\sigma =\chi$,
as desired.

It remains to construct the requisite matrices $G_\sigma$.
Our construction is based on the so-called  {\em ultrametric matrices},
see Refs.~\cite{MMM,NabenVarga}.
\begin{dfn}
\label{dfn:UM}
A symmetric real matrix $A$ is called  ultrametric iff
$0\le A_{i,j}<1$ for $i\ne j$,  $A_{i,i}=1$, and 
\begin{equation}
\label{UM1}
A_{i,j} \ge \min{(A_{i,k}, A_{j,k})}
\quad \mbox{for all $i,j,k$}.
\end{equation}
\end{dfn}
The last condition demands that for any triple of elements $A_{i,j}$, $A_{i,k}$, $A_{j,k}$  the two smallest 
elements coincide. 
The following fact was established in Refs.~\cite{MMM,NabenVarga}.
\begin{fact}
\label{fact:UM}
Suppose $A$ is an ultrametric matrix. Then $A$ is invertible and 
positive definite. Furthermore,  $\langle i| A^{-1}|j\rangle \le 0$ for all $i\ne j$.
\end{fact}
Thus it suffices to show that $G$ is a probabilistic mixture of ultrametric matrices.
Indeed, if condition~(a) holds for some ultrametric matrices $G_\sigma$ then 
condition~(c) follows directly from Definition~\ref{dfn:UM}
while conditions~(b,d) follow from Fact~\ref{fact:UM}.

The first step is to equip the Boolean cube $\{0,1\}^n$
with a distance function that obeys an analogue of the ultrametricity condition
Eq.~(\ref{UM1}).
Given a pair of bit strings $x,y\in \{0,1\}^n$,
define $d(x,y)$ as the smallest integer $j\ge 0$ such that 
the last $n-j$ bits of $x$ and $y$ coincide (that is, $x_i=y_i$ for all $i>j$).
We set $d(x,y)=n$ if $x_n\ne y_n$.
Note that $d(x,y)$ is different from the Hamming distance.
For example, $d(101,111)=2$ and $d(101,100)=3$.
By definition $d(x,y)\in [0,n]$ and $d(x,y)=0$ iff $x=y$. 
Furthermore, $d(x,y)$ depends only on $x\oplus y$.
We claim that
\begin{equation}
\label{UM2}
d(x,y)\le \max{\{ d(x,z),d(z,y)\}}
\end{equation}
for any triple of strings $x,y,z$. Indeed,
let $j=\max{\{ d(x,z),d(z,y)\}}$. Then
$x_i=z_i=y_i$ for all $i>j$, that is, 
$d(x,y)\le j$.

Suppose $q_w$ is a normalized probability distribution on the set of integers
$w=0,1,\ldots,n$ such that $q_w>0$ for all $w$. 
Define a $\chi\times \chi$ matrix $A$ such that 
\begin{equation}
\label{Aij}
A_{i,j}= \sum_{w\ge d(x^i,x^j)} \;  q_w.
\end{equation}
Here $x^i$ and $x^j$ are the bit strings from the statement of the lemma.
We claim that $A$ is ultrametric (according to Definition~\ref{dfn:UM}).
Indeed, consider any triple $i,j,k$ as in Eq.~(\ref{UM1}) and assume wlog that
$A_{i,k}\le A_{j,k}$. 
Since the matrix element $A_{i,j}$ is a monotone decreasing function
of the distance $d(x^i,x^j)$, we get $d(x^i,x^k)\ge d(x^j,x^k)$.
Then Eq.~(\ref{UM2}) gives $d(x^i,x^j)\le d(x^i,x^k)$.
Using the monotonicity again one gets $A_{i,j}\ge A_{i,k}=\min{\{A_{i,k},A_{j,k}\}}$,
confirming Eq.~(\ref{UM1}). 
The remaining conditions $0\le A_{i,j}<1$ for $i\ne j$ and $A_{i,i}=1$ follow from the 
assumption that all bit strings $x^i$ are distinct and that $q_w$ is a 
normalized probability distribution.
Thus the matrix $A$ defined by Eq.~(\ref{Aij}) is indeed ultrametric.

We are now ready to define a family of ultrametric matrices $G_\sigma$
such that $G=\sum_\sigma p_\sigma G_\sigma$.
Let us choose the label $\sigma$ as a permutation of $n$ integers, $\sigma \in S_n$.
The distribution $p_\sigma$ will be the uniform distribution on the symmetric group, that is,
$p_\sigma=1/n!$ for all $\sigma\in S_n$.
Given a permutation $\sigma$ and a bit string $x\in \{0,1\}^n$ let $\sigma(x)\in \{0,1\}^n$
be the result of permuting bits of $x$ according to $\sigma$.
We set 
\begin{equation}
\label{Gsigma1}
\langle i|G_\sigma |j\rangle = \sum_{w\ge d(\sigma(x^i),\sigma(x^j))} \;  q_w.
\end{equation}
The same argument as above confirms that $G_\sigma$ is ultrametric for any permutation $\sigma$.
Define 
\begin{equation}
\label{Gsigma2}
G'=\frac1{n!} \sum_{\sigma \in S_n} G_\sigma.
\end{equation}
We claim that  $\langle i|G'|j\rangle = \langle i|G|j\rangle = t^{|x^i\oplus x^j|}$
for a suitable choice of probabilities $q_w$.
Indeed, the identity $d(x,y)=d(0^n,x\oplus y)$ implies that 
a matrix element $\langle i |G_\sigma|j\rangle$ depends only on $x^i\oplus x^j$.
By the symmetry, matrix elements $\langle i| G'|j\rangle$ depend only on the Hamming
weight $h=|x^i\oplus x^j|$. Therefore it suffices to compute
$\langle i |G'|j\rangle$ for the special case when $x^i=0^n$ is the all-zero string
and $x^j$ is any fixed bit string with the Hamming weight $h$, for example,
$x^j=1^h 0^{n-h}$. Then
\begin{equation}
\label{Gsigma3}
\langle i|G'|j\rangle=
\frac1{n!} \sum_{\sigma \in S_n} \; \;  \sum_{w\ge d(0^n,\sigma(1^h0^{n-h}))} \;  q_w.
\end{equation}
By definition of the distance $d(x,y)$ one gets $d(0^n,\sigma(1^h0^{n-h}))\le w$ iff 
$h \le w$ and $\sigma_1,\ldots,\sigma_h\le w$.  The number of such permutations $\sigma$ is ${w \choose h} h! (n-h)!$.
Exchanging the sums over $\sigma$ and $w$ in Eq.~(\ref{Gsigma3})  one gets
\begin{equation}
\label{Gsigma4}
\langle i|G'|j\rangle=
\frac1{n!} \sum_{w=h}^n {w \choose h} h! (n-h)! \,q_w.
\end{equation}
We shall choose $q_w$ as a binomial distribution,
\begin{equation}
\label{qw}
q_w= {n \choose w} t^w (1-t)^{n-w}.
\end{equation}
Substituting Eq.~(\ref{qw}) into Eq.~(\ref{Gsigma4}) 
and introducing a variable $p=w-h$ one gets
\begin{equation}
\label{Gsigma5}
\langle i|G'|j\rangle=
\sum_{p=0}^{n-h}  {n-h \choose p} t^{p+h} (1-t)^{n-h-p}  = t^h.
\end{equation}
By definition, $h=|x^i\oplus x^j|$, so that  $G'=G$ as claimed.
Thus $G$ is indeed a probabilisitic mixture of ultrametric matrices
and the lemma is proved.
\end{proof}
\end{proof}

\section{Stabilizer fidelity and Stabilizer extent}
In the previous Section we established upper bounds on the approximate stabilizer rank of a state $\psi$ which depend on the the squared $1$-norm $\|c\|_1^2$, where
\[
|\psi\rangle=\sum_{j}c_j |\phi_j\rangle,
\]
is a given stabilizer decomposition.   Recall that the stabilizer extent $\xi(\psi)$  denotes the minimum value of $|| c ||_1^2$ over all stabilizer decompositions of $\psi$. We find that $\xi$ is easier to work with than the approximate stabilizer rank. For any fixed $n$-qubit state $\psi$, $\xi(\psi)$ can be computed using a simple convex optimization program, although the size of this computation scales poorly with $n$. In this section we develop tools that allow us to efficiently compute $\xi(\psi)$ whenever $\psi$ is a tensor product of $1, 2$ and $3$ qubit states. In particular, we prove Proposition~\ref{multi} which establishes that $\xi$ is multiplicative for tensor products of $1$, $2$, and $3$-qubit states. 

In subsection~\ref{Sec_convex_dual} we use standard convex duality to give a characterization of $\xi$ in terms of the \textit{stabilizer fidelity}, defined as the maximum overlap with respect to the set of stabilizer states
\begin{equation}
F( \psi ) := \mathrm{max}_{\phi \in \mathrm{STAB}_n} |\bk{\psi}{\phi}|^2.
\end{equation}

As a consequence, multiplicativity of $\xi$ is directly related to multiplicativity of the stabilizer fidelity. In subsection~\ref{Sec_Fid_Multi} we give sufficient and necessary conditions for multiplicativity of the stabilizer fidelity. In particular, we define the class of \textit{stabilizer-aligned} states for which multiplicativity holds.  In subsection~\ref{Sec_when_stab_aligned} we investigate the class of stabilizer-aligned states and prove that all tensor products of $1,2$ and $3$ qubit states are stabilizer-aligned.  Finally, in section~\ref{Sec_Cstar_multi} we use these results to prove Proposition~\ref{multi}. 

\subsection{Convex duality}
\label{Sec_convex_dual}

Here we show that the optimization of $\xi(\psi)$ can be recast as a dual convex problem and we prove the following:
\begin{theorem}
	\label{thm:witness}
	For any $n$-qubit state $\psi$ we have
	\begin{equation}
	\xi(\psi)=\max_{\omega} \frac{|\bk{\psi}{\omega}|^2}{F(\omega)},
	\label{eq:witness}
	\end{equation}
	where the maximum is over all $n$-qubit states $\omega$. 
\end{theorem}
 Thus any $n$-qubit state $\omega$ can act as a witness to provide a lower bound on $\xi$ and, furthermore, there exists at least one optimal witness state $\omega_{\star}$ which achieves the maximum in Eq.~\eqref{eq:witness}. For example, choosing $\omega=\psi$, we get the lower bound 
\begin{equation}
\xi(\psi) \geq \frac{1}{F(\psi)}.
\label{eq:mirrorwitness}
\end{equation}
For Clifford magic states this lower bound is tight as stated in Proposition~\ref{thm:clifmagic}.  We remark that Thm.~\ref{thm:witness} is a special case of results found in the literature on general resource theories~\cite{regula2017convex}.

\begin{proof}
We shall map the problem into the language of convex optimization and use standard results in that field~\cite{boyd2004convex}.   Using the computation basis $\{ \ket{ \vec{x} } \}$ we can decompose any stabilizer state $\ket{\psi_j} =  \sum_{\vec{x}} M_{\vec{x}, j } \ket{\vec{x}} $. Given a state $\ket{\psi} =  \sum_{\vec{x}}  a_{\vec{x}} \ket{\vec{x}} $, the primal optimization problem can be written as
\begin{align}
	\sqrt{\xi( \psi )}  &  = \mathrm{min}_{\vec{c}}  f(\vec{c})=|| \vec{c} ||_1 \\
	\mbox{such that } & M\vec{c} - \vec{a}	 = 0
\end{align}
This is clearly a convex optimization problem with affine constraints.  Because the coefficient in $\vec{c}$ are complex, rather than real, this is a second order cone problem~\cite{boyd2004convex}.  For any convex optimization problem there exists a dual function
\begin{align}
	g(\nu) &  = \mathrm{inf_{\vec{c}}}  \left( || \vec{c} ||_1 + \nu^{T}(M\vec{c} - \vec{a}) \right) \\
	&= \begin{cases}  -  \nu^{T}\vec{a} & \mbox{ when } ||M^T\nu ||_{\infty} \leq 1 \\
		-  \infty & \mbox{ otherwise }
	\end{cases}
\end{align}
where for any value of the dual variables $\nu$ we have $g(\nu)  \leq \sqrt{\xi( \psi )} $.  The dual optimization problem is the maximisation of $g(\nu)$ over $\nu$ to obtain the best lower bound on  $\sqrt{\xi( \psi )}$.   We can discount the need for two cases by adding  $||M^T \nu  ||_{\infty} \leq 1$ as a constraint, to obtain the problem
\begin{align}
	d^{\star}( \psi )  &  = \mathrm{max}_{\nu}  -  \nu \cdot \vec{a} \\ \nonumber
	\mbox{such that } & || M^T\nu ||_{\infty} \leq 1  ,
\end{align}
or more simply
\begin{align}
	d^{\star}( \psi )  &  = \mathrm{max}_{\nu} \frac{ -  \nu \cdot \vec{a}}{|| M^T\nu ||_{\infty} } .
\end{align}
Because the primal problem has affine constraints, we have strong duality and there must exist a  $\nu_\star$ such that $g(\nu_{\star}) = -  \nu_{\star}^{T}\vec{a}  =  \sqrt{\xi( \psi )} $.   Next, we restate this dual problem in terms of quantum states.   For every $\nu$ we can associate a normalised quantum state
\begin{equation}
\ket{\omega_\nu} : = \frac{1}{||\nu ||_2 } \sum_{\vec{x}} (-\nu^*_{\vec{x}}) \ket{\vec{x}} ,
\end{equation}
so that 
\begin{equation}
\bk{\omega_\nu}{\psi} =  \frac{-  \nu \cdot \vec{a}}{ ||\nu ||_2 } .
\end{equation}
Next we note that 
\begin{equation}
||M^T\nu ||_{\infty} = \frac{\mathrm{Max}_{\ket{\phi} \in \mathrm{STAB}} |	\bk{\omega_\nu}{\phi}   | }{ || \nu ||_2 } = \frac{\sqrt{F(\omega_\nu)}}{|| \nu ||_2}
\end{equation}
Therefore, the dual problem can also be stated as 
\begin{align}
	d^{\star}( \psi )  &  =  \mathrm{max}_{\ket{\omega_{\nu}}}	\frac{\bk{\omega_\nu}{\psi}}{\sqrt{F(\omega_\nu)}}  ,
\end{align}
where the factors $|| \nu ||_2$ have cancelled out.  The optimal $\nu_\star$ gives the optimal  $\ket{\omega_\star}$, which completes the proof.
\end{proof}

\subsection{Stabilizer alignment}
\label{Sec_Fid_Multi}
Combining Theorems \ref{thm:clifmagic} and \ref{thm:randomCvec} we get an upper bound
$\chi_{\delta}(\psi)\leq \delta^{-2} F(\psi)^{-1}$ on the approximate stabilizer rank of any Clifford magic state $\psi$.
We shall be  interested in the case when $\psi$ is a tensor product of 
a large number of few-qubit magic states such as $T$-type or CCZ-type states. For example, the case $\psi=CCZ^{\otimes m}$ is relevant to gadget-based simulation of 
quantum circuits composed of Clifford gates and $m$ CCZ gates. 
This motivates the question of whether the stabilizer fidelity $F(\psi)$ is multiplicative
under tensor product, i.e. \begin{equation}
F(\psi\otimes \phi)\stackrel{?}{=}F(\psi)F(\phi).
\label{eq:fmult}
\end{equation}
Note that $F(\psi\otimes \phi)\ge F(\psi)F(\phi)$ 
since the set of stabilizer states is closed under tensor product. 

Below we define a set of quantum states $\mathcal{S}$ such that 
$F(\phi \otimes \psi) =F(\phi)F(\psi)$ whenever $\phi,\psi\in \mathcal{S}$. Remarkably, this set is also closed under tensor product, that is $\phi\otimes \psi \in \mathcal{S}$ whenever $\phi,\psi\in \mathcal{S}$.
Moreover, we show that the stabilizer fidelity is not multiplicative for all states
$\phi \notin \calS$. More precisely, for any $\phi \notin \calS$ there exists a state $\psi$
such that $F(\phi\otimes \psi)>F(\phi) F(\psi)$. In that sense, our results provide
necessary and sufficient conditions under which the stabilizer fidelity is multiplicative
under tensor product.

To state our results let us generalize the definition of stabilizer fidelity as follows. For each $n\geq 1$ and $0\leq m\leq n$ define a set $S_{n,m}$ which consists of all stabilizer projectors $\Pi$ on $n$ qubits satisfying $\mathrm{Tr}[\Pi]=2^m$. 
\begin{dfn}
For any $n$-qubit state $|\phi\rangle$ define
\[
F_m(\phi)=2^{-m/2} \max_{\Pi\in S_{n,m}} \langle \phi|\Pi|\phi\rangle. \qquad \qquad m=0,\ldots, n.
\]
Let us say that
$\phi$ is  \textit{stabilizer-aligned} if $F_m(\phi)\leq F_0(\phi)$ for all $m$. 
\end{dfn}
Note that in the above $F_0=F$ is the stabilizer fidelity.  Here we investigate the consequences of stabilizer-alignment. Whether or not a given state is stabilizer-aligned is discussed in the following subsection.
\begin{theorem}
Suppose $\phi$ and $\psi$ are stabilizer-aligned. Then $\phi\otimes \psi$ is stabilizer-aligned and
\[
F(\phi\otimes \psi)=F(\phi)F(\psi).
\]
Conversely, suppose $\phi$ is not stabilizer-aligned. 
Let $\phi^{\star}$ be the complex conjugate of $\phi$.
Then 
\[
F(\phi\otimes \phi^{\star})>F(\phi)F(\phi^{\star}).
\]
\label{thm:stabaligned}
\end{theorem}
The theorem implies that the stabilizer fidelity is multiplicative for any
stabilizer-aligned states: 
\begin{corol}
	\label{Fid_Multi}	
Suppose $\psi_1,\ldots,\psi_L$ are stabilizer-aligned quantum states. Then
\[
F(\psi_1\otimes \psi_2\otimes \ldots \otimes \psi_L)=\prod_{j=1}^{L} F(\psi_j).
\]
\label{cor:mul}
\end{corol}
We prove Theorem~\ref{thm:stabaligned} using  characterization of entanglement in tripartite stabilizer states from Ref.~\cite{ghz}:
\begin{lemma}[\cite{ghz}]
\label{lemma:ghz}
	Any pure tripartite  stabilizer state can be transformed by
	local unitary Clifford  operators  to a tensor product of states from the set $\{|0\rangle,|\Psi^{+}\rangle,|\Psi^{+}_3\rangle\}$ where
	\[
	|\Psi^{+}\rangle=\frac{1}{\sqrt{2}}(|00\rangle+|11\rangle) \qquad \qquad |\Psi^{+}_3\rangle=\frac{1}{\sqrt{2}}\left(|000\rangle+|111\rangle\right).
	\]
\end{lemma}
\begin{corol}[\cite{ghz}]
Suppose $\Pi$ be a stabilizer projector describing a bipartite system $AB$. Then there exists a unitary Clifford operator $U=U_A\otimes U_B$ and
integers $a,b,c,d\ge 0$ such that 
\be
\label{PROJ1}
U\Pi U^{-1}
=\sum_{\alpha=1}^{2^a} \sum_{\beta=1}^{2^b} \sum_{\gamma=1}^{2^c}
|\omega_{\alpha\beta\gamma}\ra\la \omega_{\alpha\beta\gamma}|,
\ee
where
\be
\label{PROJ2}
|\omega_{\alpha\beta\gamma}\ra = 2^{-d/2} \sum_{\delta=1}^{2^d}  |\alpha,\gamma,\delta\ra \otimes
|\beta,\gamma,\delta\ra.
\ee
Here $|\alpha,\gamma,\delta\ra$
and $|\beta,\gamma,\delta\ra$ are the computational basis vectors of $A$ and $B$.
\label{corol:ghz}
\end{corol}
\begin{proof}
Let us apply Lemma~\ref{lemma:ghz} to a tripartite stabilizer state
\[
|\Psi\ra= (\Pi \otimes I)2^{-n/2}\sum_{z\in \{0,1\}^n }\; |z\rangle_{AB} \otimes |z\rangle_C ,
\]
where $n=|A|+|B|$ and $C$ is a system of $n$ qubits.
The lemma implies that 
$\Pi$ is equivalent modulo local Clifford operators to a tensor
product of local stabilizer projectors  $|0\ra\la 0|$ and $I=|0\ra\la 0|+|1\ra\la 1|$
as well as bipartite projectors
 $|00\ra\la 00|+|11\ra\la 11|$ 
and $|\Psi^+\ra\la \Psi^+|$ shared between $A$ and $B$.
Let $a$ and $b$ be the number of times $\Pi$ contains
the identity factor on $A$ and $B$ respectively. Let $c$ be the number of times
$\Pi$ contains the projector $|00\ra\la 00|+|11\ra\la 11|$ shared between $A$ and $B$.
Let $d$ be the number of times $\Pi$ contains the EPR projector $|\Psi^+\ra\la \Psi^+|$.
The desired family of states $\omega_{\alpha\beta\gamma}$ is then obtained
by writing each projector $I$ and $|00\ra\la 00|+|11\ra\la 11|$
as a sum of rank-$1$ projectors onto the computational basis vectors. 
\end{proof}
\begin{proof}[Proof of Theorem~\ref{thm:stabaligned}]
To prove the first two claims of the theorem 
it suffices to show that 
	\begin{equation}
	F_m(\phi\otimes \psi)\leq F_0(\phi)F_0(\psi).
	\label{eq:k00}
	\end{equation}
for all $m$.
Indeed, combining Eq.~\eqref{eq:k00} and the obvious bound
 $F_0(\phi)F_0(\psi) \le F_0(\phi\otimes \psi)$ shows that 
$F_m(\phi\otimes \psi)\le F_0(\phi\otimes \psi)$, that is,
$\phi\otimes \psi$ is stabilizer-aligned.
Using Eq.~\eqref{eq:k00} for $m=0$ gives 
multiplicativity of the stabilizer fidelity $F_0(\phi\otimes \psi)=F_0(\phi)F_0(\psi)$. 

Define a bipartite system $AB$ such that $\phi$ and $\psi$ are states of $A$ and $B$.
Let  $\Pi$  be a stabilizer projector of rank $2^m$ such that 
	\[
	F_m(\phi\otimes \psi)=2^{-m/2}\langle \phi \otimes \psi |\Pi|\phi\otimes \psi \rangle.
	\]
We shall write $\Pi$ as a sum of rank-$1$ stabilizer projectors as stated in 
Corollary~\ref{corol:ghz}.
Since local Clifford unitary operators do not change the stabilizer fidelity,
we shall absorb the unitaries $U_A$ and $U_B$
into the states $\phi$ and $\psi$ respectively.
Accordingly, below we set $U=I$.
Consider a single term $\omega_{\alpha\beta\gamma}$ in the decomposition of $\Pi$.
Applying the Cauchy-Schwarz inequality one gets
\be
\label{multi_eq1}
|\la \phi\otimes \psi|\omega_{\alpha\beta\gamma}\ra|^2 = 2^{-d} 
\left| \sum_{\delta=1}^{2^d} \la \phi|\alpha,\gamma,\delta\ra
\cdot \la \psi|\beta,\gamma,\delta\ra \right|^2
\le 2^{-d} \la \phi|\Pi^A_{\alpha\gamma}|\phi\ra
\cdot \la \psi|\Pi^B_{\beta \gamma}|\psi\ra ,
\ee
where we defined stabilizer projectors 
\be
\label{PROJ_AB}
\Pi^A_{\alpha,\gamma} = 
\sum_{\delta=1}^{2^d} |\alpha ,\gamma,\delta\ra\la \alpha ,\gamma,\delta|
\quad \mbox{and} \quad 
 \Pi^B_{\beta,\gamma}=\sum_{\delta=1}^{2^d}
 |\beta ,\gamma,\delta\ra\la \beta ,\gamma,\delta|.
\ee
By assumption, $\psi$ is stabilizer-aligned. Thus 
\be
\label{multi_eq1'}
\max_\gamma 
\la \psi| \sum_{\beta=1}^{2^b} \Pi^B_{\beta \gamma}|\psi\ra \le 2^{(b+d)/2} F_0(\psi).
\ee
Here we noted that $\sum_{\beta=1}^{2^b} \Pi^B_{\beta \gamma}$ is 
a projector of rank $2^{b+d}$ for all $\gamma$.
Combining Eq.~(\ref{multi_eq1},\ref{multi_eq1'}) gives
\be
\label{multi_eq2}
\la \phi\otimes \psi|\Pi|\phi\otimes \psi\ra
=
\sum_{\alpha=1}^{2^a} \sum_{\beta=1}^{2^b} \sum_{\gamma=1}^{2^c}
|\la \phi\otimes \psi|\omega_{\alpha\beta\gamma}\ra|^2
\le 2^{(b-d)/2} F_0(\psi) \cdot  \la \phi| \sum_{\alpha=1}^{2^a} \sum_{\gamma=1}^{2^c}
 \Pi^A_{\alpha,\gamma}|\phi\ra
\ee
The assumption that $\phi$ is stabilizer-aligned gives
\be
\label{multi_eq3}
 \la \phi| \sum_{\alpha=1}^{2^a} \sum_{\gamma=1}^{2^c}
\Pi^A_{\alpha,\gamma}|\phi\ra
\le 
2^{(a+c+d)/2} F_0(\phi).
\ee
Here we noted that $\sum_{\alpha=1}^{2^a} \sum_{\gamma=1}^{2^c}
 \Pi^A_{\alpha,\gamma}$ is a projector of rank $2^{a+c+d}$.
Combining Eqs.~(\ref{multi_eq2},\ref{multi_eq3}) gives
\[
\la \phi\otimes \psi|\Pi|\phi\otimes \psi\ra\le 2^{(a+b+c)/2} F_0(\psi)F_0(\phi).
\]
It remains to notice that $\Pi$ has rank
$2^m$, where $m=a+b+c$.
This establishes Eq.~\eqref{eq:k00}.

We now prove the converse statement from Theorem \ref{thm:stabaligned}.
\begin{lemma}
	Let $\phi$ be an $n$-qubit state which is not stabilizer-aligned. Then
	\[
	F_0(\phi\otimes \phi^{\star})>F_0(\phi)F_0(\phi^{\star}).
	\]
\end{lemma}
\begin{proof}
	If $\phi$ is not stabilizer-aligned then we have $F_m(\phi)>F_0(\phi)$ for some $m\in \{1,\ldots,n\}$. Let $\Pi$ be a stabilizer projector with
	\[
	F_m(\phi)=\frac{1}{\sqrt{2}^m} \langle \phi|\Pi|\phi\rangle.
	\]
	Let $C$ be an $n$-qubit Clifford such that 
	\[
	\Pi=C\left( |0\rangle\langle0|_{n-m}\otimes I_m\right)C^{\dagger}.
	\]
	Next consider a system of $2n$ qubits and partition them as $[2n]=ABA'B'$ where $|A|=|A'|=n-m$ and $|B|=|B'|=m$. Define a $2n$-qubit stabilizer state
	\[
	|\theta\rangle=C\otimes \alpha |0\rangle_A |\Phi\rangle_{BB'} |0\rangle_{A'} ,
	\] 
	where 
	\[
	|\Phi\rangle_{BB'}=\frac{1}{\sqrt{2}^m}\sum_{z\in \{0,1\}^m}|z\rangle_B|z\rangle_{B'}.
	\]
	Also define a normalized $m$-qubit state
	\[
	|\omega\rangle=\frac{1}{2^{m/4}\sqrt{F_m(\phi)}}\left(\langle 0|_{n-m}\otimes I_m\right) C|\phi\rangle.
	\]
	\begin{align}
		F_0(\phi\otimes \phi^{\star})&\geq \langle \phi\otimes \phi^{\star} |\theta\rangle\langle \theta| \phi\otimes \phi^{\star}\rangle\\
		&=\langle \omega\otimes \omega^\star|\Phi\rangle \langle \Phi|\omega\otimes \omega^\star\rangle 2^m (F_m(\phi))^2\\
		&=(F_m(\phi))^2\\
		&> F_0(\phi) F_0(\phi^{\star}).
	\end{align}
	where in the last line we used the fact that $F_m(\phi)>F_0(\phi)=F_0(\phi^{\star})$.
\end{proof}
\end{proof}

\subsection{Proving and disproving stabilizer alignment}
\label{Sec_when_stab_aligned}

In this section we prove that all states of $n\le 3$ qubits
are stabilizer-aligned.  We also show that 
typical  $n$-qubit states are not stabilizer-aligned for sufficiently large $n$.
An important lemma is the following
\begin{lemma}
	\label{Lem_F_ordering}
For any quantum state $\psi$ we have
$F_m(\psi)\le F_0(\psi)$ for $m=1,2,3$.	 
\end{lemma}
\noindent
It follows immediately that 
\begin{corol}
	\label{Cor_single_qubits}
All states of $n\le 3$ qubits are stabilizer-aligned.  
\end{corol}
Indeed, if we consider $n$-qubit states, it suffices to check
that $F_m(\psi)\le F_0(\psi)$ for $m\le n$.
\begin{corol}
\label{Cor_Fid_Stab_align}	
	If  $F_0(\psi) \geq 1/4$ then $\psi$ is stabilizer-aligned.  
\end{corol}
Indeed, if $m\ge 4$ then  $F_m(\psi) \le 2^{-m/2}\le 1/4\le F_0(\psi)$.

Finally,  we show that Haar-random $n$-qubit states  are not stabilizer-aligned
for sufficiently large $n$.
\begin{claim}
	\label{Claim_Harr}
	Let $\psi$ be a Haar-random $n$-qubit state.  Then
	\[
	\mathrm{Pr}[F_0(\psi\otimes \psi^{\star})\neq F_0 (\psi)F_0(\psi^{\star})]\geq 1-o(1).
	\]
	and so for large enough $n$ a typical state $\psi$ is not stabilizer-aligned.
\end{claim}
Highly structured states on a large number of qubits may be stabilizer-aligned, and for instance it is an open question whether or not all Clifford magic states are stabilizer-aligned.

\begin{proof}[Proof of Lemma~\ref{Lem_F_ordering}]
First, we claim that 
\be
\label{FFbound}
F_{m-1}(\psi)\ge 2^{-1/2} \left( 1+ \left[\frac{2^m-1}{4^m-1}\right]^{1/2}\right)
\cdot F_m(\psi)
\ee
for all $m\ge 1$. Indeed, consider a fixed $m$ and 
a rank-$2^m$ stabilizer projector $\Pi \in S_{n,m}$ such that 
$F_m(\psi)=2^{-m/2}\la \psi|\Pi|\psi\ra$.
Using the standard stabilizer formalism one can show that 
\[
U\Pi U^{-1} = I^{\otimes m}\otimes |0\ra \la 0|^{\otimes (n-m)}\equiv \Pi'
\]
for some $n$-qubit unitary Clifford operator $U$.
Define a state $|\psi'\ra=U|\psi\ra$.
We have
\[
\Pi'  |\psi' \ra = \Gamma^{1/2}  |\omega \ra \otimes |0^{n-m}\ra
\]
for some $m$-qubit normalized state $|\omega\ra$ and 
$\Gamma=\la \psi'|\Pi'|\psi'\ra=\la \psi|\Pi|\psi\ra$.
Since  $\omega$ is normalized, 
\[
\sum_{P\ne I} \la \omega|P|\omega\ra^2 = 2^m-1,
\]
where the sum runs over all $4^m-1$ non-trivial Pauli operators on $m$ qubits.
Thus there exists an $m$-qubit Pauli
operator $P\ne I$ such that 
\be
\label{goodP}
\la \omega |P|\omega\ra \ge \left(\frac{2^m-1}{4^m-1}\right)^{1/2}.
\ee
Define a stabilizer projector 
\[
\Pi''=
\frac12(I+P)\otimes |0\ra \la 0|^{\otimes (n-m)}\in S_{n,m-1}.
\]
Recalling that $\Gamma=\la \psi|\Pi|\psi\ra=2^{m/2} F_m(\psi)$ we arrive at
\begin{align}
F_{m-1}(\psi)=
F_{m-1}(\psi') & \ge 2^{-(m-1)/2}  \la \psi'|\Pi''|\psi'\ra \\ \nonumber
& =2^{-(m-1)/2} \frac{\Gamma}2 (1+ \la \omega |P|\omega\ra )  \\ \nonumber
& = 2^{-1/2} (1+ \la \omega |P|\omega\ra )\cdot F_m(\psi).
\end{align}
Combining this identity and Eq.~(\ref{goodP}) proves Eq.~(\ref{FFbound}).
Applying Eq.~(\ref{FFbound}) inductively gives
\be
\label{F1bound}
F_0(\psi)\ge  2^{-1/2}(1+\sqrt{1/3}) \cdot F_1(\psi)\approx 1.115 \cdot F_1(\psi),
\ee
\be
\label{F2bound}
F_0(\psi)\ge  2^{-1/2}(1+\sqrt{1/3}) \cdot 2^{-1/2}(1+\sqrt{3/15}) \cdot F_2(\psi)  \approx 1.141 \cdot F_2(\psi),
\ee
\be
\label{F3bound}
F_0(\psi)\ge 2^{-1/2}(1+\sqrt{1/3}) \cdot 2^{-1/2}(1+\sqrt{3/15})
\cdot 2^{-1/2}(1+\sqrt{7/63})\cdot F_3(\psi) 
  \approx 1.076 \cdot F_3(\psi).
\ee
Thus  $F_0(\psi)\ge F_m(\psi)$ for $m=1,2,3$ proving the lemma.
\end{proof}

Next, we prove claim~\ref{Claim_Harr}. \begin{proof}
	Let $w$ be any $n$-qubit state. For Haar-random $\psi$ the probability density function $p(y)$ of $y=|\langle w|\psi\rangle|^2$ does not depend on $w$ and is equal to (equation (9) of Ref.~\cite{random}), 
	\[
	p(y)=(2^n-1)(1-y)^{2^n-2}.
	\]
	Integrating this we obtain the cumulative distribution function
	\[
	\mathrm{Pr}\left[ |\langle w|\psi\rangle|^2\geq x\right]=(1-x)^{2^n-1}\leq \exp(-x(2^{n}-1)).
	\]
	Since an $n$-qubit stabilizer state is specified by $O(n^2)$ bits the cardinality of the set $\mathrm{STAB}_n$ of $n$-qubit stabilizer states is $|\mathrm{STAB}_n|\leq 2^{O(n^2)}$. Choosing $x=n^3/2^n$ and applying a union bound we get
	\[
	\mathrm{Pr} \left[ \left(\max_{w\in \mathrm{STAB}_n} |\langle \psi|w \rangle|^2\right)\geq n^3/2^n\right]\leq e^{-\Omega(n^3)}.
	\]
	This says that with probability very close to 1 a random $\psi$ has $F_0(\psi)=F_0(\psi^{\star})\leq n^3/2^n$. Next suppose $\psi$ has this property. Then
	\[
	F_0(\psi\otimes \psi^{\star})\geq \left|\frac{1}{\sqrt{2}^n}\sum_{z\in \{0,1\}^n} \langle z|\psi\rangle \langle z| \psi^{\star}\rangle \right|^2=\frac{1}{2^n},
	\]
	which is strictly greater than $F_0(\psi)F_0(\psi^{\star})\leq 2^{-2n}(n^3)^2$.
\end{proof}

\begin{figure}[t]
	\centering
	\includegraphics[scale=0.7]{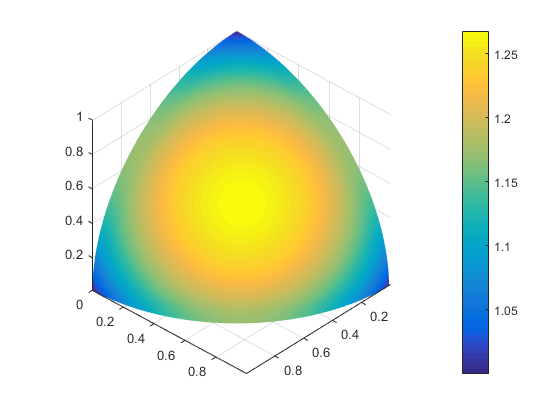}
	\caption{The color indicates the value of $\xi$ for single-qubit states in the first octant of the Bloch sphere. This function controls the upper bound on the approximate stabilizer rank as in Eq.~\eqref{eq:expsingleq}.}
	\label{fig:cstar1qubit}
\end{figure}

\subsection{Multiplicativity of stabilizer extent}
\label{Sec_Cstar_multi}
This subsection considers tensor products  of few-qubit states 
that involve at most three qubits each
and shows that $\xi$ behaves multiplicatively for such products, proving Proposition~\ref{thm:prod}. The proof will draw heavily on Theorem \ref{thm:witness} and Corollary~\ref{Cor_single_qubits}.
\begin{proof}[Proof of Proposition~\ref{thm:prod}]
By Theorem~\ref{thm:witness} there exist witness states $\{ \omega_{\star,1} , \omega_{\star,2}  , \ldots , \omega_{\star,L}  \}$ such that
	\begin{equation}
\frac{|\bk{\psi_j}{\omega_{\star, j}}|^2}{F(\omega_{\star, j})} =  \xi(\psi_j).
\end{equation}
We consider the product witness $\ket{\Omega } =  \bigotimes_j \ket{\omega_{\star,j}}$
for which
\begin{equation}
|\bk{\Psi}{\Omega}|^2 =  \prod_j |\bk{\psi_j}{\omega_{\star,j}}|^2.
\end{equation}
Furthermore,  using Corollary~\ref{Cor_single_qubits} and Theorem \ref{thm:stabaligned} we get
\begin{equation}
	F( \Omega ) = \prod_j  	F( \omega_{\star, j} ) .
\end{equation}
Putting this together yields
\begin{equation}  
\frac{|\bk{\Psi}{\Omega}|^2}{ F( \Omega ) } =  \prod_j  \frac{ |\bk{\psi_j}{\omega_{\star,j}}|^2 }{ F( \omega_{\star, j} ) } = \prod_{j=1}^L  \xi(\psi_j) .
\end{equation}
Thus, using $\Omega$ as a witness, we get
\begin{equation}  
\prod_{j=1}^L  \xi(\psi_j) \leq \xi(\Psi) .
\end{equation}
Furthermore, $\xi$ is inherently sub-multiplicative and so we must have equality. 
\end{proof}
Now let us see how this can be used to bound the approximate stabilizer rank of a product state $\alpha^{\otimes n}$ where $\alpha$ is a single-qubit state. Combining Theorem \ref{thm:prod} with Lemma \ref{lem:randomCvec} we get

\begin{equation}
\chi_{\delta}(\alpha^{\otimes n})\leq \delta^{-1} \xi(\alpha^{\otimes n})=\delta^{-2}(\xi(\alpha))^{n}.
\label{eq:expsingleq}
\end{equation}
Note that since $\alpha$ is a single-qubit state we can easily compute $\xi(\alpha)$ by solving a small convex optimization program.  In Figure \ref{fig:cstar1qubit} we plot $\xi(\alpha)$ as a function of the single-qubit state $\alpha$ on the first octant of the Bloch sphere.

\begin{figure}[t]
\centering
\includegraphics[scale=0.7]{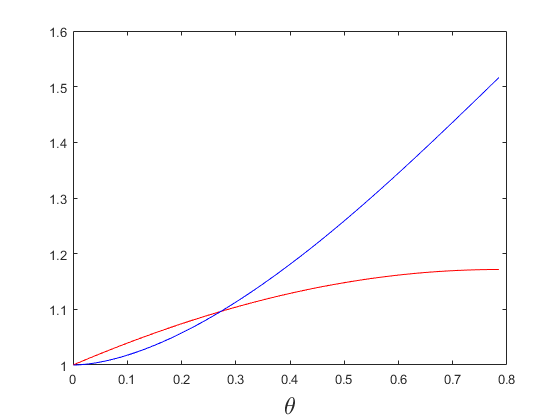}
\caption{The approximate stabilizer rank of $|\theta^{\otimes n}\rangle$ is upper bounded as $\chi_\delta(\theta^{\otimes n})\leq \delta^{-2}\xi(\theta)^{n}$, where $\xi(\theta)=(\cos(\theta/2)+\tan(\pi/8)\sin(\theta/2))^2$ is attained by the stabilizer decomposition from Eq.~\eqref{Eq_SingleQubitsDecomp}. The red line shows the function $\xi(\theta)$ for $\theta \in [0,\pi/4]$ and the blue line shows the function $g(\theta)=2^{h_2(\cos^2(\theta/2))}$ where $h_2$ is the binary entropy.  Our upper bound on the approximate stabilizer rank of $\theta^{\otimes n}$ performs better that obtained by a naive expansion in the $0,1$ basis whenever the red line lies below the blue line.}
\label{fig:redblue}
\end{figure}

The maximum value plotted in Figure \ref{fig:cstar1qubit} is $\xi(f)=2/(1+1/\sqrt{3})\approx 1.2679$, which is achieved by the so-called face state $|f\rangle$ which lies in the center of the surface and is defined by
\[
|f\rangle\langle f|=\frac{1}{2}\left(I+\frac{1}{\sqrt{3}}(X+Y+Z)\right).
\]

The single-qubit states in Figure \ref{fig:cstar1qubit} which lie in the $x$-$z$ plane are of the form
\begin{equation}
	\label{Eq_SingleQubitsDecomp}
|\theta\rangle=\cos(\theta/2)|0\rangle+\sin(\theta/2)|1\rangle= \left(\cos(\theta/2)-\sin(\theta/2)\right)|0\rangle+\sqrt{2}\sin(\theta/2)|+\rangle
\end{equation}
for $\theta\in [0,\pi/2]$. In this case, the stabilizer decomposition on the right hand side achieves the optimal value of $\xi$. We can use this example to show that in the general case the upper bound on approximate stabilizer rank given in Theorem \ref{thm:randomCvec} is not tight (for $\delta=O(1)$, say). When $\theta$ is close to $0$ it becomes advantageous to expand $\theta^{\otimes n}$ in the standard $0,1$ basis and truncate amplitudes which are very small. Using this approach one obtains an approximate stabilizer rank scaling as $2^{h_2(\cos^2(\theta/2))}$ where $h_2$ is the binary entropy. In Figure \ref{fig:redblue} we compare the performance of these upper bounds as a function of  $\theta$.

\section{Acknowledgements}

EC and MH are supported by the EPSRC (Grant No. EP/M024261/1). PC is supported by the EPSRC (Grant No. EP/L015242/1). The collaboration benefited from support by the NQIT project partnership fund (Grant No. EP/M013243/1), an EPSRC IIKE award and the IBM Research Frontiers Institute.

\bibliography{StabRank}

\end{document}